\documentclass[12pt]{article}

\makeatletter
\let\@fnsymbol\@arabic
\makeatother

\usepackage{amsmath,amsthm,amsfonts,amscd}
\usepackage{fullpage}
\usepackage{float}
\usepackage{url}
\usepackage{graphicx}
\def\mone{{\overline{1}}}
\DeclareMathOperator{\pow}{pow}
\DeclareMathOperator{\conj}{conj}
\DeclareMathOperator{\Th}{Th}
\DeclareMathOperator{\dr}{dr}
\DeclareMathOperator{\cmp}{cmp}
\DeclareMathOperator{\mr}{mr}
\DeclareMathOperator{\fab}{FAB}
\DeclareMathOperator{\zc}{ZC}
\def\Enn{{\mathbb{N}}}
\def\Zee{{\mathbb{Z}}}

\usepackage{pgf}
\usepackage{tikz}
\usetikzlibrary{arrows,automata,positioning}

\def \nodiv{{\, |\kern-4.5pt/}\, }
\def\divides{\ | \ }
\def\modd#1 #2{#1\ ({\rm mod}\ #2)}
\newcommand{\delete}{\setminus}
\newcommand{\coloneq}{\mathrel{\mathop:}=}
\providecommand{\floor}[1]{\left\lfloor#1\right\rfloor}
\newcommand{\dotdot}{\ldotp\ldotp}
\newcommand{\AND}{\wedge}
\newcommand{\OR}{\vee}
\newcommand{\IMPLY}{\rightarrow}
\providecommand{\ip}[1]{\left\langle#1\right\rangle}
\newcommand{\st}{\;:\;}

\makeatletter
\def\Ddots{\mathinner{\mkern1mu\raise\p@
\vbox{\kern7\p@\hbox{.}}\mkern2mu
\raise4\p@\hbox{.}\mkern2mu\raise7\p@\hbox{.}\mkern1mu}}
\makeatother

\author{Chen Fei Du\thanks{School of Computer Science,
University of Waterloo,
Waterloo,  ON  N2L 3G1,
Canada; \newline
{\tt cfdu@uwaterloo.ca}, {\tt sh2mousa@uwaterloo.ca},
        {\tt shallit@uwaterloo.ca} .}\;,
Hamoon Mousavi$^1$, Luke Schaeffer\thanks{Computer Science and Artificial Intelligence Laboratory,
The Stata Center, MIT Building 32, 32 Vassar Street, Cambridge, MA 02139 USA;
{\tt lrschaeffer@gmail.com} .}\;, and 
Jeffrey Shallit$^1$}

\title{Decision Algorithms for Fibonacci-Automatic Words, 
with Applications to Pattern Avoidance}

\begin{document}

\maketitle

\theoremstyle{plain}
\newtheorem{theorem}{Theorem}
\newtheorem{corollary}[theorem]{Corollary}
\newtheorem{almosttheorem}[theorem]{(Almost) Theorem}
\newtheorem{lemma}[theorem]{Lemma}
\newtheorem{proposition}[theorem]{Proposition}

\theoremstyle{definition}
\newtheorem{definition}[theorem]{Definition}
\newtheorem{example}[theorem]{Example}
\newtheorem{conjecture}[theorem]{Conjecture}
\newtheorem{openproblem}[theorem]{Open Problem}
\newtheorem{proc}[theorem]{Procedure}

\theoremstyle{remark}
\newtheorem{remark}[theorem]{Remark}

\begin{abstract}
We implement a decision procedure for answering questions about
a class of infinite words that might be called (for lack of a better
name) ``Fibonacci-automatic''.  This class includes, for example,
the famous Fibonacci word ${\bf f} = 01001010\cdots$,
the fixed point of the 
morphism $0 \rightarrow 01$ and $1 \rightarrow 0$.  
We then recover many results about the Fibonacci
word from the literature (and improve some of them),
such as assertions about the
occurrences in $\bf f$ of squares, cubes, palindromes, and so forth.
As an application of our method we prove a new result: 
there exists an aperiodic infinite binary word avoiding the pattern
$x x x^R$.  This is the first avoidability result concerning a
nonuniform morphism proven purely mechanically.
\end{abstract}

\section{Decidability}
\label{decide}

As is well-known, the logical theory $\Th(\Enn,+)$, sometimes called
Presburger arithmetic, is decidable \cite{Presburger:1929,Presburger:1991}.
B\"uchi \cite{Buchi:1960} showed that if we add
the function $V_k(n) = k^e$, for some fixed integer $k \geq 2$,
where $e = \max \{ i \ : \ k^i \, | \, n \}$,
then the resulting theory is still decidable. 
This theory is powerful enough to define finite automata;
for a survey, see \cite{Bruyere&Hansel&Michaux&Villemaire:1994}.  

As a consequence, we have the following theorem 
(see, e.g., \cite{Shallit:2013}):
\begin{theorem}
There is an algorithm that, given a proposition phrased using only the
universal and existential quantifiers, indexing into one or more $k$-automatic
sequences, addition, subtraction, logical operations, and
comparisons, will decide the truth of that proposition.
\label{one}
\end{theorem}
Here, by a $k$-automatic sequence, we mean a sequence $\bf a$
computed by deterministic finite automaton with
output (DFAO) $M = (Q, \Sigma_k, \Delta, \delta, q_0, \kappa) $.
Here $\Sigma_k := \lbrace 0,1,\ldots, k-1 \rbrace$ is the input 
alphabet,
$\Delta$ is the output alphabet,
and outputs are associated with the states given by the map
$\kappa:Q \rightarrow \Delta$ in the following manner:  if $(n)_k$ denotes
the canonical expansion of $n$ in base $k$, then
${\bf a}[n] = \kappa(\delta(q_0, (n)_k))$.  The prototypical example of
an automatic sequence is the Thue-Morse sequence
${\bf t} = t_0 t_1 t_2 \cdots$, the fixed point (starting with $0$) of
the morphism $0 \rightarrow 01$, $1 \rightarrow 10$.

It turns out that many results in the literature about properties of automatic
sequences, for which some had only
long and involved proofs, can be proved purely mechanically using a decision
procedure.
It suffices to express the property as an appropriate logical
predicate, convert the predicate into an automaton accepting
representations of integers for which the predicate is true, and
examine the automaton.
See, for example, the recent papers
\cite{Allouche&Rampersad&Shallit:2009,Goc&Henshall&Shallit:2012,Goc&Saari&Shallit:2013,Goc&Mousavi&Shallit:2013,Goc&Schaeffer&Shallit:2013}. 
Furthermore, in many cases we can explicitly enumerate various aspects
of such sequences, such as subword complexity
\cite{Charlier&Rampersad&Shallit:2012}.

Beyond base $k$, more exotic numeration systems are known, and one
can define automata taking representations in these systems as input.
It turns out that in the so-called Pisot numeration systems, addition
is computable \cite{Frougny:1992a,Frougny&Solomyak:1996},
and hence a theorem analogous to Theorem~\ref{one} holds
for these systems.  See, for example, \cite{Bruyere&Hansel:1997}.  
It is our contention that the power of this approach has not been
widely appreciated, and that
many results, previously proved using long and involved ad hoc techniques,
can be proved with much less effort by phrasing them as logical predicates
and employing a decision procedure.  Furthermore, many enumeration questions
can be solved with a similar approach.

We have implemented a decision algorithm for one such
system; namely, Fibonacci representation.   
In this paper we report on
our results obtained using this implementation.  We have 
reproved many results in the literature purely mechanically, as well as
obtained new results, using this implementation.

The paper is organized as follows.  In Section~\ref{fibrep}, we briefly
recall the details of Fibonacci representation.  In
Section~\ref{proofsf} we report on our mechanical proofs of properties
of the infinite Fibonacci word; we reprove many old results and we
prove some new ones.  In Section~\ref{finitefib} we apply our ideas to
prove results about the finite Fibonacci words.  In
Section~\ref{rotefib} we study a special infinite word, the
Rote-Fibonacci word, and prove many properties of it, including a new
avoidability result.  In Section~\ref{other} we look briefly at another
sequence, the Fibonacci analogue of the Thue-Morse sequence.
In Section~\ref{additive} we apply our methods
to another avoidability problem involving additive squares.
In Section~\ref{enumer} we report on mechanical
proofs of some enumeration results.  Some details about our
implementation are given in the last section.

\section{Fibonacci representation}
\label{fibrep}

Let the Fibonacci numbers be defined, as usual, by
$F_0 = 0$, $F_1 = 1$, and $F_n = F_{n-1} + F_{n-2}$
for $n \geq 2$.  (We caution the reader that some authors use a
different indexing for these numbers.)  

It is well-known, and goes back to Ostrowski \cite{Ostrowski:1922},
Lekkerkerker \cite{Lekkerkerker:1952}, 
and Zeckendorf \cite{Zeckendorf:1972}, that every 
non-negative integer can be represented, in an essentially unique
way, as a sum of Fibonacci numbers $(F_i)_{i\geq 2}$,
subject to the constraint that no two consecutive Fibonacci numbers
are used.  For example, $43 = F_9 + F_6 + F_2$.    Also see
\cite{Carlitz:1968,Fraenkel:1985}.

Such a representation can be written as a binary string 
$a_1 a_2 \cdots a_n$ representing
the integer
$\sum_{1 \leq i \leq n} a_i F_{n+2-i}$.  For example,
the binary string $10010001$ is the Fibonacci representation of $43$.

For $w = a_1 a_2 \cdots a_n \in \Sigma_2^*$, we
define $[a_1 a_2 \cdots a_n]_F := \sum_{1 \leq i \leq n} a_i F_{n+2-i}$,
even if $a_1 a_2 \cdots a_n$ has leading zeroes or consecutive $1$'s.
By $(n)_F$ we mean the {\it canonical} Fibonacci representation for
the integer $n$, having no leading zeroes or consecutive $1$'s.  Note
that $(0)_F = \epsilon$, the empty string.  The language of all
canonical representations of elements of $\Enn$ is 
$\epsilon + 1(0+01)^*$.

Just as Fibonacci representation is the analogue of base-$k$ representation,
we can define the notion of {\it Fibonacci-automatic sequence} as the
analogue of the more familiar notation of $k$-automatic sequence
\cite{Cobham:1972,Allouche&Shallit:2003}.  We say that an infinite word
${\bf a} = (a_n)_{n \geq 0}$ is Fibonacci-automatic if there exists an
automaton with output $M = (Q, \Sigma_2, q_0, \delta, \kappa, \Delta)$
that $a_n = \kappa(\delta(q_0, (n)_F))$ for all $n \geq 0$.  An example
of a Fibonacci-automatic sequence is the infinite Fibonacci word,
$${\bf f} = f_0 f_1 f_2 \cdots =  01001010\cdots$$
which is generated by the following 2-state automaton:

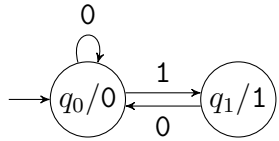
\begin{figure}[H]
\begin{center}
\begin{tikzpicture}[node distance=2cm,on grid,>=stealth',initial text=,auto,
					every state/.style={inner sep=1pt,minimum size=1cm},
					every loop/.style={shorten >=0,looseness=0}]

\node[state,initial]	(q_0)					{$q_0/{\tt 0}$};
\node[state] 			(q_1) [right=of q_0]	{$q_1/{\tt 1}$};

\path[->]	(q_0)		edge [loop above]	node {\tt 0} ()
			(q_0.10)	edge 				node {\tt 1} (q_1.170)
			(q_1.190)	edge 				node {\tt 0} (q_0.350);
\end{tikzpicture}
\end{center}
\caption{Canonical Fibonacci representation DFAO generating the Fibonacci word} \label{fig:f-dfao}
\end{figure}


To compute $f_i$, we express $i$ in canonical Fibonacci representation,
and feed it into the automaton.  Then $f_i$ is the output associated with
the last state reached (denoted by the symbol after the slash).
Another characterization of Fibonacci-automatic sequences can be
found in \cite{Shallit:1988a}.

A basic fact about Fibonacci representation is that addition can
be performed by a finite automaton.  To make this precise, we need to
generalize our notion of Fibonacci representation to $r$-tuples of
integers for $r \geq 1$.  A representation for $(x_1, x_2,\ldots, x_r)$
consists of a string
of symbols $z$ over the alphabet $\Sigma_2^r$, such that the projection
$\pi_i(z)$ over the $i$'th coordinate gives a Fibonacci representation
of $x_i$.  Notice that since the canonical
Fibonacci representations of the individual $x_i$
may have different lengths, padding with leading zeroes will often
be necessary.  A representation for $(x_1, x_2, \ldots, x_r)$ is called
canonical if it has no leading $[0,0,\ldots 0]$ symbols and the projections
into individual coordinates have no occurrences of $11$.  We write the
canonical representation 
as $(x_1, x_2, \ldots, x_r)_F$.  Thus,
for example, the canonical representation for $(9,16)$ is
$[0,1][1,0][0,0][0,1][0,0][1,0]$.

Thus, our claim about addition in Fibonacci representation is that there
exists a 
deterministic finite automaton (DFA) $M_{\rm add}$
that takes input words of the form $[0,0,0]^* (x,y,z)_F$,
and accepts if and only if $x +y =z$.
Thus, for example, $M_{\rm add}$ accepts $[0,0,1][1,0,0][0,1,0][1,0,1]$,
since the three strings obtained by projection are $0101, 0010, 1001$,
which represent, respectively, $4$, $2$, and $6$ in Fibonacci representation.
This result is apparently originally due to Berstel \cite{Berstel:1982};
also see 
\cite{Berstel:1986b,Frougny:1988,Frougny:1991b,Ahlbach&Usatine&Frougny&Pippenger:2013}.

Since this automaton does not appear to have been given explicitly in
the literature and it is essential to our implementation,
we give it here.
The states of $M_{\rm add}$
are $Q = \lbrace 0,1,2,\ldots, 16 \rbrace$, the input alphabet
is $\Sigma_2 \times \Sigma_2 \times \Sigma_2$, the final states
are $F = \lbrace 1,7,11 \rbrace$, the initial state is $q_0 = 1$, and
the transition function $\delta$ is given below.  The automaton is
incomplete, with any unspecified transitions going to a non-accepting
dead state that transitions to itself on all inputs.  This automaton
actually works even for non-canonical expansions having
consecutive $1$'s; an automaton working
only for canonical expansions can easily be obtained by intersection
with the appropriate regular languages.  The state $0$ is a ``dead
state'' that can safely be ignored.

\begin{table}[H]
\begin{center}
\begin{tabular}{c|cccccccc}
  & [0,0,0] & [0,0,1] & [0,1,0] & [0,1,1] & [1,0,0] & [1,0,1] & [1,1,0] & [1,1,1] \\
\hline
 0 &    0 &  0 &  0 &  0 &  0 &  0 &  0 &  0 \\
 1 &    1 &  2 &  3 &  1 &  3 &  1 &  0 &  3 \\
 2 &    4 &  5 &  6 &  4 &  6 &  4 &  7 &  6 \\
 3 &    0 &  8 &  0 &  0 &  0 &  0 &  0 &  0 \\
 4 &    5 &  0 &  4 &  5 &  4 &  5 &  6 &  4 \\
 5 &    0 &  0 &  0 &  0 &  0 &  0 &  9 &  0 \\
 6 &    2 & 10 &  1 &  2 &  1 &  2 &  3 &  1 \\
 7 &    8 & 11 &  0 &  8 &  0 &  8 &  0 &  0 \\
 8 &    3 &  1 &  0 &  3 &  0 &  3 &  0 &  0 \\
 9 &    0 &  0 &  5 &  0 &  5 &  0 &  4 &  5 \\
10 &    0 &  0 &  9 &  0 &  9 &  0 & 12 &  9 \\
11 &    6 &  4 &  7 &  6 &  7 &  6 & 13 &  7 \\
12 &   10 & 14 &  2 & 10 &  2 & 10 &  1 &  2 \\
13 &    0 & 15 &  0 &  0 &  0 &  0 &  0 &  0 \\
14 &    0 &  0 &  0 &  0 &  0 &  0 & 16 &  0 \\
15 &    0 &  3 &  0 &  0 &  0 &  0 &  0 &  0 \\
16 &    0 &  0 &  0 &  0 &  0 &  0 &  5 &  0 \\
\end{tabular}
\end{center}
\caption{Transition table for $M_{\rm add}$ for Fibonacci addition}
\end{table}

We briefly sketch a proof of the correctness of this automaton.  
States can be identified with certain sequences, as follows:
if $x,y,z$ are the identical-length strings arising from projection
of a word that takes $M_{\rm add}$ from the initial state $1$ to the
state $t$, then $t$ is identified with the integer sequence
$([x0^n]_F + [y0^n]_F - [z0^n]_F)_{n \geq 0}$.  With this correspondence,
we can verify the following table by a tedious induction.  In the table
$L_n$ denotes the familiar Lucas numbers, defined by
$L_n = F_{n-1} + F_{n+1}$ for $n \geq 0$ (assuming $F_{-1} = 1$).  
If a sequence $(a_n)_{n \geq 0}$ is the sequence identified 
with a state $t$, then
$t$ is accepting iff $a_0 = 0$.

\begin{table}[H]
\begin{center}
\begin{tabular}{c|c}
state & sequence \\
\hline
1 & {\bf 0} \\
2 & $(-F_{n+2})_{n \geq 0}$ \\
3 & $(F_{n+2})_{n \geq 0}$ \\
4 & $(-F_{n+3})_{n \geq 0}$ \\
5 & $(-F_{n+4})_{n \geq 0}$ \\
6 & $(-F_{n+1})_{n \geq 0}$ \\
7 & $(F_n)_{n \geq 0}$ \\
8 & $(F_{n+1})_{n \geq 0}$ \\
9 & $(-L_{n+2})_{n \geq 0}$ \\
10 & $(-2F_{n+2})_{n \geq 0}$ \\
11 & $(-F_n)_{n \geq 0}$ \\
12 & $(-2F_{n+1})_{n \geq 0}$ \\
13 & $(L_{n+1})_{n \geq 0}$ \\
14 & $(-3F_{n+2})_{n \geq 0}$ \\
15 & $(2F_{n+1})_{n \geq 0}$ \\
16 & $(-2F_n-3L_n)_{n \geq 0}$
\end{tabular}
\end{center}
\caption{Identification of states with sequences}
\end{table}

Note that the state $0$ actually represents a set of sequences, not just
a single sequence.  The set corresponds to those representations
that are so far ``out of synch'' that they can never ``catch up''
to have $x +y = z$, no matter how many digits are appended.

\begin{remark}
We note that, in the spirit of the paper, this adder itself can,
in principle, be
checked mechanically (in $\Th(\Enn, 0)$, of course!), as follows:

First we show the adder $\cal A$ is specifying a function of $x$ and $y$.  To do
so, it suffices to check that
$$ \forall x \ \forall y \ \exists z \ {\cal A}(x,y,z)$$
and
$$ \forall x \ \forall y \ \forall z \ \forall z' \ {\cal A}(x,y,z) \wedge {\cal A}(x,y,z') \implies z = z' .$$
The first predicate says that there is at least one sum of $x$ and $y$
and the second says that there is at most one.

If both of these are verified, we know that $\cal A$ computes a function
$A= A(x,y)$.

Next, we verify associativity, which amounts to checking that
$$ \forall x \ \forall y \ \forall z \ A(A(x,y),z) = A(x,A(y,z)) .$$
We can do this by checking that
$$ \forall x\ \forall y \ \forall z \ \forall w
\ \forall r \ \forall s \ \forall t \ ({\cal A}(x,y,r) \ \wedge \ 
{\cal A}(r,z,t) \ \wedge \ {\cal A}(y,z,s) ) \ \implies \ {\cal A}(x,s,t) . $$

Finally, we ensure that $\cal A$ is an adder by induction.  First, we
check that
$\forall x \ A(x,0) = x$, which amounts to
$$ \forall x\ \forall y \ {\cal A}(x,0,y) \iff x = y .$$

Second, we check that
if $A(x,1) = y$ then $x < y$ and there does not exist $z$ such that $x < z < y$.
This amounts to
$$ \forall x, y, {\cal A}(x,1,y) \implies ((x < y) \ \wedge \ 
	\neg \exists z \ (x<z) \ \wedge \ (z<y) ) .$$

This last condition shows that
$A(x,1) = x+1$. By associativity $A(x,y+1) = A(x,A(y,1)) = A(A(x,y),1) = A(x,y) + 1$. By induction, $A(x,y) = A(x,0)+y = x+y$, so we are done.
\end{remark}

Another basic fact about Fibonacci representation is that, for
canonical representations containing no two consecutive $1$'s or
leading zeroes, the radix order on representations is the same 
as the ordinary ordering on $\Enn$.    It follows that a very
simple automaton can, on input $(x,y)_F$, decide whether $x < y$.

Putting this all together, we get the analogue of
Theorem~\ref{one}:

\begin{proc}[Decision procedure for Fibonacci-automatic words] \label{proc:Fib-auto-decide} \ \\
{\bf Input:} $m,n \in \Enn$, $m$ DFAOs witnessing Fibonacci-automatic words ${\bf w}_1,{\bf w}_2,\dots,{\bf w}_m$, a first-order proposition with $n$ free variables $\varphi(v_1,v_2,\dots,v_n)$ using constants and relations definable in $\Th(\Enn,0,1,+)$ and indexing into ${\bf w}_1,{\bf w}_2,\dots,{\bf w}_m$. \\
{\bf Output:} DFA with input alphabet $\Sigma_2^n$ accepting $\{ (k_1,k_2,\dots,k_n)_F \st \varphi(k_1,k_2,\dots,k_n) \text{ holds} \}$.
\end{proc}

We remark that there was substantial skepticism that any implementation
of a decision procedure for Fibonacci-automatic words would be practical,
for two reasons:
\begin{itemize}

\item first, because the running time 
is bounded above by an expression of the form
$$2^{2^{\Ddots^{ 2^{p(N)}}}}$$
where $p$ is a polynomial, $N$ is the number of states in the original
automaton specifying the word in question,
and the number of exponents in the tower
is one less than the number of quantifiers in the logical formula
characterizing the property being checked.

\item second, because of
the complexity of checking addition (15 states) compared to the 
analogous automaton for base-$k$ representation (2 states).
\end{itemize}
Nevertheless, we were able to carry out nearly all
the computations described in this paper 
in a matter of a few seconds on an ordinary laptop.

\section{Mechanical proofs of properties of the infinite Fibonacci word}
\label{proofsf}

Recall that a word $x$, whether finite or infinite, is said to have
period $p$ if $x[i] = x[i+p]$ for all $i$ for which this equality is meaningful.
Thus, for example, the English word ${\tt alfalfa}$ has period $3$.
The {\it exponent} of a finite word $x$, written $\exp(x)$, is
$|x|/P$, where $P$ is the smallest period of $x$.
Thus $\exp({\tt alfalfa}) = 7/3$.

If $\bf x$ is an infinite word with a finite period, we say it is
{\it ultimately periodic}.  An infinite
word $\bf x$ is ultimately periodic if and only if there are finite
words $u, v$ such that $x = uv^\omega$, where $v^\omega= vvv \cdots$.

A nonempty word of the form $xx$ is called a {\it square}, and a 
nonempty word of
the form $xxx$ is called a {\it cube}.  More generally, a nonempty word
of the form $x^n$ is called an $n$'th power.
By the {\it order} of a square $xx$,
cube $xxx$, or $n$'th power $x^n$, we mean the length $|x|$.

The infinite Fibonacci word ${\bf f} = 01001010 \cdots = f_0 f_1 f_2 \cdots$
can be described in many
different ways.  In addition to our definition in terms of automata,
it is also the fixed point of 
the morphism $\varphi(0) = 01$ and $\varphi(1) = 0$.  This word
has been studied extensively in the literature; see, for example,
\cite{Berstel:1980b,Berstel:1986b}.  

In the next subsection, we use our implementation to prove a variety of
results about repetitions in $\bf f$.

\subsection{Repetitions}
\label{repe-subsec}

\begin{theorem}
The word $\bf f$ is not ultimately periodic.
\end{theorem}

\begin{proof}
We construct a predicate asserting that the integer
$p \geq 1$ is a period of some suffix of $\bf f$:
$$ (p \geq 1) \ \wedge \ \exists n \ \forall i \geq n\  {\bf f}[i] =
{\bf f}[i+p] . $$
(Note:  unless otherwise indicated, whenever we refer to a variable in
a predicate, the range of the variable is assumed to be
$\Enn = \lbrace 0, 1, 2, \ldots \rbrace$.)
From this predicate, using our program, we constructed
an automaton accepting the language
$$ L = 0^*\ \lbrace (p)_F \ : \ (p \geq 1) \ \wedge \ \exists n
\ \forall i \geq n \ {\bf f}[i] = {\bf f}[i+p] \rbrace .$$
This automaton accepts the empty language, and so it follows
that ${\bf f}$ is not ultimately periodic.


Here is the log of our program:
\begin{verbatim}
p >= 1 with 4 states, in 60ms
 i >= n with 7 states, in 5ms
  F[i] = F[i + p] with 12 states, in 34ms
   i >= n => F[i] = F[i + p] with 51 states, in 15ms
    Ai i >= n => F[i] = F[i + p] with 3 states, in 30ms
     p >= 1 & Ai i >= n => F[i] = F[i + p] with 2 states, in 0ms
      En p >= 1 & Ai i >= n => F[i] = F[i + p] with 2 states, in 0ms
overall time: 144ms
\end{verbatim}
The largest intermediate automaton during the computation had 63
states.

A few words of explanation are in order:  here ``{\tt F}'' 
refers to the sequence
$\bf f$, and ``{\tt E}'' is our abbreviation for $\exists$ and
``{\tt A}'' is our abbreviation for $\forall$.  The symbol ``{\tt =>}'' 
is logical
implication, and ``{\tt \&}'' is logical and.
\end{proof}

From now on, whenever we discuss the language accepted by an automaton,
we will omit the $0^*$ at the beginning.  

We recall an old result of Karhum\"aki \cite[Thm.~2]{Karhumaki:1983}:

\begin{theorem}
$\bf f$ contains no fourth powers.
\end{theorem}

\begin{proof}
We create a predicate for the orders of all fourth powers occurring
in $\bf f$:
$$(n > 0) \ \wedge \ \exists i \ \forall t<3n \  {\bf f}[i+t] = {\bf f}[i+n+t] .
$$

The resulting automaton accepts nothing, so there are no fourth powers.


\begin{verbatim}
n > 0 with 4 states, in 46ms
 t < 3 * n with 30 states, in 178ms
  F[i + t] = F[i + t + n] with 62 states, in 493ms
   t < 3 * n => F[i + t] = F[i + t + n] with 352 states, in 39ms
    At t < 3 * n => F[i + t] = F[i + t + n] with 3 states, in 132ms
     Ei At t < 3 * n => F[i + t] = F[i + t + n] with 2 states, in 0ms
      n > 0 & Ei At t < 3 * n => F[i + t] = F[i + t + n] with 2 states, in 0ms
overall time: 888ms
\end{verbatim}
\end{proof}

The largest intermediate automaton in the computation had 952 states.

Next, we move on to a description of the orders of squares occurring
in $\bf f$.  An old
result of S\'e\'ebold \cite{Seebold:1985b} (also see
\cite{Iliopoulos&Moore&Smyth:1997,Fraenkel&Simpson:1999}) states

\begin{theorem}
All squares in $\bf f$ are of order $F_n$ for some $n \geq 2$.
Furthermore, for all $n \geq 2$, there exists a square of order
$F_n$ in $\bf f$.
\label{squares}
\end{theorem}

\begin{proof}
We create a predicate for the lengths of squares:
$$(n > 0) \ \wedge \ \exists i \ \forall t<n \  {\bf f}[i+t] = {\bf f}[i+n+t] .$$


When we run this predicate, we obtain an automaton that accepts
exactly the language $10^*$.    Here is the log file:

\begin{verbatim}
n > 0 with 4 states, in 38ms
 t < n with 7 states, in 5ms
  F[i + t] = F[i + t + n] with 62 states, in 582ms
   t < n => F[i + t] = F[i + t + n] with 92 states, in 12ms
    At t < n => F[i + t] = F[i + t + n] with 7 states, in 49ms
     Ei At t < n => F[i + t] = F[i + t + n] with 3 states, in 1ms
      n > 0 & Ei At t < n => F[i + t] = F[i + t + n] with 3 states, in 0ms
overall time: 687ms
\end{verbatim}
\end{proof}

The largest intermediate automaton had 236 states.

We can easily get much, much more information about the square
occurrences in $\bf f$.  The positions of all squares
in $\bf f$ were computed by Iliopoulos, Moore, and Smyth 
\cite[\S~2]{Iliopoulos&Moore&Smyth:1997}, but their description is
rather complicated and takes 5 pages to prove.  Using our approach, we created
an automaton accepting the language
$$ 
\{ (n,i)_F \ : \ (n > 0) \ \wedge \ \forall t<n \  {\bf f}[i+t] = {\bf f}[i+n+t]
\} . $$


This automaton has only 6 states and efficiently encodes the
orders and starting positions of each square in $\bf f$.  During
the computation, the largest intermediate automaton had 236 states.
Thus we have proved

\begin{theorem}
The language
$$ \{ (n,i)_F \ : \ \text{there is a square of order $n$ beginning at
	position $i$ in {\bf f}} \}$$
is accepted by the automaton in Figure~\ref{squareorders}.

\begin{figure}[H]
\begin{center}
\includegraphics[width=6.5in]{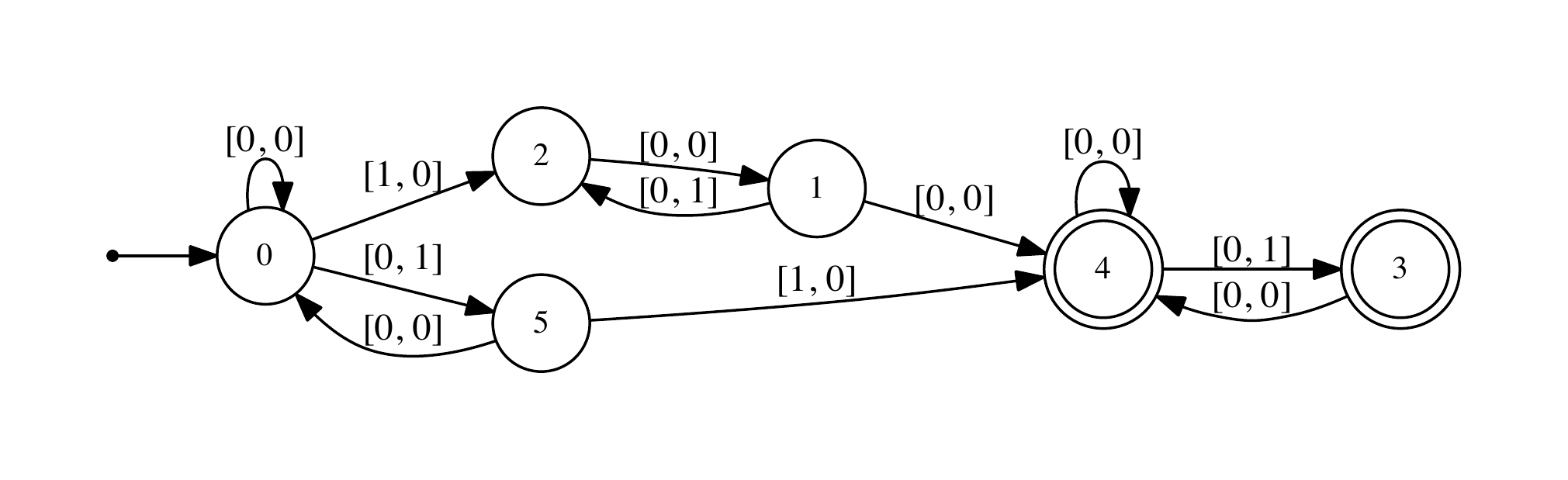}
\caption{Automaton accepting orders and positions of all squares in $\bf f$}
\label{squareorders}
\end{center}
\end{figure}
\end{theorem}

Next, we examine the cubes in $\bf f$.
Evidently Theorem~\ref{squares} implies that any cube in $\bf f$
must be of order $F_n$ for some $n$.   However, not every order
occurs.

\begin{theorem}
The cubes in $\bf f$ are of order $F_n$ for $n \geq 4$, and a cube of
each such order occurs.
\end{theorem}

\begin{proof}
We use the predicate
$$(n > 0) \ \wedge \ \exists i \ \forall t<2n \  {\bf f}[i+t] = {\bf f}[i+n+t] .$$

When we run our program, we obtain an automaton accepting exactly the
language $(100)0^*$, which corresponds to $F_n$ for $n \geq 4$.

\begin{verbatim}
n > 0 with 4 states, in 34ms
 t < 2 * n with 16 states, in 82ms
  F[i + t] = F[i + t + n] with 62 states, in 397ms
   t < 2 * n => F[i + t] = F[i + t + n] with 198 states, in 17ms
    At t < 2 * n => F[i + t] = F[i + t + n] with 7 states, in 87ms
     Ei At t < 2 * n => F[i + t] = F[i + t + n] with 5 states, in 1ms
      n > 0 & Ei At t < 2 * n => F[i + t] = F[i + t + n] with 5 states, in 0ms
overall time: 618ms
\end{verbatim}
\end{proof}

The largest intermediate automaton had 674 states.

Next, we encode the orders and positions of all cubes.  We build a 
DFA accepting the language
$$ 
\{ (n,i)_F \ : \ (n > 0) \ \wedge \ \forall t<2n \  {\bf f}[i+t] = {\bf f}[i+n+t]
\} . $$


\begin{theorem}
The language
$$ \{ (n,i)_F \ : \ \text{there is a cube of order $n$ beginning at
	position $i$ in {\bf f}} \}$$
is accepted by the automaton in Figure~\ref{cubeorders}.

\begin{figure}[H]
\begin{center}
\includegraphics[width=6.5in]{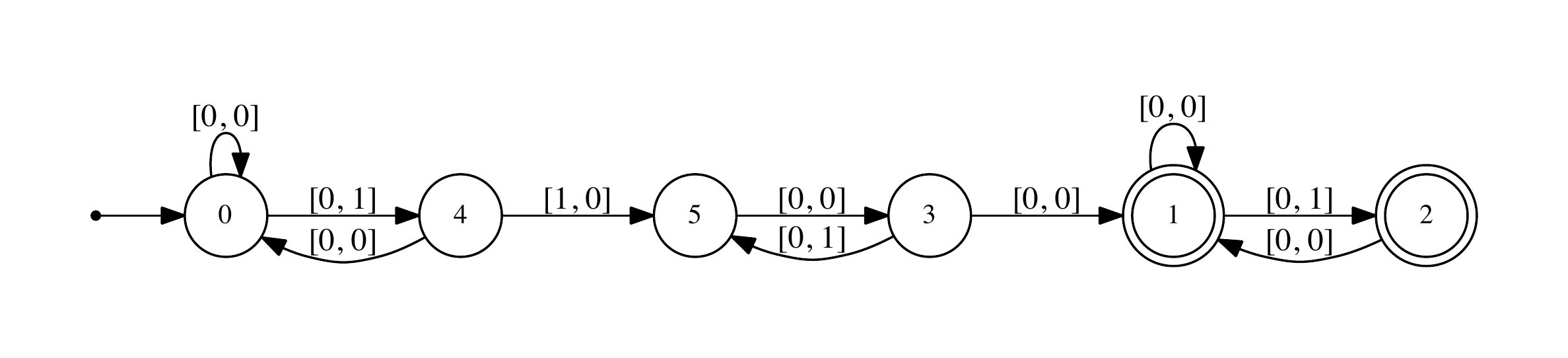}
\caption{Automaton accepting orders and positions of all cubes in $\bf f$}
\label{cubeorders}
\end{center}
\end{figure}
\end{theorem}

Finally, we consider all the maximal repetitions in $\bf f$.  
Let $p(x)$ denote the length of the least period of $x$.  
If ${\bf x} = a_0 a_1 \cdots$, by ${\bf x}[i..j]$
we mean $a_i a_{i+1} \cdots a_j$.
Following Kolpakov and Kucherov \cite{Kolpakov&Kucherov:1999a}, we
say that ${\bf f}[i..i+n-1]$ is a {\it maximal repetition} if
\begin{itemize}
\item[(a)] $p({\bf f}[i..i+n-1]) \leq n/2$;
\item[(b)] $p({\bf f}[i..i+n-1]) < p({\bf f}[i..i+n]) $;
\item[(c)] If $i > 0$ then $p({\bf f}[i..i+n-1]) < p({\bf f}[i-1..i+n-1])$.
\end{itemize}

\begin{theorem}
The factor ${\bf f}[i..i+n-1]$ is a maximal repetition of $\bf f$
iff $(n,i)_F$ is accepted by the automaton depicted in 
Figure~\ref{maxreps2}.


\begin{figure}[H]
\begin{center}
\includegraphics[width=3.5in]{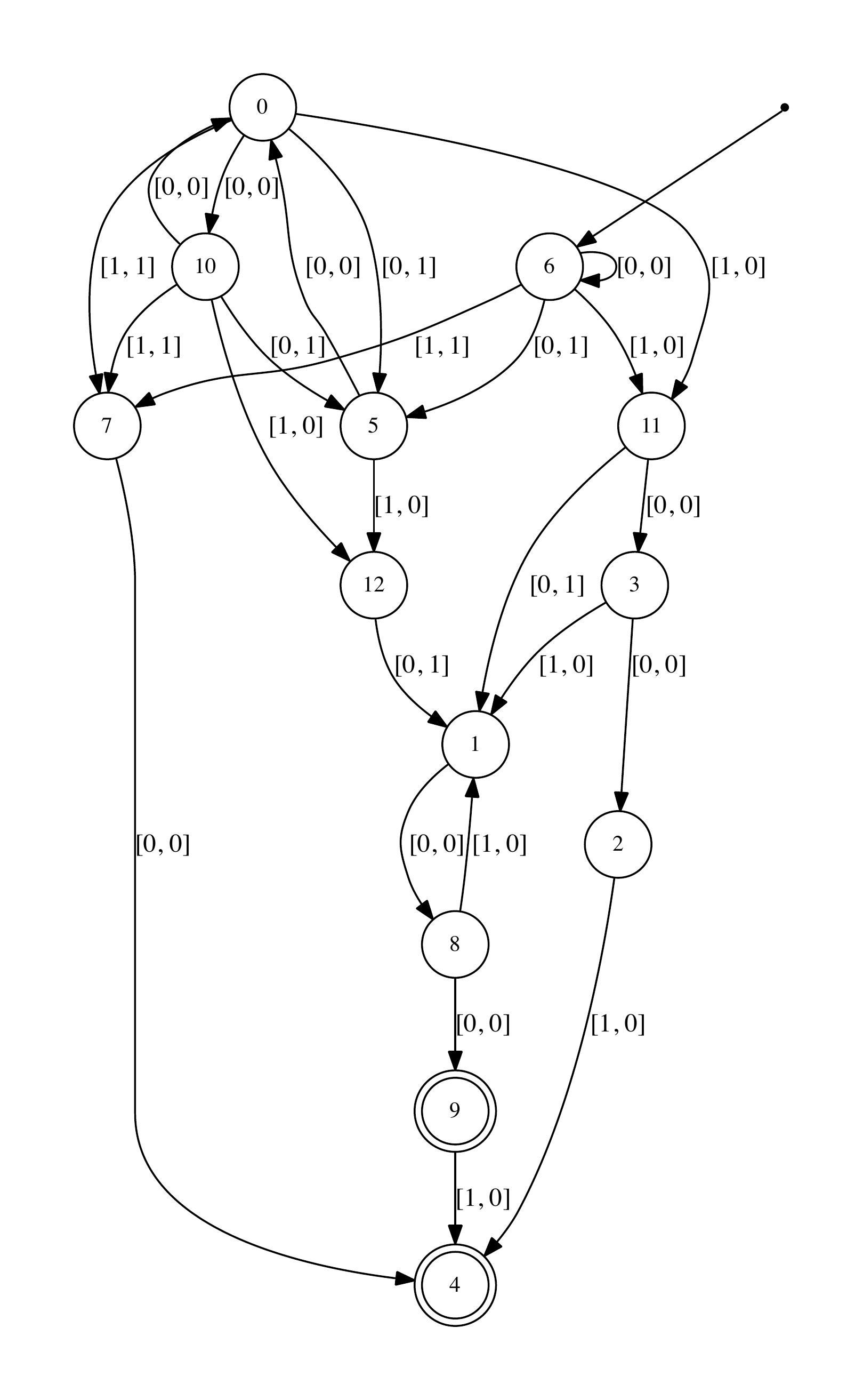}
\caption{Automaton accepting occurrences of maximal repetitions in $\bf f$}
\label{maxreps2}
\end{center}
\end{figure}
\end{theorem}

An {\it antisquare} is a nonempty word of the form $x \overline{x}$, where
$\overline{x}$ denotes the complement of $x$ ($1$'s changed to $0$'s and
vice versa).  Its order is $|x|$.  For a new (but small) result we prove

\begin{theorem}
The Fibonacci word $\bf f$ contains exactly four antisquare factors:
$01, 10, 1001, $ and $10100101$.
\end{theorem}

\begin{proof}
The predicate for having an antisquare of length $n$ is
$$ \exists i \ \forall k < n \ {\bf f}[i+k] \not= {\bf f}[i+k+n] .$$
When we run this we get the automaton depicted in Figure~\ref{antisquare},
specifying that the only possible orders are $1$, $2$, and $4$, which
correspond to words of length $2$, $4$, and $8$.

\begin{figure}[H]
\begin{center}
\includegraphics[width=5.5in]{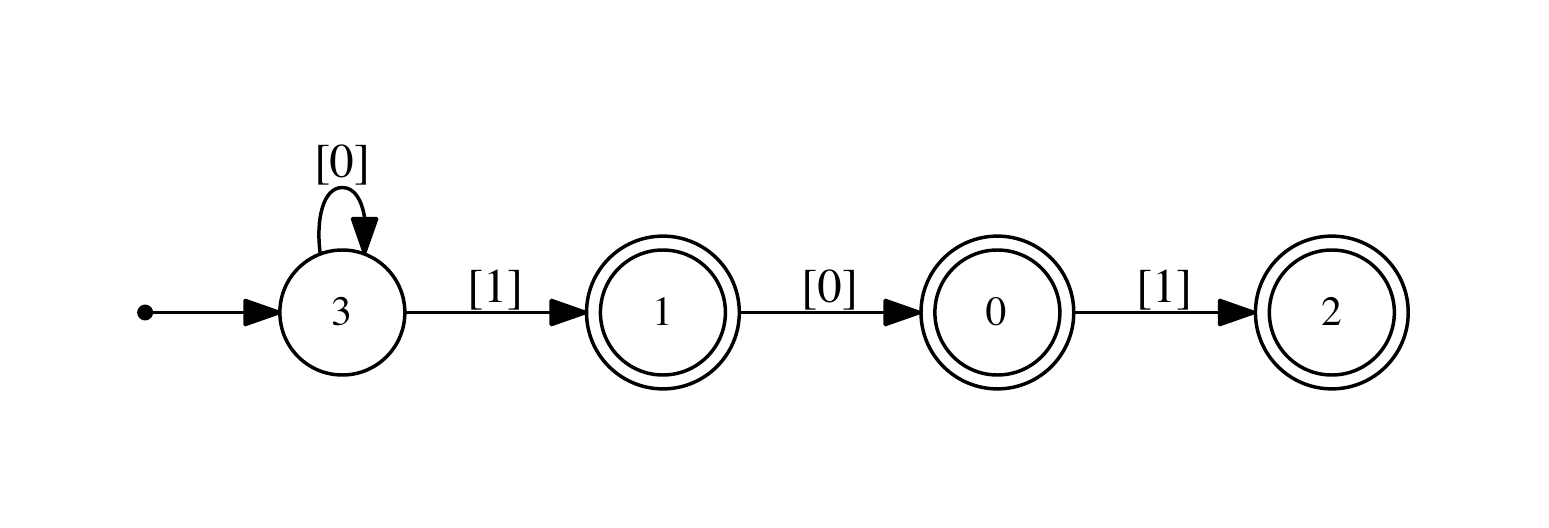}
\caption{Automaton accepting orders of antisquares in $\bf f$}
\label{antisquare}
\end{center}
\end{figure}

Inspection of the factors of these lengths proves the result.
\end{proof}

\subsection{Palindromes and antipalindromes}

We now turn to a characterization of the palindromes in $\bf f$.
Using the predicate
$$ \exists i \ \forall j<n \ {\bf f}[i+j] = {\bf f}[i+n-1-j], $$
we specify those lengths $n$ for which there is a palindrome of 
length $n$.      Our program then
recovers the following
result of Chuan \cite{Chuan:1993b}:

\begin{theorem}
There exist palindromes of every length $\geq 0$ in $\bf f$.
\end{theorem}

We could also characterize the positions of all 
nonempty palindromes.  The resulting 21-state automaton is
not particularly enlightening, but is included here to show
the kind of complexity that can arise.

\begin{figure}[H]
\begin{center}
\includegraphics[width=4.8in]{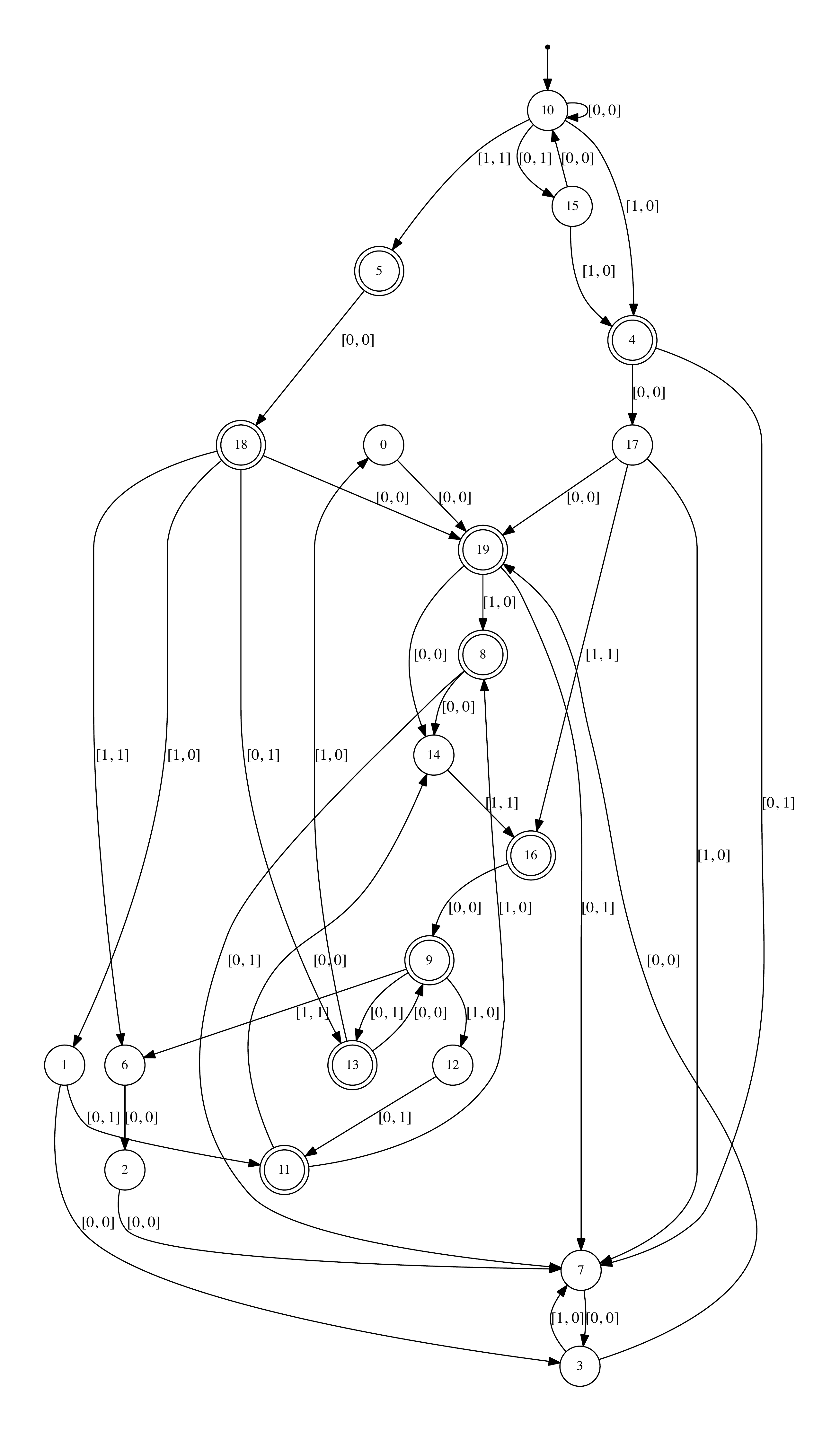}
\caption{Automaton accepting orders and positions of all 
nonempty palindromes in $\bf f$}
\label{palindrome-orders}
\end{center}
\end{figure}

Although the automaton in Figure~\ref{palindrome-orders}
encodes all palindromes, more specific information is
a little hard to deduce from it.  For example, let's prove
a result of Droubay \cite{Droubay:1995}:

\begin{theorem}
The Fibonacci word $\bf f$ has exactly one palindromic factor of
length $n$ if $n$ is even, and exactly
two palindromes of
length $n$ if $n$ odd.
\end{theorem}

\begin{proof}
First, we obtain an expression for the lengths $n$ for which there is
exactly one palindromic factor of length $n$.
\begin{multline*}
\exists i \ (\forall t<n \ {\bf f}[i+t] = {\bf f}[i+n-1-t]) 
\ \wedge \  \\
\forall j \ (\forall s<n\ {\bf f}[j+s] = {\bf f}[j+n-1-s]) \implies
( \forall u<n\ {\bf f}[i+u] = {\bf f}[j+u])
\end{multline*}
The first part of the predicate asserts that ${\bf f}[i..i+n-1]$
is a palindrome, and the second part asserts that any
palindrome ${\bf f}[j..j+n-1]$ of the same length must in fact be equal to
${\bf f}[i..i+n-1]$.

When we run this predicate through our program we get the automaton
depicted below in Figure~\ref{onepal}.

\begin{figure}[H]
\begin{center}
\includegraphics[width=6.5in]{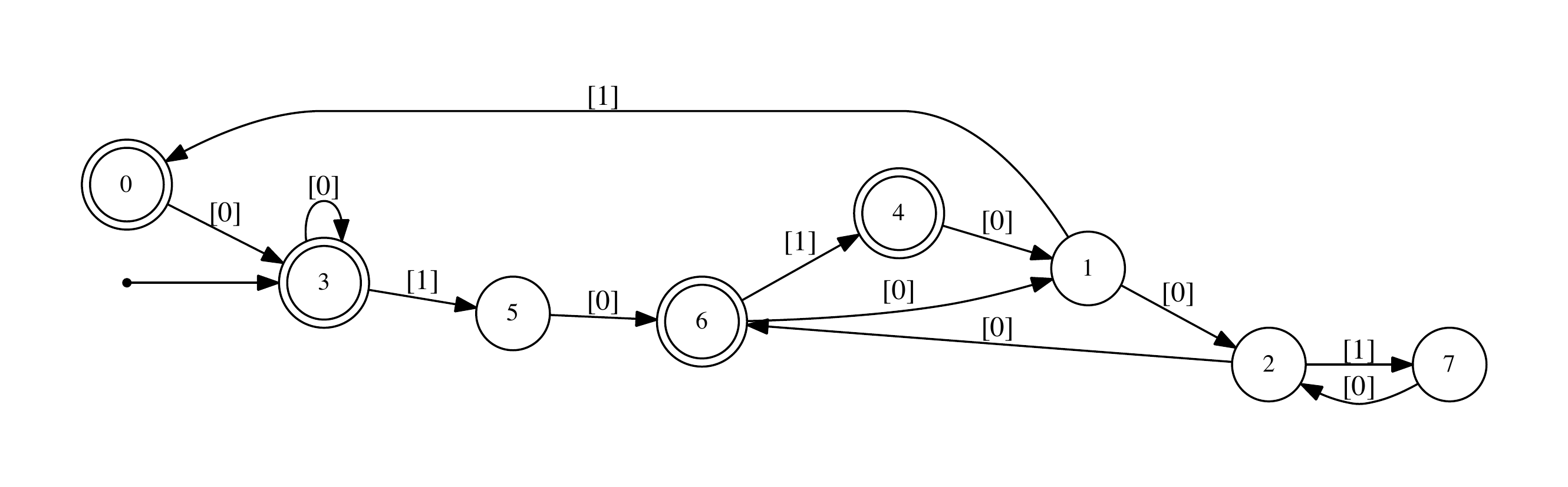}
\caption{Automaton accepting lengths with exactly one palindrome}
\label{onepal}
\end{center}
\end{figure}

It may not be obvious, but this automaton accepts exactly the
Fibonacci representations of the even numbers.  The easiest way to
check this is to use our program on the predicate
$\exists i \ n = 2i$ and verify that the resulting automaton
is isomorphic to that in Figure~\ref{onepal}.

Next, we write down a predicate for the existence of exactly two
distinct palindromes of length $n$.  The predicate asserts the
existence of two palindromes ${\bf x}[i..i+n-1]$ and
${\bf x}[j..j+n-1]$ that are distinct and for which any palindrome
of the same length must be equal to one of them.

\begin{multline*}
\exists i\ \exists j\  (\forall t<n\ {\bf f}[i+t] = {\bf f}[i+n-1-t])
\ \wedge \ 
(\forall s<n\ {\bf f}[j+s] = {\bf f}[j+n-1-s]) 
\ \wedge \  \\
(\exists m<n\ {\bf f}[i+m] \not= {\bf f}[j+m]) 
\ \wedge \  \\
( \forall u  (\forall k < n\ {\bf f}[u+k] = {\bf f}[u+n-1-k]) \implies
(( \forall l<n\  {\bf f}[u+l] = {\bf f}[i+l]) \ \vee \ (\forall p<n \ {\bf f}[u+p] = {\bf f}[j+p]))) 
\end{multline*}

Again, running this through our program gives us an automaton accepting
the Fibonacci representations of the odd numbers.  We omit the automaton.
\end{proof}

The prefixes are factors of particular interest.  Let us determine
which prefixes are palindromes:

\begin{theorem}
The prefix ${\bf f}[0..n-1]$ of length $n$ is a palindrome if and
only if $n = F_i - 2$ for some $i \geq 3$.
\end{theorem}

\begin{proof}
We use the predicate
$$ \forall i<n\ {\bf f}[i] = {\bf f}[n-1-i]$$
obtaining an automaton accepting
$\epsilon + 1 + 10(10)^*(0+01)$, which
are precisely the representations of $F_i - 2$.
\end{proof}

Next, we turn to the property of ``mirror invariance''.  We say
an infinite word $\bf w$ is mirror-invariant 
if whenever $x$ is a factor of $\bf w$, then so is $x^R$.
We can check this for $\bf f$ by creating a predicate for the assertion that
for each factor $x$ of length $n$, the factor $x^R$ appears somewhere
else:
$$\forall i \geq 0 \ \exists j \text{ such that }
	{\bf f}[i..i+n-1] = {\bf f}[j..j+n-1]^R .$$
When we run this through our program we discover that it accepts
the representations of all $n \geq 0$.  Here is the log:

\begin{verbatim}
t < n with 7 states, in 99ms
 F[i + t] = F[j + n - 1 - t] with 264 states, in 7944ms
  t < n => F[i + t] = F[j + n - 1 - t] with 185 states, in 89ms
   At t < n => F[i + t] = F[j + n - 1 - t] with 35 states, in 182ms
    Ej At t < n => F[i + t] = F[j + n - 1 - t] with 5 states, in 2ms
     Ai Ej At t < n => F[i + t] = F[j + n - 1 - t] with 3 states, in 6ms
overall time: 8322ms
\end{verbatim}

Thus we have proved:
\begin{theorem}
The word ${\bf f}$ is mirror invariant.
\end{theorem}

An {\it antipalindrome} is a word $x$ satisfying $x = \overline{x^R}$.
For a new
(but small) result, we determine all possible antipalindromes in
$\bf f$:

\begin{theorem}
The only nonempty antipalindromes in $\bf f$ are 
$01$, $10$, $(01)^2$, and $(10)^2$.
\end{theorem}

\begin{proof}
Let us write a predicate specifying that ${\bf f}[i..i+n-1]$ is a
nonempty antipalindrome, and further that it is a 
first occurrence
of such a factor:
$$
(n > 0) \ \wedge\ (\forall j<n \ {\bf f}[i+j] \not= {\bf f}[i+n-1-j]) \ \wedge \ 
(\forall i' < i \ \exists j<n\  {\bf f}[i'+j] \not= {\bf f}[i+j]) .
$$

When we run this through our program, the language of $(n,i)_F$ satisfying
this predicate is accepted by the following automaton:

\begin{figure}[H]
\begin{center}
\includegraphics[width=5in]{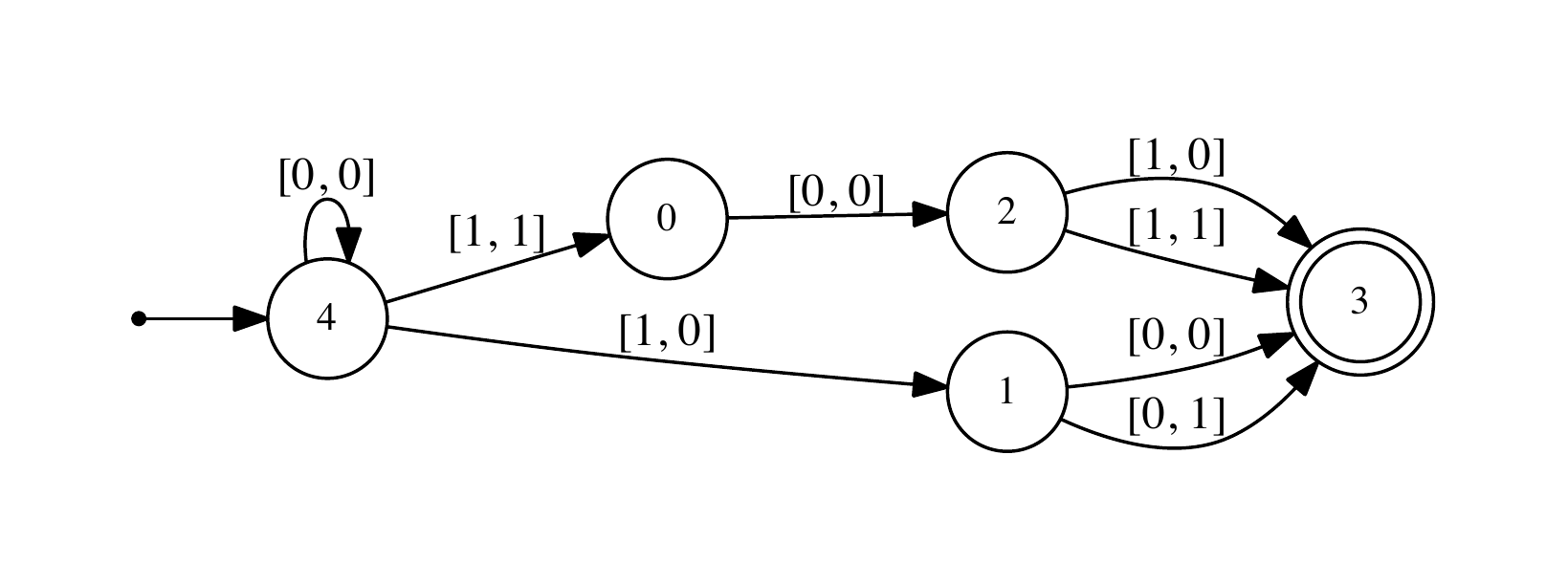}
\caption{Automaton accepting orders and positions of first occurrences of
nonempty antipalindromes in $\bf f$}
\label{antipal}
\end{center}
\end{figure}

It follows that the only $(n,i)$ pairs accepted are
$(2,0), (2,1), (4,3), (4,4)$, corresponding, respectively, to
the strings $01$, $10$, $(01)^2$, and $(10)^2$.
\end{proof}

\subsection{Special factors}

Next we turn to special factors. It is well-known (and we will prove it
in Theorem~\ref{sturmcomp} below), that ${\bf f}$ has exactly $n+1$ distinct
factors of length $n$ for each $n \geq 0$.  This implies that there
is exactly one factor $x$ of each length $n$ with the property that
both $x0$ and $x1$ are factors.  Such a factor is called {\it right-special}
or sometimes just {\it special}.
We can write a predicate that expresses the assertion that the
factor ${\bf f}[i..i+n-1]$ is the unique special factor of length $n$,
and furthermore, that it is the first occurrence of that factor, as follows:
\begin{multline*}
(\forall i' < i \ \exists s < n \  {\bf f}[i'+s] \not= {\bf f}[i+s])
\ \wedge \ 
\exists j \ \exists k \ ((\forall t < n\ {\bf f}[j+t] = {\bf f}[i+t]) \\
\wedge \ 
(\forall u < n\ {\bf f}[k+u] = {\bf f}[i+u]) \ \wedge \ 
({\bf f}[j+n] \not= {\bf f}[k+n])) .
\end{multline*}

\begin{theorem}
The automaton depicted below in Figure~\ref{special} accepts the language
$$\{ (i,n)_F \ : \ \text{the factor } {\bf f}[i..i+n-1] 
\text{ is the first occurrence of the unique special factor of length $n$} \} .$$

\begin{figure}[H]
\begin{center}
\includegraphics[width=3.5in]{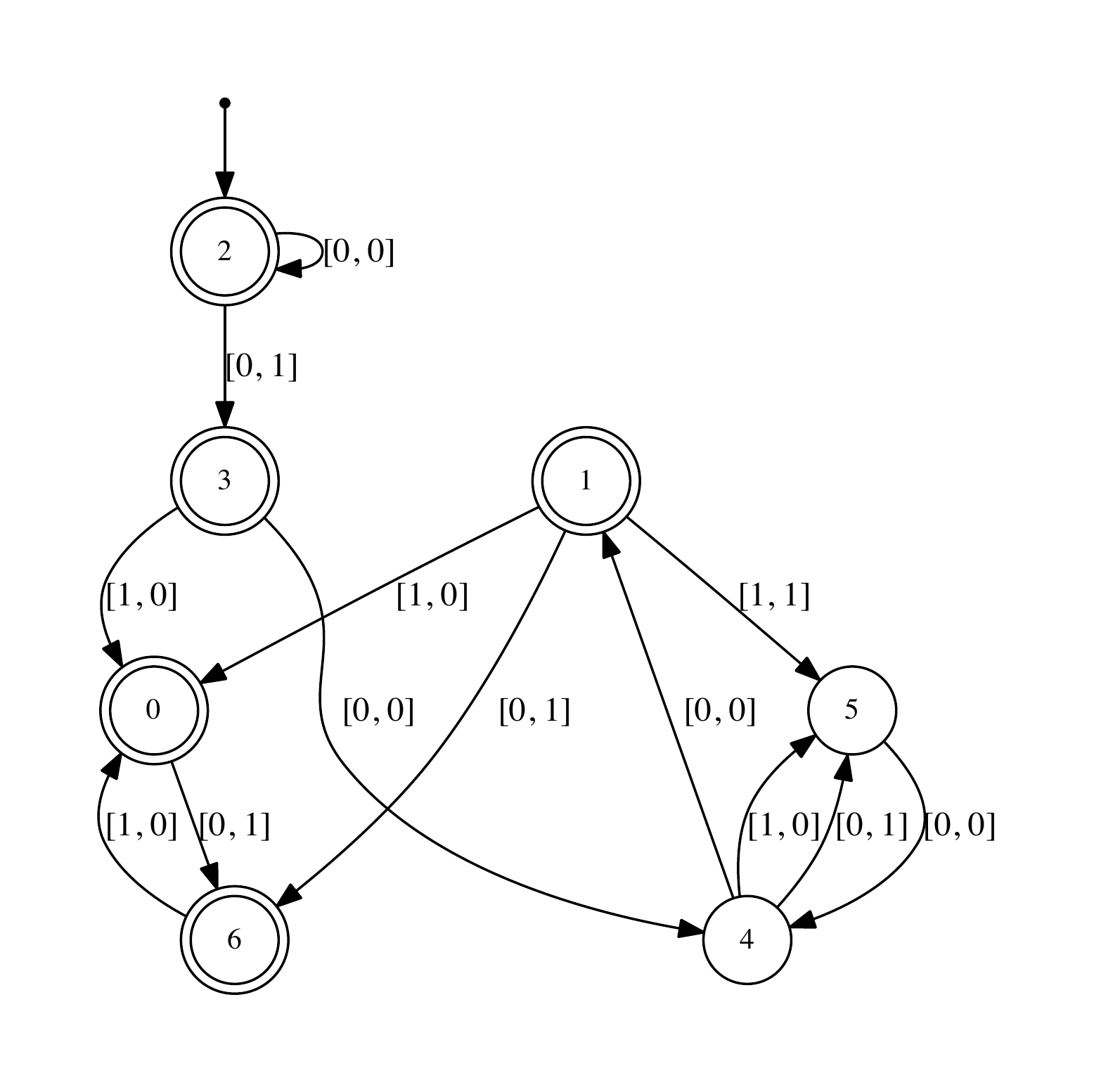}
\caption{Automaton accepting first positions and lengths of special factors
in $\bf f$}
\label{special}
\end{center}
\end{figure}
\end{theorem}

Furthermore it is known (e.g., \cite[Lemma 5]{Pirillo:1997}) that

\begin{theorem}
The unique special factor of length $n$ is ${\bf f}[0..n-1]^R$.
\end{theorem}

\begin{proof}
We create a predicate that says that if a factor is special then it
matches ${\bf f}[0..n-1]^R$.  When we run this we discover that all 
lengths are accepted.
\end{proof}

\subsection{Least periods}

We now turn to least periods of factors of ${\bf f}$; see
\cite{Saari:2007} and \cite{Epple&Siefken:2014} and 
\cite[Corollary 4]{Currie&Saari:2009}.

Let $P$ denote the assertion that $n$ is a period of
the factor ${\bf f}[i..j]$, as follows:
\begin{eqnarray*}
P(n,i,j) &=& {\bf f}[i..j-n] = {\bf f}[i+n..j]  \\
&=& \forall \ t\ \text{ with $i \leq t \leq j-n$ we have } 
	{\bf f}[t] = {\bf f}[t+n] .
\end{eqnarray*}
Using this, we can express the predicate $LP$ that $n$ is the least
period of ${\bf f}[i..j]$:
$$ LP(n,i,j) = P(n,i,j) \text{ and } \forall n' 
\text{ with } 1 \leq n' < n \  \neg P(n',i,j).$$
Finally, we can express the predicate that $n$ is a least period
as follows
$$L(n) = \exists i, j \geq 0 \text{ with $0 \leq i+n \leq j-1$ }
	LP(n, i, j) .$$

Using an implementation of this, we can
reprove the following theorem of Saari 
\cite[Thm.~2]{Saari:2007}:

\begin{theorem}
If a word $w$ is a nonempty factor of the Fibonacci word, then
the least period of $w$ is a Fibonacci number $F_n$ for $n \geq 2$.
Furthermore, each such period occurs.
\end{theorem}

\begin{proof}
We ran our program on the appropriate predicate and found the resulting
automaton accepts $10^+$, corresponding to $F_n$ for $n \geq 2$.
\end{proof}

Furthermore, we can actually encode information about all least periods.
The automaton depicted in Figure~\ref{leastp} accepts triples $(n,p,i)$
such that $p$ is a least period of ${\bf f}[i..i+n-1]$.

\begin{figure}[H]
\begin{center}
\includegraphics[width=5.5in]{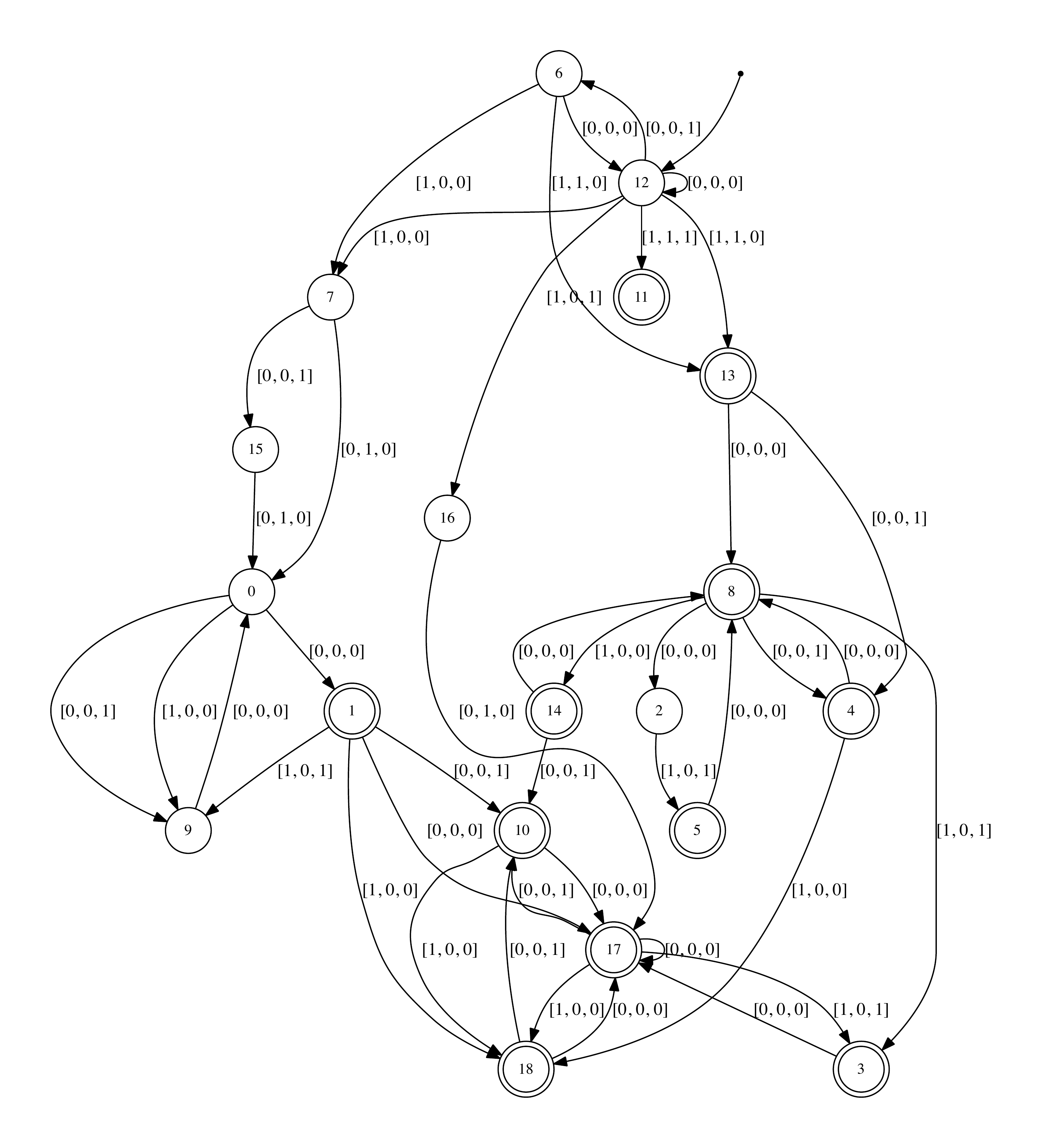}
\caption{Automaton encoding least periods of all factors
in $\bf f$}
\label{leastp}
\end{center}
\end{figure}

We also have the following result, which seems to be new.

\begin{theorem}
Let $n \geq 1$, and define
$\ell(n)$ to be the smallest integer that is the 
least period of some length-$n$ factor of $\bf f$.
Then $\ell(n) = F_j$ for $j \geq 1$ if
$L_j-1 \leq n \leq L_{j+1}-2$, where $L_j$ is the $j$'th Lucas number
defined in Section~\ref{fibrep}.
\label{allpers}
\end{theorem}

\begin{proof}
We create an automaton accepting $(n,p)_F$ such that (a) there exists
at least one length-$n$ factor of period $p$ and (b) for all
length-$n$ factors $x$, if $q$ is a period of $x$, then $q \geq p$.
This automaton is depicted in Figure~\ref{least-period-over}
below.


\begin{figure}[H]
\begin{center}
\includegraphics[width=6.5in]{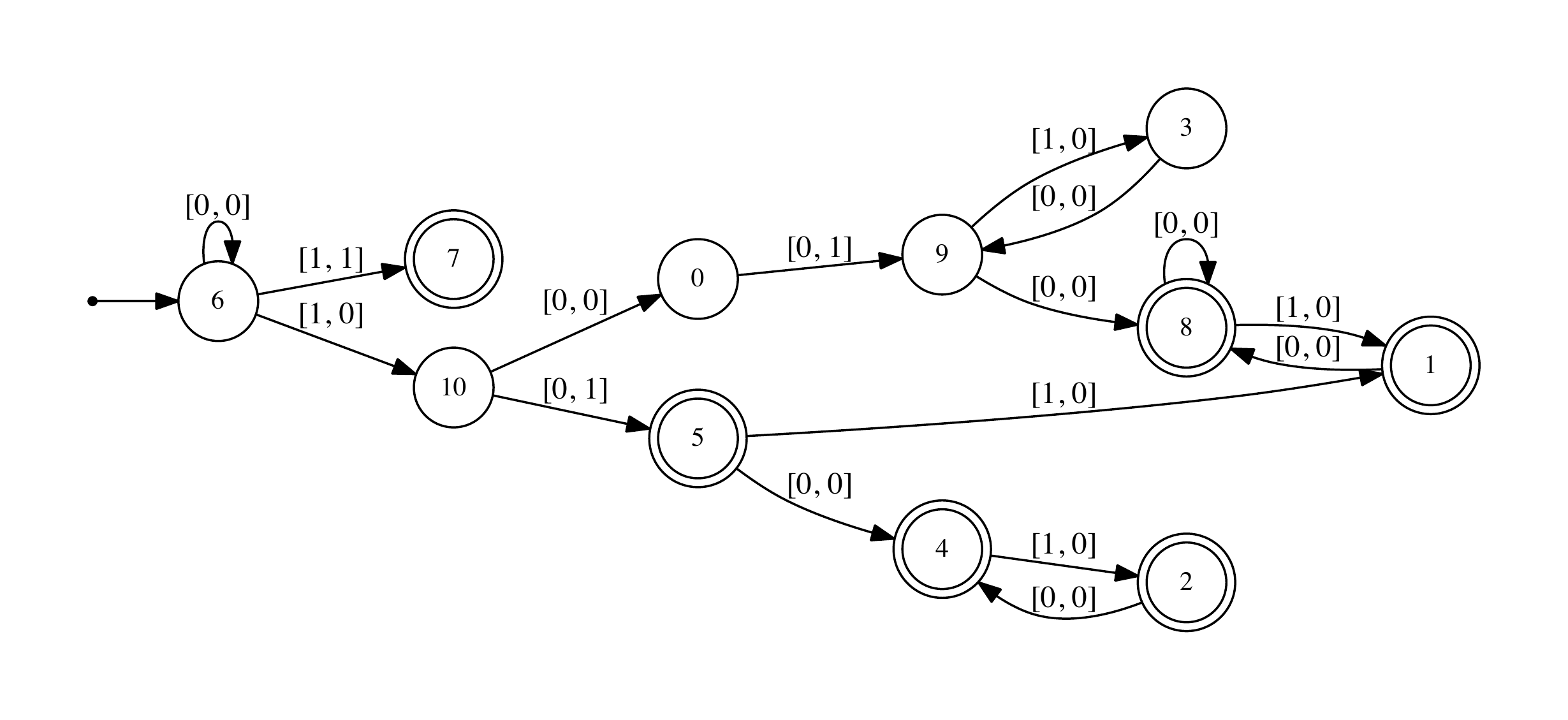}
\caption{Automaton encoding smallest period over all length-$n$ factors
in $\bf f$}
\label{least-period-over}
\end{center}
\end{figure}
The result now follows by inspection and the fact that
$(L_j-1)_F = 10 (01)^{(j-2)/2}$ if $j \geq 2$ is even,
and $100 (10)^{(j-3)/2}$ if $j \geq 3$ is odd.
\end{proof}

\subsection{Quasiperiods}

We now turn to quasiperiods.  An infinite word $\bf a$ is said to be
{\it quasiperiodic} if there is some finite nonempty word $x$ such that
${\bf a}$ can be completely ``covered'' with translates of $x$.  Here
we study the stronger version of quasiperiodicity where the first copy
of $x$ used must be aligned with the left edge of $\bf w$ and is not
allowed to ``hang over''; these are called {\it aligned covers} in
\cite{Christou&Crochemore&Iliopoulos:2012}.  More precisely, for us
${\bf a} = a_0 a_1 a_2 \cdots$ is quasiperiodic if there exists
$x$ such that for all $i \geq 0$ there exists $j\geq 0$ with $i-n < j \leq i$
such that $a_j a_{j+1} \cdots a_{j+n-1} = x$, where $n = |x|$.
Such an $x$ is called a {\it quasiperiod}.
Note that the condition $j \geq 0$ implies that, in this interpretation,
any quasiperiod must actually be a prefix of $\bf a$.

The quasiperiodicity of the Fibonacci word $\bf f$ was studied by
Christou, Crochemore, and Iliopoulos \cite{Christou&Crochemore&Iliopoulos:2012},
where we can (more or less) find the following theorem:

\begin{theorem}
A nonempty length-$n$ prefix of $\bf f$ is a quasiperiod of $\bf f$ if
and only if $n$ is not of the form $F_n - 1$ for $n \geq 3$.
\end{theorem}

In particular, the following prefix lengths are not quasiperiods:
$1$, $2$, $4$, $7$, $12$, and so forth.

\begin{proof}
We write a predicate for the assertion that the length-$n$ prefix is 
a quasiperiod:
$$\forall i \geq 0 \ \exists j \text{ with }  i-n < j \leq i
\text{ such that } \forall t<n \ {\bf f}[t] = {\bf f}[j+t] .$$
When we do this, we get the automaton in Figure~\ref{quasi} below.
Inspection shows that this DFA accepts all canonical representations,
except those of the form $1(01)^*(\epsilon + 0)$, which are precisely
the representations of $F_n - 1$.

\begin{figure}[H]
\begin{center}
\includegraphics[width=4in]{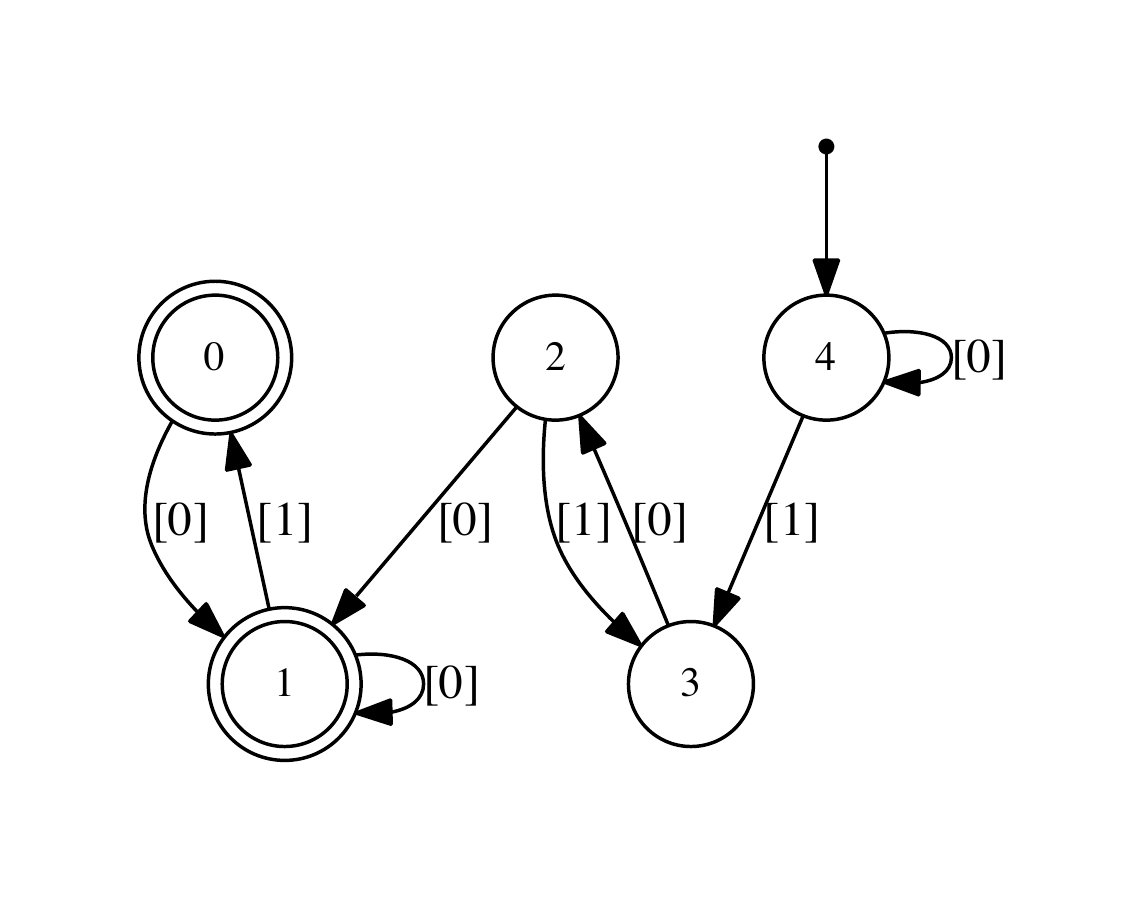}
\caption{Automaton accepting lengths of prefixes of
$\bf f$ that are quasiperiods}
\label{quasi}
\end{center}
\end{figure}
\end{proof}

\subsection{Unbordered factors}

Next we look at unbordered factors.  A word $y$ is said to be a {\it
border} of $x$ if $y$ is both a nonempty proper prefix and suffix of
$x$.  A word $x$ is {\it bordered} if it has at least one border.  It
is easy to see that if a word $y$ is bordered iff it has a border of
length $\ell$ with $0 < \ell \leq |y|/2$.

\begin{theorem}
The only unbordered nonempty factors of $\bf f$ are of length $F_n$ for
$n \geq 2$, and there are two for each such length.
For $n \geq 3$ these two unbordered factors have the property that
one is a reverse of the other.
\end{theorem}

\begin{proof}
We can express the property of having an unbordered factor of length $n$
as follows
$$ \exists i\ \forall j, 1 \leq j \leq n/2, \
	\exists t<j\ {\bf f}[i+t] \not= {\bf f}[i+n-j+t] .$$

Here is the log:
{\footnotesize
\begin{verbatim}
j >= 1 with 4 states, in 155ms
 2 * j <= n with 16 states, in 91ms
  j >= 1 & 2 * j <= n with 21 states, in 74ms
   t < j with 7 states, in 17ms
    F[i + t] != F[i + n - j + t] with 321 states, in 10590ms
     t < j & F[i + t] != F[i + n - j + t] with 411 states, in 116ms
      Et t < j & F[i + t] != F[i + n - j + t] with 85 states, in 232ms
       j >= 1 & 2 * j <= n => Et t < j & F[i + t] != F[i + n - j + t] with 137 states, in 19ms
        Aj j >= 1 & 2 * j <= n => Et t < j & F[i + t] != F[i + n - j + t] with 7 states, in 27ms
         Ei Aj j >= 1 & 2 * j <= n => Et t < j & F[i + t] != F[i + n - j + t] with 3 states, in 0ms
overall time: 11321ms
\end{verbatim}
}

The automaton produced accepts the Fibonacci representation of $0$ and
$F_n$ for $n \geq 2$.

Next, we make the assertion that there are exactly two such factors for
each appropriate length.
We can do this by saying there is an unbordered factor of length $n$
beginning at position $i$, another one beginning at position $k$, and
these factors are distinct, and for every unbordered factor of length $n$,
it is equal to one of these two.  When we do this we discover that
the representations of all $F_n$ for $n \geq 2$ are accepted.

Finally, we make the assertion that for any two unbordered factors of
length $n$, either they are equal or one is the reverse of the other.
When we do this we discover all lengths except length $1$ are accepted.
(That is, for all lengths other than $F_n$, $n \geq 2$, the assertion is
trivially true since there are no unbordered factors; for $F_2 = 1$ it
is false since $0$ and $1$ are the unbordered factors and one is not the
reverse of the other; and for all larger $F_i$ the property holds.)
\end{proof}

\subsection{Recurrence, uniform recurrence, and linear recurrence}

We now turn to various questions about recurrence.  A factor $x$ of an
infinite word $\bf w$ is said to be {\it recurrent} if it occurs  infinitely
often.  The word $\bf w$ is recurrent if every factor that occurs at
least once is recurrent.  A factor $x$ is {\it uniformly recurrent}
if there exists a constant $c = c(x)$ such that any factor ${\bf w}[i..i+c]$
is guaranteed to contain an occurrence of $x$.  If all factors are uniformly
recurrent then $\bf w$ is said to be uniformly recurrent.  Finally,
${\bf w}$ is {\it linearly recurrent} if the constant $c(x)$ is $O(|x|)$.

\begin{theorem}
The word {\bf f} is recurrent, uniformly recurrent, and linearly recurrent.
\end{theorem}

\begin{proof}
A predicate for all length-$n$ factors being recurrent:
$$ \forall i \geq 0\ \forall j \geq 0\ \exists k > j\ \forall t<n \ 
	{\bf f}[i+t] = {\bf f}[k+t] .$$
This predicate says that for every factor $z = {\bf f}[i..i+n-1]$ and
every position $j$ we can find another occurrence of 
$z$ beginning at a position $k > j$.  When we run this we discover that
the representations of all $n \geq 0$ are accepted.  So $\bf f$ is recurrent.


A predicate for uniform recurrence:
$$ \forall i\ \exists \ell\ \forall j \ \exists s, \ j \leq s \leq j+l-n \
\forall p<n \ {\bf f}[s+p] = {\bf f}[i+p] .$$
Once again, when we run this we discover that
the representations of all $n \geq 0$ are accepted.   So $\bf f$ is
uniformly recurrent.

A predicate for linear recurrence with constant $C$:
$$ \forall i\ \forall j \ \exists s, \ j \leq s \leq j+Cn \
\forall p<n \ {\bf f}[s+p] = {\bf f}[i+p] .$$
When we run this with $C = 4$, we discover that
the representations of all $n \geq 0$ are accepted (but, incidentally,
not for $C = 3$).
So $\bf f$ is linearly recurrent.
\end{proof}

\begin{remark}
We can decide the property of
linear recurrence for Fibonacci-automatic words
even without knowing an explicit value for the constant $C$.
The idea is to accept those pairs $(n,t)$ such that there exists a
factor of length $n$ with two consecutive occurrences separated by 
distance $t$.  Letting $S$ denote the set of such pairs, then
a sequence is linearly recurrent iff $\limsup_{(n,t)\in S} t/n < \infty$,
which can be decided using an argument like that in
\cite[Thm.~8]{Schaeffer&Shallit:2012}.  However, we do not
know how to compute, in general,
the exact value of the $\limsup$ for Fibonacci representation (which we do indeed
know for base-$k$ representation), although we can approximate it arbitrarily
closely.
\end{remark}

\subsection{Lyndon words}

Next, we turn to some results about Lyndon words.  Recall that a
nonempty word $x$ is a {\it Lyndon word\/} if it is lexicographically
less than all of its nonempty proper prefixes.\footnote{There is also a version
where ``prefixes'' is replaced by ``suffixes''.}  We reprove some recent
results of Currie and Saari \cite{Currie&Saari:2009} and Saari
\cite{Saari:2014}.

\begin{theorem} Every Lyndon factor of $\bf f$ is of length $F_n$ for
some $n \geq 2$, and each of these lengths has a Lyndon factor.
\end{theorem}

\begin{proof}
Here is the predicate 
specifying that there is a factor of length $n$ that is Lyndon:
$$
\exists i\ \forall j, 1 \leq j < n, \ 
\exists t < n-j \ (\forall u<t \ {\bf f}[i+u]={\bf f}[i+j+u]) \ \wedge \ 
{\bf f}[i+t] < {\bf f}[i+j+t] .$$
When we run this we get the representations $10^*$, which proves the result.
\end{proof}

\begin{theorem}
For $n \geq 2$,
every length-$n$ Lyndon factor of $\bf f$ is a conjugate of ${\bf f}[0..n-1]$.
\end{theorem}

\begin{proof}
Using the predicate from the previous theorem as a base, we can create a predicate
specifying that every length-$n$ Lyndon factor is a conjugate of ${\bf f}[0..n-1]$.
When we do this we discover that all lengths except $1$ are accepted.  (The only
lengths having a Lyndon factor are $F_n$ for $n \geq 2$, so all but $F_2$ have
the desired property.)
\end{proof}

\subsection{Critical exponents}

Recall from Section~\ref{proofsf} that $\exp(w) = |w|/P$, where $P$ is the
smallest period of $w$.  The {\it critical exponent} of an infinite
word $\bf x$ is the supremum, over all factors $w$ of $\bf x$, of
$\exp(w)$.  

A classic result of \cite{Mignosi&Pirillo:1992} is 

\begin{theorem}
The critical exponent of $\bf f$ is $2+ \alpha$, where
$\alpha = (1+\sqrt{5})/2$.
\end{theorem}

Although it is known that the critical exponent is computable
for $k$-automatic sequences \cite{Schaeffer&Shallit:2012}, we do not yet
know this for Fibonacci-automatic sequences (and more generally
Pisot-automatic sequences).  However, with a little inspired guessing
about the maximal repetitions, we can complete the proof.

\begin{proof}
For each length $n$, the smallest possible period $p$
of a factor is given by Theorem~\ref{allpers}.  Hence
the critical exponent is given by $\lim_{j \rightarrow \infty}
(L_{j+1}-2)/F_j$, which is $2+\alpha$.
\end{proof}

   We can also ask the same sort of questions about the
{\it initial critical exponent} of a word $\bf w$,
which is the supremum over the exponents of all  prefixes of $\bf w$.

\begin{theorem}
The initial critical exponent of $\bf f$ is $1+\alpha$.
\end{theorem}

\begin{proof}
We create an automaton $M_{\rm ice}$ accepting the language
$$L = \{ (n,p)_F \ : \ {\bf f}[0..n-1] \text{ has least period } p \} .$$
It is depicted in Figure~\ref{ice} below.
From the automaton, it is easy to see that the least period of 
the prefix of length $n \geq 1$ is $F_j$ for
$j \geq 2$ and $F_{j+1}-1 \leq n \leq F_{j+2} - 2$.  Hence the
initial critical exponent is given by
$\limsup_{j \rightarrow \infty} (F_{j+2} - 2)/F_j$, which is $1+\alpha$.

\begin{figure}[H]
\begin{center}
\includegraphics[width=6.5in]{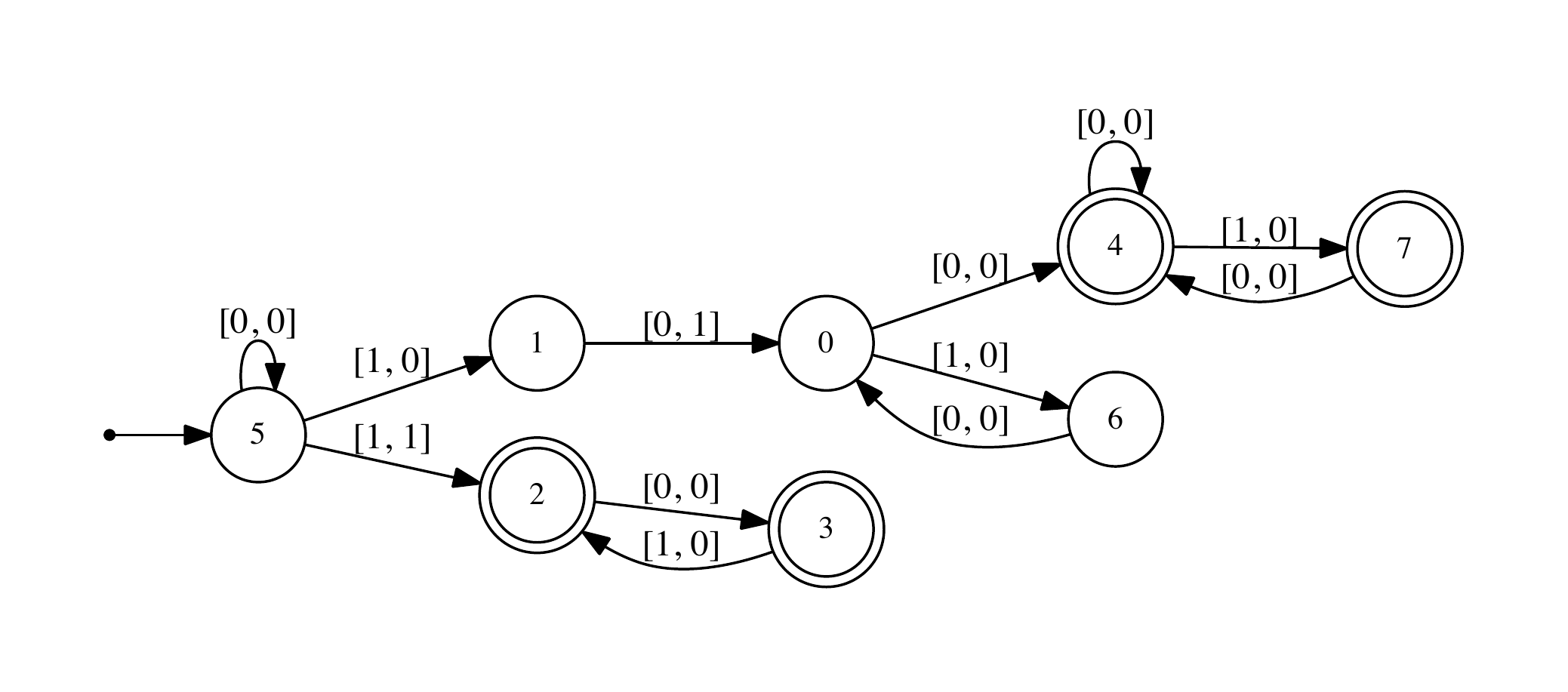}
\caption{Automaton accepting least periods of prefixes of length $n$}
\label{ice}
\end{center}
\end{figure}

\end{proof}

\subsection{The shift orbit closure}

The {\it shift orbit closure} of a sequence $\bf x$ is the set of all
sequences $\bf t$ with the property that each prefix of $\bf t$
appears as a factor of $\bf x$.  Note that this set can be much larger
than the set of all suffixes of $\bf x$.

The following theorem
is well known \cite[Prop.~3, p.~34]{Borel&Laubie:1993}:

\begin{theorem}
The lexicographically least sequence in the 
shift orbit closure of $\bf f$ is $0{\bf f}$, and the lexicographically
greatest is $1 {\bf f}$.
\end{theorem}

\begin{proof}
We handle only the lexicographically least, leaving the lexicographically
greatest to the reader.

The idea is to create a predicate $P(n)$ for the lexicographically least
sequence ${\bf b} = b_0 b_1 b_2 \cdots$ which is true iff 
$b_n = 1$.  The following predicate encodes, first, that $b_n = 1$, and
second, that if one chooses any length-($n+1$) factor $t$ of $\bf f$,
then $b_0 \cdots b_n$ is equal or lexicographically smaller than $t$.  

\begin{multline*}
\exists j \ {\bf f}[j+n]=1 \ \wedge \ 
\forall k \ (( \forall s \leq n \ {\bf f}[j+s] = {\bf f}[k+s] )  \ \vee \  \\
(\exists i\leq n\ {\text s. t. }\ {\bf f}[j+i] < {\bf f}[k+i] 
	\ \wedge \ ( \forall t<i \  {\bf f}[j+t]={\bf f}[k+t] )))
\end{multline*}
When we do this we get the following automaton, which is easily
seen to generate the sequence $0 {\bf f}$.

\begin{figure}[H]
\begin{center}
\includegraphics[width=6.5in]{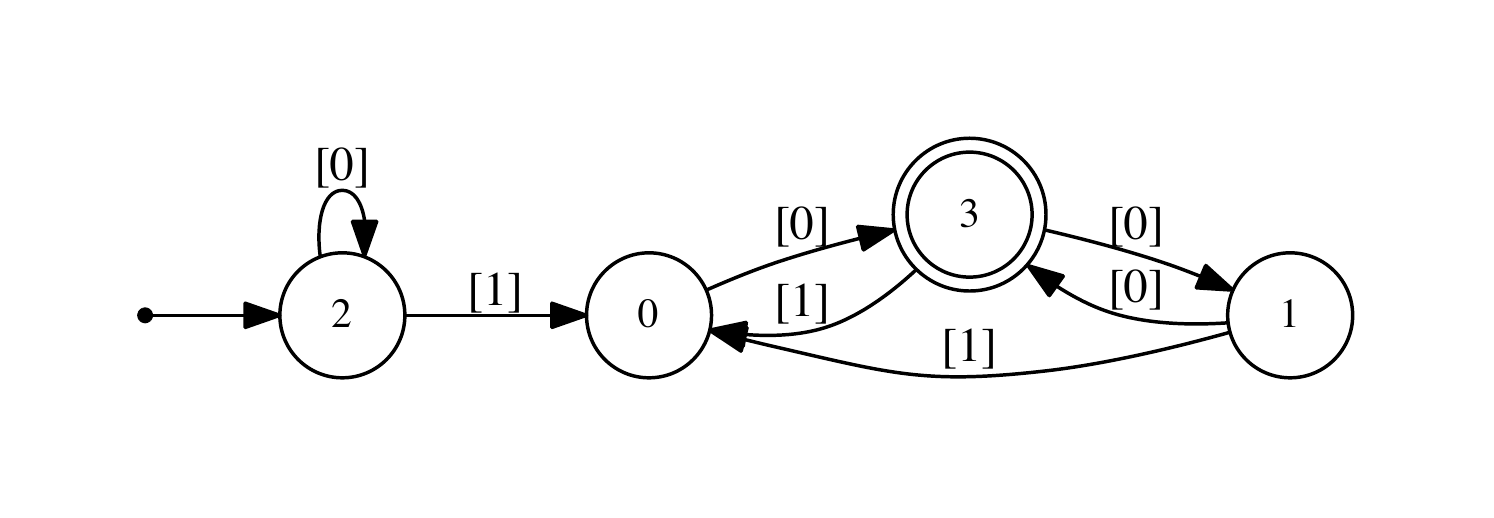}
\caption{Automaton accepting lexicographically least sequence in shift orbit
closure of ${\bf f}$}
\label{lexleastorbit}
\end{center}
\end{figure}

\end{proof}

\subsection{Minimal forbidden words}

Let ${\bf x}$ be an infinite word.
A finite word $z = a_0 \cdots a_n$ is said to be
{\it minimal forbidden} if $z$ is not a factor of $\bf x$, but
both $a_1 \cdots a_n$ and $a_0 \cdots a_{n-1}$ are 
\cite{Currie&Rampersad&Saari:2013}.

We can characterize all minimal forbidden words as follows:
we create an automaton accepting the language
\begin{multline*}
 \{ (i,n)_F \ : \  {\bf f}[i..i+n-1] \, \overline{{\bf f}[n]} 
\text{ is not a factor of $\bf f$ and } \\
{\bf f}[i+1..i+n-1] \, \overline{{\bf f}[n]} \text{ is a factor } 
\text{and } i \text{ is as small as possible } \}.
\end{multline*}

When we do so we find the words accepted are
$$ [1,1] ([0,0][1,1])^* (\epsilon + [0,0]) .$$
This corresponds to the words
$$ {\bf f}[F_n - 1..2F_n -3] \, \overline{{\bf f}[2F_n -2]} $$
for $n \geq 3$.
The first few are
$$ 11, 000, 10101, 00100100, 1010010100101, \ldots .$$

\subsection{Grouped factors}

Cassaigne \cite{Cassaigne:1998} introduced the notion of \textit{grouped factors}.  A
sequence ${\bf a} = (a_i)_{i \geq 0}$ has grouped factors if,
for all $n \geq 1$, there exists some position $m = m(n)$ such that
${\bf a}[m..m+\rho(n)+n-2]$ contains all the $\rho(n)$ length-$n$
blocks of $\bf a$, each block occurring exactly once.  One consequence of
his result is 
that the Fibonacci word has grouped factors.

We can write a predicate for the property of having grouped
factors, as follows:
\begin{multline*}
\forall n \geq 1 \quad \exists m, s \geq 0 \quad 
	\forall i \geq 0  \\
\exists j \text{ s.t. }  m \leq j \leq m+s \text{ and } {\bf a}[i..i+n-1] = {\bf a}[j..j+n-1]  \text{ and } \\
\forall j', \ m \leq j' \leq m+s, \quad j \not= j' 
\text{ we have } {\bf a}[i..i+n-1] \not= {\bf a}[j'..j'+n-1]  .
\end{multline*}


The first part of the predicate says that every length-$n$ block
appears somewhere in the desired window, and the second says
that it appears exactly once.   

(This five-quantifier definition
can be viewed as a response to the question of Homer and
Selman \cite{Homer&Selman:2011}, 
``...in what sense would a problem that required at least three 
alternating quantifiers to describe be natural?")


Using this predicate and our decision method,
we verified that the Fibonacci word does indeed
have grouped factors.

\section{Mechanical proofs of properties of the finite Fibonacci words}
\label{finitefib}

Although our program is designed to answer questions about the properties
of the infinite Fibonacci word $\bf f$, it can also be used to solve 
problems concerning the finite Fibonacci words $(X_n)$, defined as follows:
$$
X_n = \begin{cases}
	\epsilon, & \text{if $n = 0$}; \\
	1, & \text{if $n = 1$}; \\
	0, & \text{if $n = 2$}; \\
	X_{n-1} X_{n-2}, & \text{if $n > 2$}.
	\end{cases}
$$
Note that $|X_n| = F_n$ for $n \geq 1$.  (We caution the reader that
there exist many variations on this definition in the literature,
particularly with regard to indexing and initial values.)
Furthermore, we have $\varphi(X_n) = X_{n+1}$ for $n \geq 1$.  

Our strategy for the
the finite Fibonacci words has two parts:

\begin{itemize}

\item[(i)]  Instead of phrasing statements in terms of factors, we
phrase them in terms of occurrences of factors (and hence in terms of the
indices defining a factor).

\item[(ii)] Instead of phrasing statements about finite Fibonacci words, we
phrase them instead about {\it all\/} length-$n$ prefixes of $\bf f$.  Then,
since $X_i = {\bf f}[0..F_i - 1]$, we can
deduce results about the finite Fibonacci words by considering the case
where $n$ is a Fibonacci number $F_i$.

\end{itemize}

To illustrate this idea, consider
one of the most famous properties of the Fibonacci words, the
{\it almost-commutative} property:  letting $\eta(a_1 a_2 \cdots a_n) =
a_1 a_2 \cdots a_{n-2} a_n a_{n-1}$ be the map that interchanges
the last two letters of a string of length at least $2$,
we have

\begin{theorem}
$X_{n-1} X_n = \eta(X_n X_{n-1})$ for $n \geq 2$.
\end{theorem}

We can verify this, and prove even more, using our method.  

\begin{theorem}
Let $x = {\bf f}[0..i-1]$ and $y = {\bf f}[0..j-1]$ for $i > j > 1$.
Then $xy = \eta(yx)$ if and only if $i = F_n$, $j = F_{n-1}$ for $n \geq 3$.
\end{theorem}

\begin{proof}
The idea is to check, for each $i > j > 1$, whether
$${\bf f}[0..i-1] {\bf f}[0..j-1] = \eta({\bf f}[0..j-1] {\bf f}[0..i-1]).$$
We can do this with the following predicate:
\begin{multline*}
(i>j) \ \wedge \ (j\geq 2) \ \wedge \ (\forall t,\ j\leq t<i,\  {\bf f}[t]=
{\bf f}[t-j]) 
\ \wedge \\
(\forall s \leq j-3\ {\bf f}[s]={\bf f}[s+i-j]) \ \wedge \ ({\bf f}[j-2]={\bf f}[i-1]) 
\ \wedge \ ({\bf f}[j-1]={\bf f}[i-2]) .
\end{multline*}
The log of our program is as follows:
{\tiny
\begin{verbatim}
i > j with 7 states, in 49ms
 j >= 2 with 5 states, in 87ms
  i > j & j >= 2 with 12 states, in 3ms
   j <= t with 7 states, in 3ms
    t < i with 7 states, in 17ms
     j <= t & t < i with 19 states, in 6ms
      F[t] = F[t - j] with 16 states, in 31ms
       j <= t & t < i => F[t] = F[t - j] with 62 states, in 31ms
        At j <= t & t < i => F[t] = F[t - j] with 14 states, in 43ms
         i > j & j >= 2 & At j <= t & t < i => F[t] = F[t - j] with 12 states, in 9ms
          s <= j - 3 with 14 states, in 72ms
           F[s] = F[s + i - j] with 60 states, in 448ms
            s <= j - 3 => F[s] = F[s + i - j] with 119 states, in 14ms
             As s <= j - 3 => F[s] = F[s + i - j] with 17 states, in 58ms
              i > j & j >= 2 & At j <= t & t < i => F[t] = F[t - j] & As s <= j - 3 => F[s] = F[s + i - j] with 6 states, in 4ms
               F[j - 2] = F[i - 1] with 20 states, in 34ms
                i > j & j >= 2 & At j <= t & t < i => F[t] = F[t - j] & As s <= j - 3 => F[s] = F[s + i - j] & F[j - 2] = F[i - 1] with 5 states, in 1ms
                 F[j - 1] = F[i - 2] with 20 states, in 29ms
                  i > j & j >= 2 & At j <= t & t < i => F[t] = F[t - j] & As s <= j - 3 => F[s] = F[s + i - j] & F[j - 2] = F[i - 1] & F[j - 1] = F[i - 2] with 5 states, in 1ms
overall time: 940ms
\end{verbatim}
}

The resulting automaton accepts $[1,0][0,1][0,0]^+$, which corresponds
to $i = F_n$, $j = F_{n-1}$ for $n \geq 4$.
\end{proof}

An old result of S\'e\'ebold \cite{Seebold:1985b} is
\begin{theorem}
If $uu$ is a square occurring in $\bf f$, then $u$ is conjugate
to some finite Fibonacci word.
\end{theorem}

\begin{proof}
Assertion $\conj(i,j,k,\ell)$ means ${\bf f}[i..j]$ is a 
conjugate of ${\bf f}[k..\ell]$ (assuming $j-i = \ell-k$)
$$\conj(i,j,k,\ell) :=
\exists m \ {\bf f}[i..i+\ell-m] = {\bf f}[m..\ell] \text{ and }
	{\bf f}[i+\ell-m+1..j] = {\bf f}[k..m-1].$$

Predicate:  $$ ({\bf f}[i..i+n-1] = {\bf f}[i+n..i+2n-1])
	\implies \conj(i,i+n-1,0,n-1) $$

This asserts that any square $uu$ of order $n$ appearing in $\bf f$
is conjugate to ${\bf f}[0..n-1]$.  
When we implement this, we discover that all lengths are accepted.
This makes sense since the only lengths corresponding to squares
are $F_n$, and for all other lengths the base of the implication is
false.
\end{proof}

We now reprove an old result of de Luca \cite{deLuca:1981}.  Recall
that a primitive word is a non-power; that is, a word that cannot
be written in the form $x^n$ where
$n$ is an integer $\geq 2$.

\begin{theorem}
All finite Fibonacci words are primitive.
\end{theorem}

\begin{proof}
The factor ${\bf f}[i..j]$ is a power if and only if there exists $d$,
$0 < d < j-i+1 $, such that ${\bf f}[i..j-d] = {\bf f}[i+d..j]$ and
${\bf f}[j-d+1..j] = {\bf f}[i..i+d-1]$.   Letting $\pow(i,j)$ denote
this predicate, the predicate 
$$  \neg \pow(0,n-1) $$
expresses the claim that the length-$n$ prefix ${\bf f}[0..n-1]$ is
primitive.  When we implement this, we discover that the prefix of
every length is primitive, except those prefixes of length
$2 F_n$ for $n \geq 4$.
\end{proof}

A theorem of Chuan \cite[Thm.~3]{Chuan:1993b} states that the finite Fibonacci word
$X_n$, for $n \geq 5$,
is the product of two palindromes in exactly one way:  where the
first factor of length $F_{n-1} -2$ and the second of length
$F_{n-2} + 2$.  (Actually, Chuan claimed this was true for all Fibonacci words,
but, for example, for $010$ there are evidently two different
factorizations of the form $(\epsilon)(010)$ and $(010)\epsilon$.)
We can prove something more
general using our method, by generalizing:

\begin{theorem}
If the length-$n$ prefix ${\bf f}[0..n-1]$ of $\bf f$ is the product
of two (possibly empty) palindromes, then $(n)_F$ is accepted by the
automaton in Figure~\ref{pal2} below.
\begin{figure}[H]
\begin{center}
\includegraphics[width=6.5in]{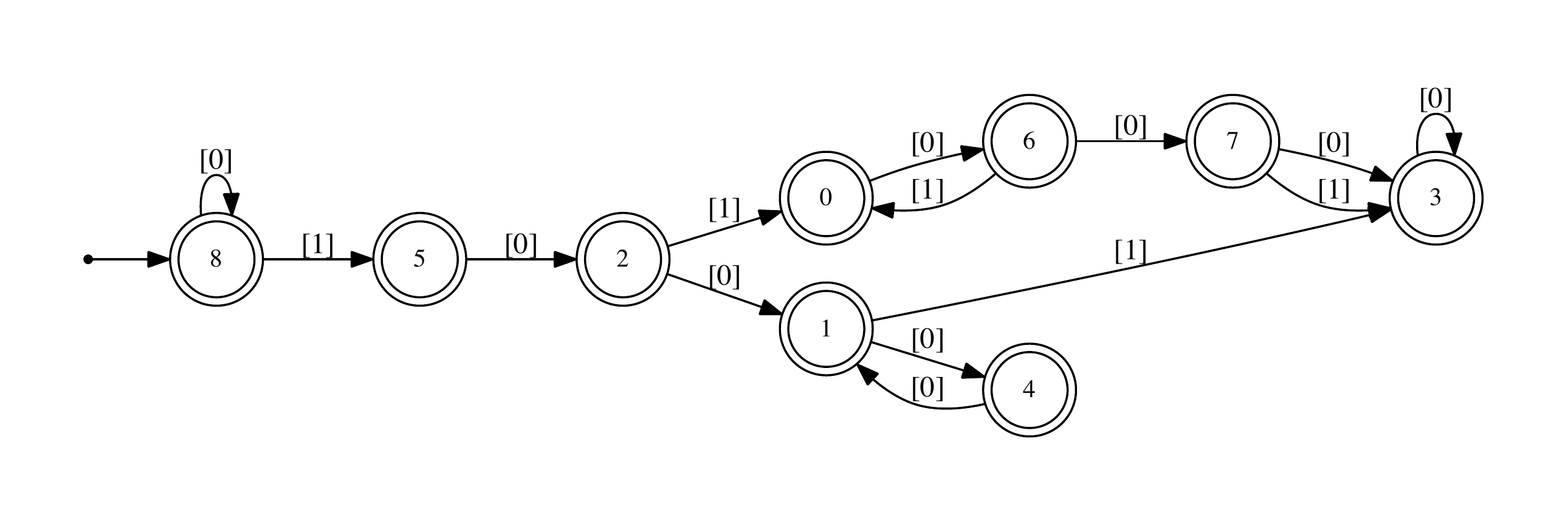}
\caption{Automaton accepting lengths of prefixes that are the product of
two palindromes}
\label{pal2}
\end{center}
\end{figure}
Furthermore, if the length-$n$ prefix ${\bf f}[0..n-1]$ of $\bf f$ is
the product of two (possibly empty) palindromes in exactly one way,
then $(n)_F$ is accepted by the automaton in Figure~\ref{pal2u} below.
\begin{figure}[H]
\begin{center}
\includegraphics[width=6.5in]{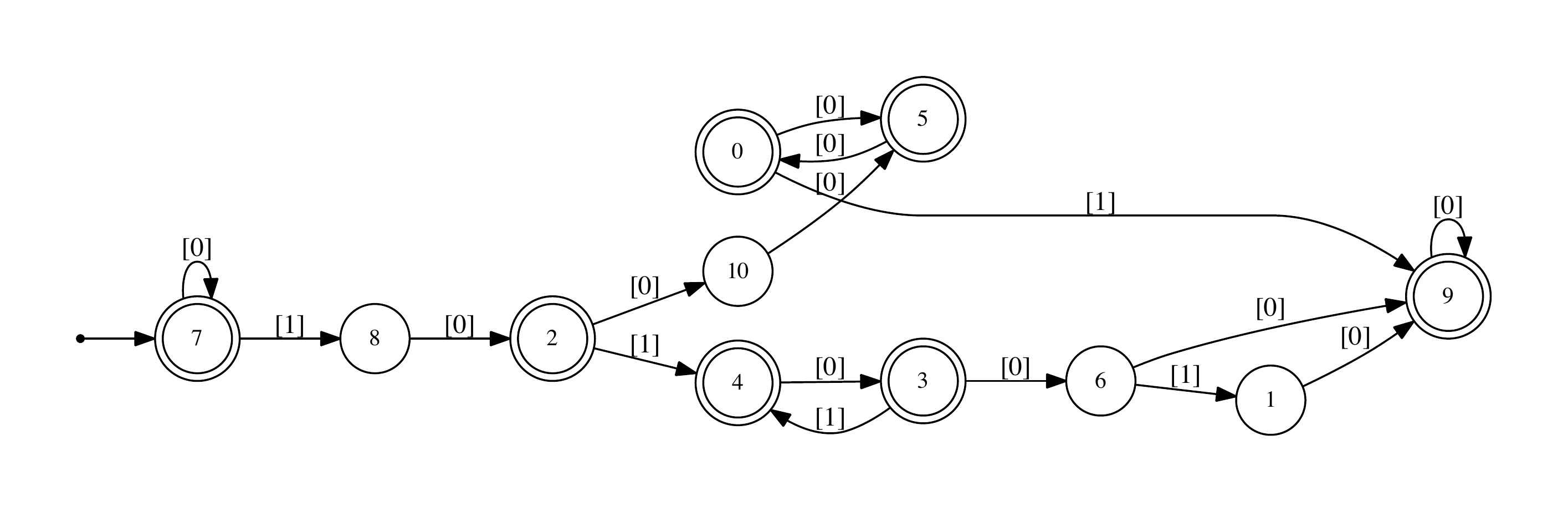}
\caption{Automaton accepting lengths of prefixes that are the product of
two palindromes in exactly one way}
\label{pal2u}
\end{center}
\end{figure}
Evidently, this includes all $n$ of the form $F_j$ for $j \geq 5$.
\end{theorem}

\begin{proof}
For the first, we use the predicate
$$ \exists p\leq n\ \left( (\forall t<p\ {\bf f}[t] = {\bf f}[p-1-t]) \
\wedge\ (\forall u< n-p\  {\bf f}[p+u] = {\bf f}[n-1-u]) \right) .$$

For the second, we use the predicate
\begin{multline*}
\exists p\leq n\ ( (\forall t<p\ {\bf f}[t] = {\bf f}[p-1-t]) \
\wedge\  
(\forall u< n-p\  {\bf f}[p+u] = {\bf f}[n-1-u])  )) \ \wedge \  
\\
(\forall q\leq n \ ( (\forall m<q\ {\bf f}[m] = {\bf f}[q-1-m]) \ \wedge\ 
(\forall v < n-q\ {\bf f}[q+v] = {\bf f}[n-1-v]) )  \implies p=q )   . 
\end{multline*}
\end{proof}

A result of  Cummings, Moore, and Karhum\"aki
\cite{Cummings&Moore&Karhumaki:1996} states that
the borders of the finite Fibonacci word
${\bf f}[0..F_n - 1]$ are precisely the words
${\bf f}[0..F_{n-2k} - 1]$ for $2k < n$.  We can prove this, and more:

\begin{proof}
Consider the pairs $(n,m)$ such that $1 \leq m < n$ and
${\bf f}[0..m-1]$ is a border of ${\bf f}[0..n-1]$.  Their
Fibonacci representations are accepted by the automaton below in
Figure~\ref{borders}.

\begin{figure}[H]
\begin{center}
\includegraphics[width=3in]{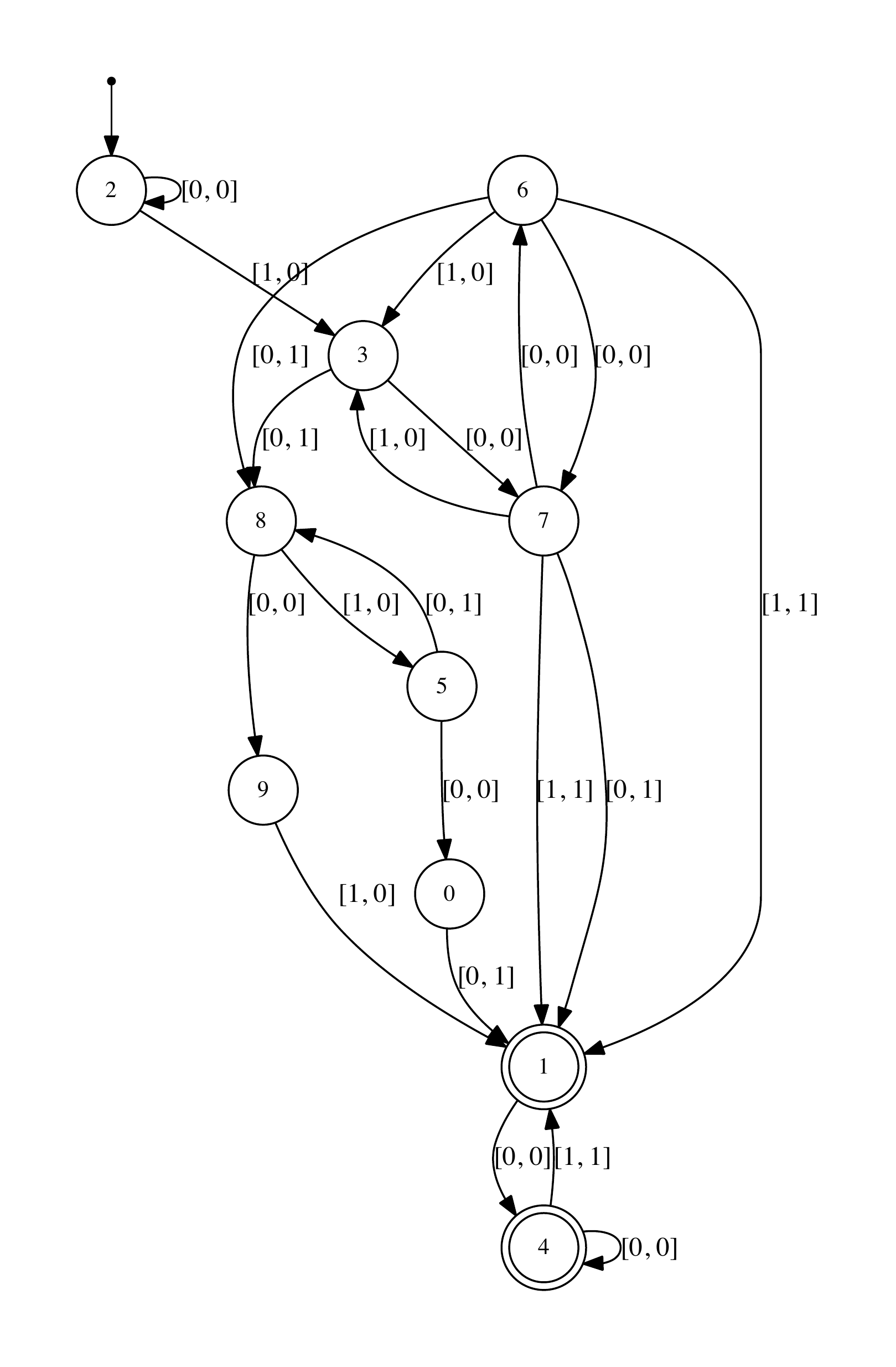}
\caption{Automaton encoding borders of prefixes of $\bf f$}
\label{borders}
\end{center}
\end{figure}

We use the predicate
$$ (n > m) \ \wedge \ (m \geq 1) \ \wedge \  \forall i<m \ {\bf f}[i] = {\bf f}[n-m+i] .$$
By following the paths with first coordinate of the form $10^+$ we recover
the result of Cummings, Moore, and Karhum\"aki as a special case.
\end{proof}

\section{Avoiding the pattern $x x x^R$ and the Rote-Fibonacci word}
\label{rotefib}

In this section we show how to apply our decision method to an
interesting and novel avoidance property:  avoiding the pattern
$x x x^R$ .   An example matching this pattern
in English is a factor of the word
{\tt bepepper}, with $x = {\tt ep}$.
Here, however, we are concerned only with
the binary alphabet $\Sigma_2 = \lbrace 0, 1 \rbrace$.

Although avoiding patterns with reversal has been considered
before (e.g.,
\cite{Rampersad&Shallit:2005,Bischoff&Nowotka:2011,Currie:2011,Bischoff&Currie&Nowotka:2012}),
it seems our particular problem has not been studied.

If our goal is just to produce some infinite word avoiding $x x x^R$, then
a solution seems easy:  namely, the infinite word
$(01)^\omega$ clearly avoids $x x x^R$, since if $|x| = n$ is odd,
then the second factor of length $n$ cannot equal the first
(since the first symbol differs), while if $|x| = n$ is even,
the first symbol of the third factor of length $n$ cannot be the
last symbol of $x$.  In a moment we will see that even this question
seems more subtle than it first appears, but for the moment, we'll
change our question to

\medskip

\centerline{\it Are there infinite aperiodic binary words 
avoiding $x x x^R$?}

\medskip

To answer this question, we'll study a special infinite word,
which we call the {\it Rote-Fibonacci word}.  (The name comes
from the fact that it is a special case of a class of words
discussed in 1994 by Rote \cite{Rote:1994}.)
Consider the following transducer $T$:


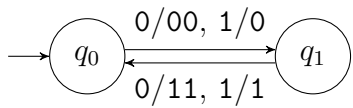
\begin{figure}[H]
\begin{center}
\begin{tikzpicture}[node distance=3cm,on grid,>=stealth',initial text=,auto,
					every state/.style={inner sep=1pt,minimum size=1cm},
					every loop/.style={shorten >=0,looseness=0}]

\node[state,initial]	(q_0)					{$q_0$};
\node[state] 			(q_1) [right=of q_0]	{$q_1$};

\path[->]	(q_0.10)	edge 				node {{\tt 0}/{\tt 00}, {\tt 1}/{\tt 0}} (q_1.170)
			(q_1.190)	edge 				node {{\tt 0}/{\tt 11}, {\tt 1}/{\tt 1}} (q_0.350);
\end{tikzpicture}
\end{center}
\caption{Transducer converting Fibonacci words to Rote-Fibonacci words} \label{fig:trans-f-rf}
\end{figure}

This transducer acts on words by following the transitions and
outputting the concatenation of the outputs associated with
each transition.  Thus, for example, the input
$01001$ gets transduced to the output $00100110$.

\begin{theorem}
The Rote-Fibonacci word 
$${\bf r} = 001001101101100100110110110010010011011001001001101100100100 \cdots = r_0 r_1 r_2 \cdots$$
has the following equivalent descriptions:

\bigskip

0.  As the output of the transducer $T$, starting in state $0$, on input
$\bf f$.

\bigskip

1.  As $\tau(h^\omega(a))$  where 
$h$ and $\tau$ are defined by
\begin{align*}
h(a) &= a b_1  &\quad \tau(a) = 0 \\
h(b) &= a & \quad \tau(b) = 1 \\
h(a_0) &= a_2 b & \quad \tau(a_0) = 0 \\
h(a_1) &= a_0 b_0 & \quad \tau(a_1) = 1 \\
h(a_2) &= a_1 b_2 & \quad \tau(a_2) = 1 \\
h(b_0) &= a_0 & \quad \tau(b_0) = 0  \\
h(b_1) &= a_1 & \quad \tau(b_1) = 0 \\
h(b_2) &= a_2 & \quad \tau(b_2) = 1 \\
\end{align*}

2.  As the binary sequence generated by the following DFAO, with 
outputs given in the states, and inputs in the Fibonacci representation
of $n$.


\begin{figure}[H]
\begin{center}
\begin{tikzpicture}[node distance=2cm,on grid,>=stealth',initial text=,auto,
					every state/.style={inner sep=1pt,minimum size=1cm},
					every loop/.style={shorten >=0,looseness=0}]

\node[state,initial]	(a)						{$a/{\tt 0}$};
\node[state] 			(b_1) [right=of a]		{$b_1/{\tt 0}$};
\node[state] 			(a_1) [right=of b_1]	{$a_1/{\tt 1}$};
\node[state] 			(b_0) [right=of a_1]	{$b_0/{\tt 0}$};
\node[state] 			(b)   [right=of b_0]	{$b/{\tt 1}$};
\node[state] 			(a_0) [right=of b]		{$a_0/{\tt 0}$};
\node[state] 			(a_2) [right=of a_0]	{$a_2/{\tt 1}$};
\node[state] 			(b_2) [right=of a_2]	{$b_2/{\tt 1}$};

\path[->]	(a)			edge [loop above]	node				{\tt 0} ()
						edge				node				{\tt 1} (b_1)
			(b_1)		edge				node				{\tt 0} (a_1)
			(a_1)		edge				node				{\tt 1} (b_0)
			(a_1.315)	edge [bend right=20]node [pos=0.1,swap] {\tt 0} (a_0.225)
			(b_0.45)	edge [bend left=45]	node [pos=0.5,swap] {\tt 0} (a_0.135)
			(b.210)		edge [bend left=24]	node [pos=0.1,swap] {\tt 0} (a_1.330)
			(a_0)		edge				node [swap]			{\tt 1} (b)
			(a_0)		edge				node [swap]			{\tt 0} (a_2)
			(a_2.135)	edge [bend right=24]node [pos=0.1,swap] {\tt 0} (a_1.45)
			(a_2.10)	edge 				node				{\tt 1} (b_2.170)
			(b_2.190)	edge 				node				{\tt 0} (a_2.350);
\end{tikzpicture}
\end{center}
\caption{Canonical Fibonacci representation DFAO generating the Rote-Fibonacci word} \label{fig:rf-dfao}
\end{figure}
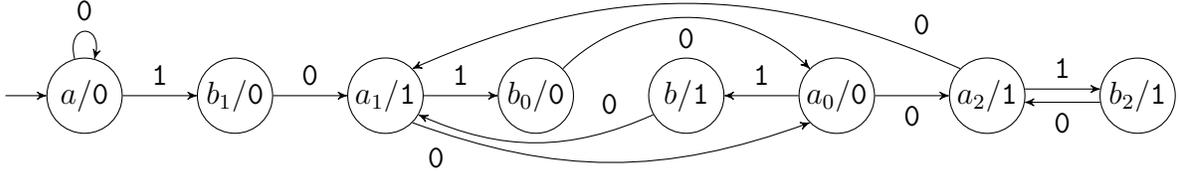

3.  As the limit, as $n \rightarrow \infty$, of the sequence of finite
Rote-Fibonacci words $(R_n)_n$ defined as follows:
$R_0 = 0$, $R_1 = 00$, and for $n \geq 3$
$$R_n = 
\begin{cases}
R_{n-1} R_{n-2}, & \text{ if $n \equiv 0$ (mod 3);} \\
R_{n-1} \overline{R_{n-2}}, & \text{ if $n \equiv 1, 2$ (mod 3).} 
\end{cases}$$

4.  As the sequence obtained from the Fibonacci sequence
${\bf f} = f_0 f_1 f_2 \cdots = 0100101001001 \cdots$
as follows:  first, change every $0$ to $1$ and every one to $0$ in
${\bf f}$, obtaining $\overline{\bf f} = 1011010110110 \cdots$.
Next, in $\overline{\bf f}$ change every second $1$ that appears to
$-1$ (which we write as $\mone$ for clarity):
$1 0 \mone 1 0 \mone 0 1 \mone 0 1 \mone 0 \cdots$.
Now take the running sum of this sequence, obtaining
$1101100100100 \cdots$, and finally, complement it to get $\bf r$.

\bigskip

5.  As $\rho(g^\omega (a))$, where $g$ and $\rho$ are defined as follows
\begin{align*}
g(a) &= abcab  \quad & \rho(a) = 0 \\
g(b) &= cda \quad & \rho(b) = 0 \\
g(c) &= cdacd \quad & \rho(c) = 1 \\
g(d) &= abc \quad &  \rho(d) = 1
\end{align*}
\end{theorem}

\begin{proof}

$(0) \iff (3)$:  Let $T_0 (x)$ (resp., $T_1 (x)$) denote the output
of the transducer $T$ starting in state $q_0$ (resp., $q_1$) on input 
$x$.  Then a simple induction on $n$ shows that
$T_0 (X_{n+1}) = R_n$ and $T_1(X_{n+1}) = \overline{R_n}$.
We give only the induction step for the first claim:
\begin{align*}
T_0 (X_{n+1}) &= T_0 (X_n X_{n-1}) \\
&= \begin{cases} 
T_0 (X_n) T_0 (X_{n-1}), & \text{if $|X_n|$ is even}; \\
T_0 (X_n) T_1 (X_{n-1}), & \text{if $|X_n|$ is odd};
\end{cases} \\
&= \begin{cases}
R_{n-1} R_{n-2}, & \text{if $n \equiv 0$ (mod 3)}; \\
R_{n-1} \overline{R_{n-2}}, & \text{if $n \not\equiv 0$ (mod 3)}; 
\end{cases} \\
&= R_n .
\end{align*}
Here we have used the easily-verified
fact that $|X_n|= F_n$ is even iff
$n \equiv 0$ (mod $3$).

\bigskip
$(1) \iff (3)$:  we verify by a tedious induction on
$n$ that for $n \geq 0$ we have
\begin{align*}
\tau(h^n(a)) &= \tau(h^{n+1} (a)) = R_n \\
\tau(h^n(a_i)) &= \tau(h^{n+1} (b_i)) = \begin{cases}
	R_i, & \text{if $n \equiv i$ (mod 3)}; \\
	\overline{R_i}, & \text{if $n \not\equiv i$ (mod 3)}.
	\end{cases}
\end{align*}

\bigskip
$(2) \iff (4)$:  Follows from the well-known transformation from
automata to morphisms and vice versa (see, e.g.,
\cite{Holton&Zamboni:2001}).

\bigskip
$(3) \iff (4)$:  We define some transformations on
sequences, as follows: 
\begin{itemize}
\item $C(x)$ denotes $\overline{x}$, the complement of $x$;
\item $s(x)$ denotes the sequence arising from a binary
sequence $x$ by changing every second $1$ to $-1$;
\item $a(x)$ denotes the running sum of the sequence $x$;
that is, if $x = a_1 a_2 a_3 \cdots $ then $a(x)$ is
$a_1 (a_1 + a_2) (a_1 + a_2 +a_3) \cdots$.
\end{itemize}
Note that 
$$
a (s (xy)) = 
\begin{cases}
a(s(x)) \ a(s(y)), & \text{if $|x|_1$ even}; \\
a(s(x)) \ C(a(s(y))), & \text{if $|x|_1$ odd}.
\end{cases}
$$
Then we claim that $ C(R_n) = a(s(C(X_{n+2})))$.   This
can be verified by induction on $n$.  We give only the
induction step:
\begin{align*}
a(s(C(X_{n+2}))) &= a(s( C(X_{n+1}) C(X_{n}) )) \\
&= \begin{cases}
a(s(C(X_{n+1}))) \ a(s(C(X_{n}))), & \text{ if $C(X_{n+1})_1$ even}; \\
a(s(C(X_{n+1}))) \ C(a(s(C(X_{n})))), & \text{ if $C(X_{n+1})_1$ odd}; 
\end{cases} \\
&= \begin{cases}
C(R_{n-1}) \ C(R_{n-2}), & \text{ if $n \equiv 0$ (mod 3)}; \\
C(R_{n-1}) \ R_{n-2}, & \text{ if $n \not\equiv 0$ (mod 3)};
\end{cases} \\
&= R_{n}.
\end{align*}

\bigskip
$(3) \iff (5)$:   Define $\gamma$ by
\begin{align*}
\gamma(a) &= \gamma(a_0) = a \\
\gamma(b_0) &= \gamma(b_1) = b \\
\gamma(a_1) &= \gamma(a_2) = c \\
\gamma(b) &= \gamma(b_2) = d .
\end{align*}
We verify by a tedious induction on $n$
that for $n \geq 0$ we have
\begin{align*}
g^n (a) &= \gamma(h^{3n} (a)) = \gamma(h^{3n} (a_0)) \\
g^n (b) &= \gamma(h^{3n} (b_0)) = \gamma(h^{3n} (b_1)) \\
g^n (c) &= \gamma(h^{3n} (a_1)) = \gamma(h^{3n} (a_2)) \\
g^n (d) &= \gamma(h^{3n} (b)) = \gamma(h^{3n} (b_2)) .
\end{align*}
\end{proof}

\begin{corollary}
The first differences $\Delta {\bf r}$ of the Rote-Fibonacci word $\bf r$, taken modulo
$2$, give the complement of the Fibonacci word $\overline{f}$,
with its first symbol omitted.
\label{rotecor}
\end{corollary}

\begin{proof}
Note that if ${\bf x} = a_0 a_1 a_2 \cdots$ is a binary sequence, then
$\Delta(C({\bf x})) = -\Delta({\bf x})$.
Furthermore $\Delta(a(x)) = a_1 a_2 \cdots$.
Now from the description in part 4, above, we know that
${\bf r} = C(a(s(C({\bf f}))))$.
Hence $\Delta({\bf r}) = \Delta ( C(a(s(C({\bf f}))))) =
-\Delta( a(s(C({\bf f})))) = \dr(-s(C({\bf f})))$,
where $\dr$ drops the first symbol of its argument.  Taking the last
result modulo $2$ gives the result.
\end{proof}

We are now ready to prove our avoidability result.

\begin{theorem}
The Rote-Fibonacci word $\bf r$ avoids the pattern $x x x^R$.
\label{rote-avoid-thm}
\end{theorem}

\begin{proof}
We use our decision procedure to prove this.  A predicate is as
follows:
$$ \exists i\ \forall t<n\ ({\bf r}[i+t]={\bf r}[i+t+n]) \ \wedge \  ({\bf r}[i+t]={\bf r}[i+3n-1-t])  .$$
When we run this on our program, we get the following log:

{\footnotesize
\begin{verbatim}
t < n with 7 states, in 36ms
 R[i + t] = R[i + t + n] with 245 states, in 1744ms
  R[i + t] = R[i + 3 * n - 1 - t] with 1751 states, in 14461ms
   R[i + t] = R[i + t + n] & R[i + t] = R[i + 3 * n - 1 - t] with 3305 states, in 565ms
    t < n => R[i + t] = R[i + t + n] & R[i + t] = R[i + 3 * n - 1 - t] with 2015 states, in 843ms
     At t < n => R[i + t] = R[i + t + n] & R[i + t] = R[i + 3 * n - 1 - t] with 3 states, in 747ms
      Ei At t < n => R[i + t] = R[i + t + n] & R[i + t] = R[i + 3 * n - 1 - t] with 2 states, in 0ms
overall time: 18396ms
\end{verbatim}
}

Then the only length $n$ accepted is $n = 0$,
so the Rote-Fibonacci word
$\bf r$ contains no occurrences of the pattern $x x x^R$.
\end{proof}

We now prove some interesting properties of $\bf r$.

\begin{theorem}
The minimum $q(n)$ over all periods of all length-$n$ factors
of the Rote-Fibonacci word is as follows:
$$ q(n) = 
\begin{cases}
1, & \text{if $1 \leq n \leq 2$;} \\
2, & \text{if $n = 3$;} \\
F_{3j+1}, & \text{if $j \geq 1$ and $L_{3j} \leq n < L_{3j+2}$;} \\
L_{3j+1}, & \text{if $j \geq 1$ and $L_{3j+2} \leq n < L_{3j+2}+F_{3j-2}$;} \\
F_{3j+2}+L_{3j}, & \text{if $j \geq 2$ and $L_{3j+2} + F_{3j-2} \leq n
	< L_{3j+2}+ F_{3j-1}$;} \\
2F_{3j+2}, & \text{if $L_{3j+2}+F_{3j-1} \leq n < L_{3j+3}$} .
\end{cases}
$$
\label{rfperiods-thm}
\end{theorem}

\begin{proof}
To prove this, we mimic the proof of Theorem~\ref{allpers}.  The
resulting automaton is displayed below in Figure~\ref{leastper-rote}.
\begin{figure}[H]
\begin{center}
\includegraphics[width=4in]{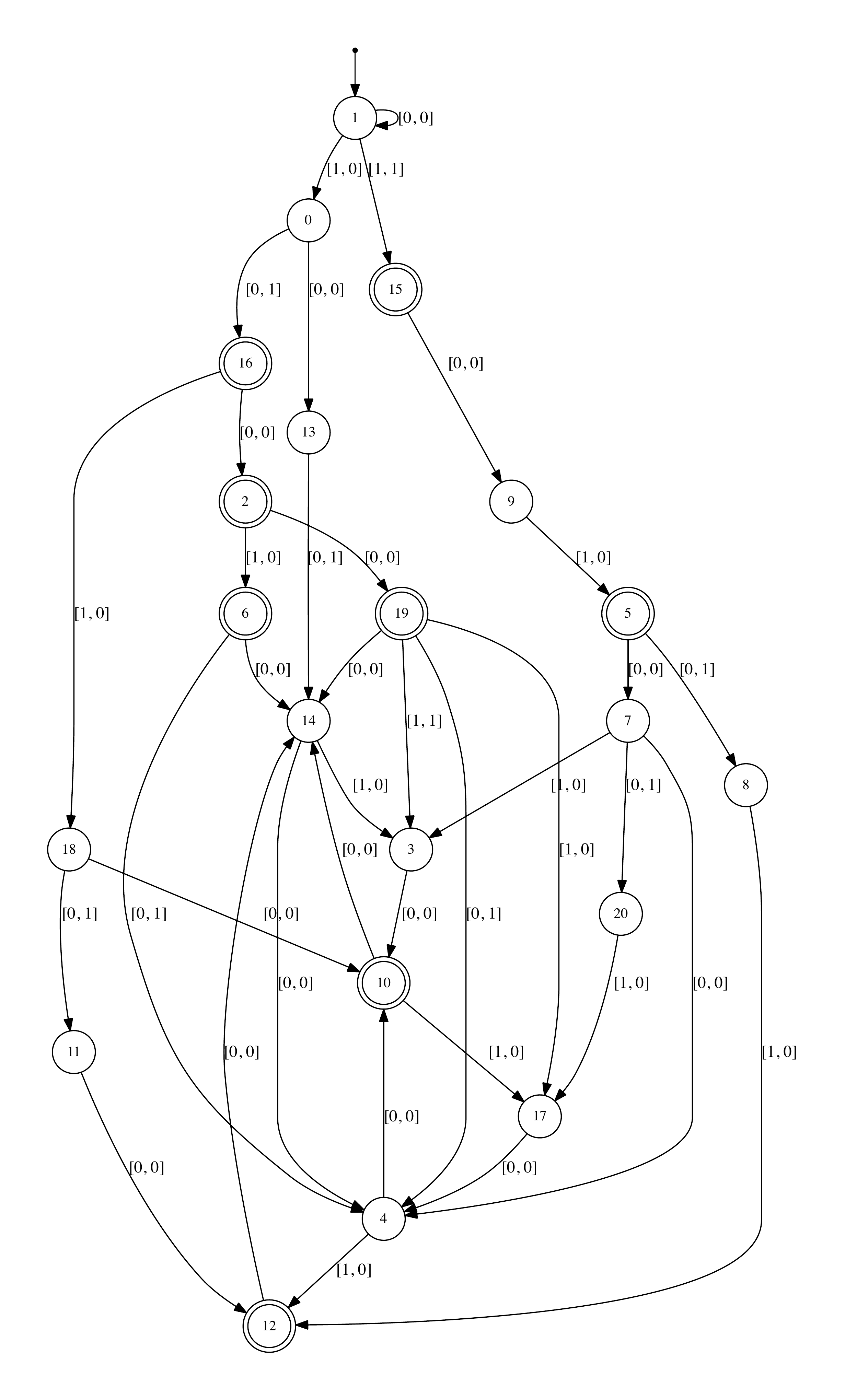}
\caption{Automaton accepting least periods of prefixes of length $n$}
\label{leastper-rote}
\end{center}
\end{figure}
\end{proof}

\begin{corollary}
The critical exponent of the Rote-Fibonacci word is $2+\alpha$.
\label{critical-rote}
\end{corollary}

\begin{proof}
An examination of the cases in Theorem~\ref{rfperiods-thm}
show that the words of maximum exponent
are those corresponding to $n = L_{3j+2}-1$, $p = F_{3j+1}$.  As
$j \rightarrow \infty$, the quantity $n/p$ approaches
$2 + \alpha$ from below.
\end{proof}

\begin{theorem}
All squares in the Rote-Fibonacci word are 
of order $F_{3n+1}$ for $n \geq 0$, and each such order occurs.
\label{rote3n}
\end{theorem}

\begin{proof}
We use the predicate
$$ (n \geq 1) \ \wedge \ \exists i \ \forall j<n\ ({\bf r}[i+j] = 
{\bf r}[i+j+n]) .$$
The resulting automaton is depicted in Figure~\ref{rotesquares}.
The accepted words correspond to $F_{3n+1}$ for $n \geq 0$.
\begin{figure}[H]
\begin{center}
\includegraphics[width=6.5in]{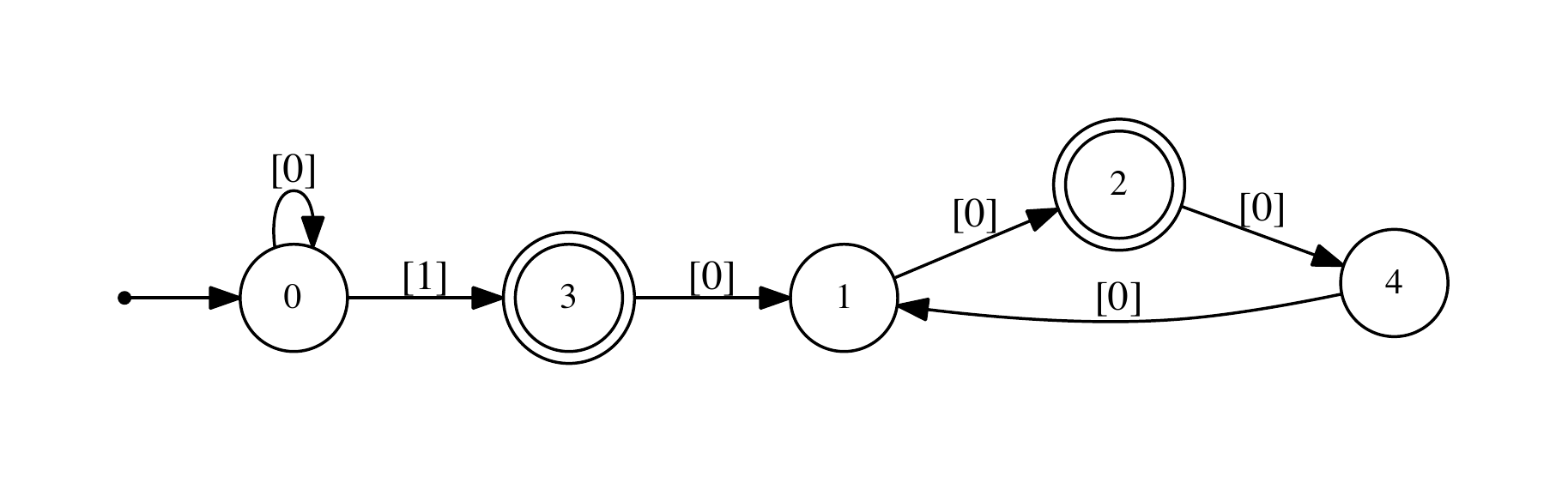}
\caption{Automaton accepting orders of squares in the Rote-Fibonacci word}
\label{rotesquares}
\end{center}
\end{figure}
\end{proof}

We now turn to problems considering prefixes of the Rote-Fibonacci
word $\bf r$.

\begin{theorem}
A length-$n$ prefix of the Rote-Fibonacci word $\bf r$ is an antipalindrome
iff $n = F_{3i+1} - 3$ for some $i \geq 1$.
\end{theorem}

\begin{proof}
We use our decision method on the predicate
$$ \forall j<n\ {\bf r}[j] \not= {\bf r}[n-1-j] .$$
The result is depicted in Figure~\ref{rote-antipal}.
The only accepted expansions are given by the regular expression
$\epsilon + 1(010101)^* 0 (010+101000)$, which corresponds
to $F_{3j+1} - 3$.
We use the predicate
$$ (n \geq 1) \ \wedge \ \exists i \ \forall j<n \ {\bf r}[i+j] = 
{\bf r}[i+j+n]) .$$
The resulting automaton is depicted in Figure~\ref{rote-antipal}.
The accepted words correspond to $F_{3n+1}$ for $n \geq 0$.

\begin{figure}[H]
\begin{center}
\includegraphics[width=6.5in]{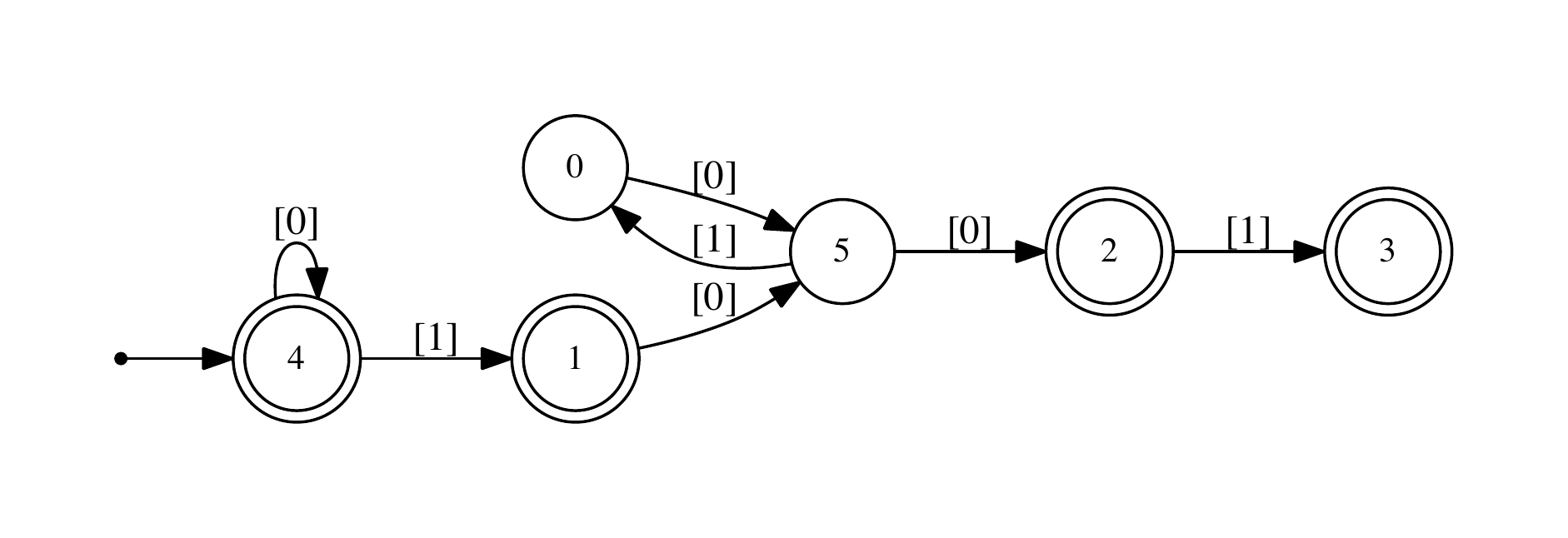}
\caption{Automaton accepting lengths of antipalindrome prefixes in the Rote-Fibonacci word}
\label{rote-antipal}
\end{center}
\end{figure}
\end{proof}

\begin{theorem}
A length-$n$ prefix of the Rote-Fibonacci word is an antisquare
if and only if $n = 2F_{3k+2}$ for some $k \geq 1$.
\end{theorem}

\begin{proof}
The predicate for having an antisquare prefix of length $n$ is
$$ \forall k < n \ {\bf r}[i+k] \not= {\bf r}[i+k+n] .$$
When we run this we get the automaton depicted in Figure~\ref{rote-antisquare-prefix}.

\begin{figure}[H]
\begin{center}
\includegraphics[width=5.5in]{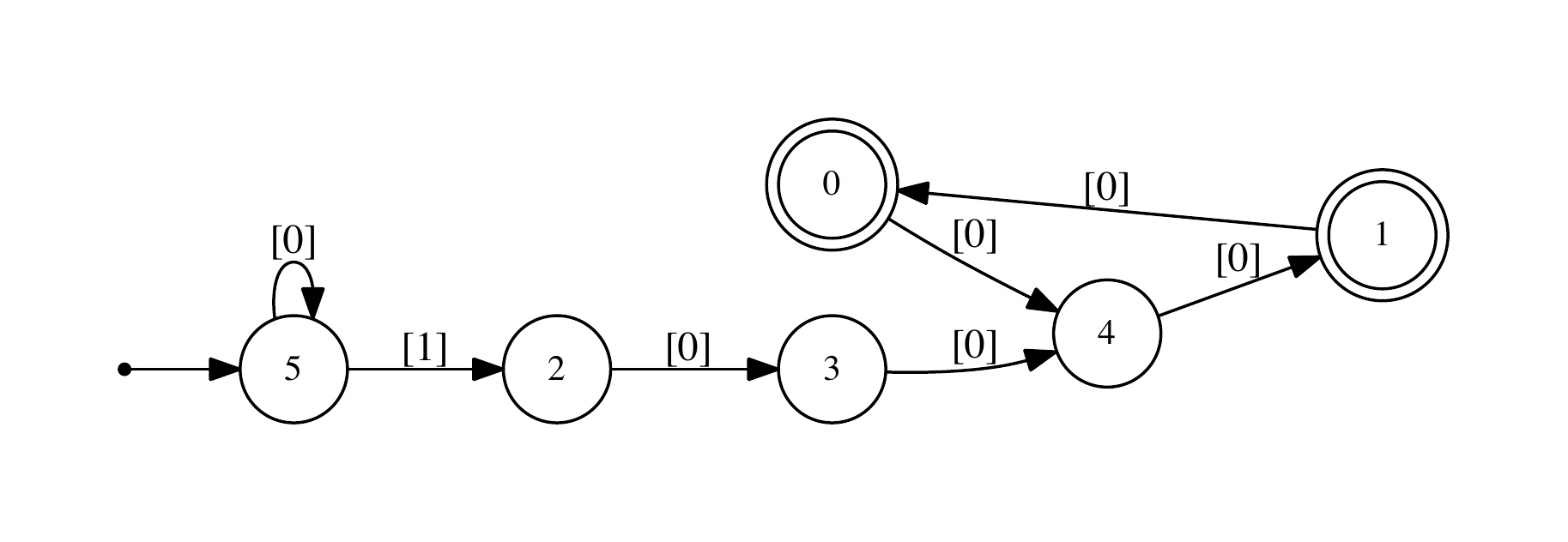}
\caption{Automaton accepting orders of antisquares that are prefixes of $\bf f$}
\label{rote-antisquare-prefix}
\end{center}
\end{figure}

\end{proof}

\begin{theorem}
The Rote-Fibonacci word has subword complexity $2n$.
\end{theorem}

\begin{proof}
Follows from Corollary~\ref{rotecor} together with
\cite[Thm.~3]{Rote:1994}.
\end{proof}

\begin{theorem}
The Rote-Fibonacci word is mirror invariant.  That is, if
$z$ is a factor of $\bf r$ then so is $z^R$.
\label{rotemi}
\end{theorem}

\begin{proof}
We use the predicate
$$ \forall i \ \exists j \ \forall t < n \ 
	{\bf r}[i+t] = {\bf r}[j+n-1-t] .$$
The resulting automaton accepts all $n$, so the conclusion follows.
The largest intermediate automaton has
2300 states and the calculation took about 6 seconds on a laptop.
\end{proof}

\begin{corollary}
The Rote-Fibonacci word avoids the pattern $x x^R x^R$.
\end{corollary}

\begin{proof}
Suppose $x x^R x^R$ occurs in $\bf r$.  Then by Theorem~\ref{rotemi}
we know that $(x x^R x^R)^R = x x x^R$ occurs in $\bf f$.  But this
is impossible, by Theorem~\ref{rote-avoid-thm}.
\end{proof}

As it turns out, the Rote-Fibonacci word has (essentially) appeared
before in several places.    For example, in a 2009 preprint of
Monnerot-Dumaine \cite{Monnerot-Dumaine:2009}, the author studies a plane
fractal called the ``Fibonacci word fractal'', specified by
certain drawing instructions, which can be coded over the alphabet
$S, R, L$ by taking the fixed point $g^\omega (a)$ and applying the
coding $\gamma(a) = S$, $\gamma(b) = R$, $\gamma(c) = S$, and
$\gamma(d) = L$.  
Here $S$ means ``move
straight one unit'', ``$R$'' means ``right turn one unit''
and ``$L$'' means
``left turn one unit''.

More recently,
Blondin Mass\'e, Brlek, Labb\'e, and Mend\`es France
studied a remarkable sequence of words closely related to $\bf r$
\cite{BlondinMasse&Brlek&Garon&Labbe:2011,BlondinMasse&Brlek&Labbe&MendesFrance:2011,BlondinMasse&Brlek&Labbe&MendesFrance:2012}.
For example, in their paper ``Fibonacci snowflakes''
\cite{BlondinMasse&Brlek&Garon&Labbe:2011} they defined a certain sequence
$q_i$ which has the following relationship to $g$:
let $\xi(a) = \xi(b) = L$, $\xi(c) = \xi(d) = R$.  Then
$$ R \xi(g^n(a)) = q_{3n+2} L .$$

\subsection{Conjectures and open problems about the Rote-Fibonacci word}

In this section we collect some conjectures we have not yet
been able to prove.  We have made some progress and hope to
completely resolve them in the future.

\begin{conjecture}
Every infinite binary word avoiding the pattern $x x x^R$ has critical
exponent $\geq 2+\alpha$.
\end{conjecture}

\begin{conjecture}
Let $z$ be a finite nonempty primitive binary word.  If $z^\omega$ avoids
$x x x^R$, then $|z| = 2 F_{3n+2}$ for some integer $n \geq 0$.  
Furthermore, $z$ is a conjugate of the 
prefix ${\bf r}[0..2F_{3n+2} - 1]$, for some $n \geq 0$.  Furthermore,
for $n \geq 1$ we have that $z$ is a conjugate of $y \overline{y}$,
where $y = \tau(h^{3n} (a))$.  
\end{conjecture}

We can make some partial progress on this conjecture, as follows:

\begin{theorem}
Let $k \geq 1$ and define $n = 2F_{3k+2}$.  Let
$z = {\bf r}[0..n-1]$.  Then $z^\omega$ contains no occurrence of
the pattern $x x x^R$.
\end{theorem}

\begin{proof}
We have already seen this for $k = 0$, so assume $k \geq 1$.

Suppose that $z^\omega$ does indeed contain an occurrence of $x x x^R$
for some $|x| = \ell > 0$.  We consider each possibility for $\ell$ and
eliminate them in turn.

\bigskip

Case I:  $\ell \geq n$.

There are two subcases:

\bigskip

Case Ia:  $n \nodiv \ell$:  In this case, by considering the first
$n$ symbols of each of the two occurrences of $x$ in $x x x^R$ in $z^\omega$,
we see that there are two different cyclic shifts of $z$ that are
identical.  This can only occur if ${\bf r}[0..n-1]$ is a power, and
we know from Theorem~\ref{rote3n} and Corollary~\ref{critical-rote}
that this implies that
$n = 2F_{3k+1}$ or
$n = 3F_{3k+1}$ for some $k \geq 0$.  But $2F_{3k+1} \not= 2F_{3k'+2}$ 
and $3F_{3k+1} \not= 2F_{3k'+2}$ 
provided $k, k' > 0$, so this case cannot occur.

Case Ib:  $n \divides \ell$:  Then $x$ is a conjugate of $z^e$, where
$e = \ell/n$.   By a well-known result, a conjugate of a power is a
power of a conjugate; hence there exists a conjugate $y$ of $z$ such
that $x = y^e$.   Then $x^R = y^e$, so $x$ and hence $y$ is a
palindrome.  We can now create a predicate that says that some
conjugate of ${\bf r}[0..n-1]$ is a palindrome:
$$\exists i<n \ \wedge\  (\forall j<n \ \cmp(i+j,n+i-1-j))$$
where 
\begin{multline*}
\cmp(k,k') :=
(((k<n) \ \wedge\ (k'<n)) \implies ({\bf r}[k] = {\bf r}[k'])) \ \wedge \  \\
(((k<n)\ \wedge \ (k' \geq n)) \implies ({\bf r}[k] = {\bf r}[k'-n])) \ \wedge \  \\
(((k \geq n)\ \wedge \ (k'<n)) \implies ({\bf r}[k-n] = {\bf r}[k']))  \ \wedge \  \\
(((k \geq n)\ \wedge \ (k' \geq n)) => ({\bf r}[k-n] = {\bf r}[k'-n]))  .
\end{multline*}

When we do this we discover the only $n$ with Fibonacci representation
of the form $10010^i$ accepted are those with $i \equiv 0, 2$ (mod $3$),
which means that $2F_{3k+2}$ is not among them.  So this case
cannot occur.

\bigskip

Case II:   $\ell < n$.

There are now four subcases to consider, depending on the number
of copies of $z$ needed to ``cover'' our occurrence of $x x x^R$.
In Case II.$j$, for $1 \leq j \leq 4$, we consider $j$ copies
of $z$ and the possible positions of $x x x^R$ inside that copy.

Because of the complicated nature of comparing one copy of $x$ to
itself in the case that one or both overlaps a boundary between
different copies of $z$, it would be very helpful to be able to encode
statements like ${\bf r}[k \bmod n] = {\bf r}[\ell \bmod n]$ in our
logical language.  Unfortunately, we cannot do this if $n$ is
arbitrary.  So instead, we use a trick:  assuming that the indices $k,
k'$ satisfy $0 \leq k, k' < 2n$, we can use the $\cmp(k,k')$
predicate introduced above
to simulate the assertion ${\bf r}[k \bmod n] = {\bf r}[k'
\bmod n]$.  Of course, for this to work we must ensure that $0 \leq k,
k' < 2n$ holds.

The cases are described in Figure~\ref{rotecon}.  We assume that
that $|x| = \ell$ and $x x x^R$ begins at position $i$ of
$z^\omega$.   We have the inequalities $i < n$ and $\ell < n$
which apply to each case.  Our predicates are
designed to compare the first copy of $x$ to
the second copy of $x$, and the first copy of $x$ to the $x^R$.

\begin{figure}[H]
\begin{center}
\includegraphics[width=5in]{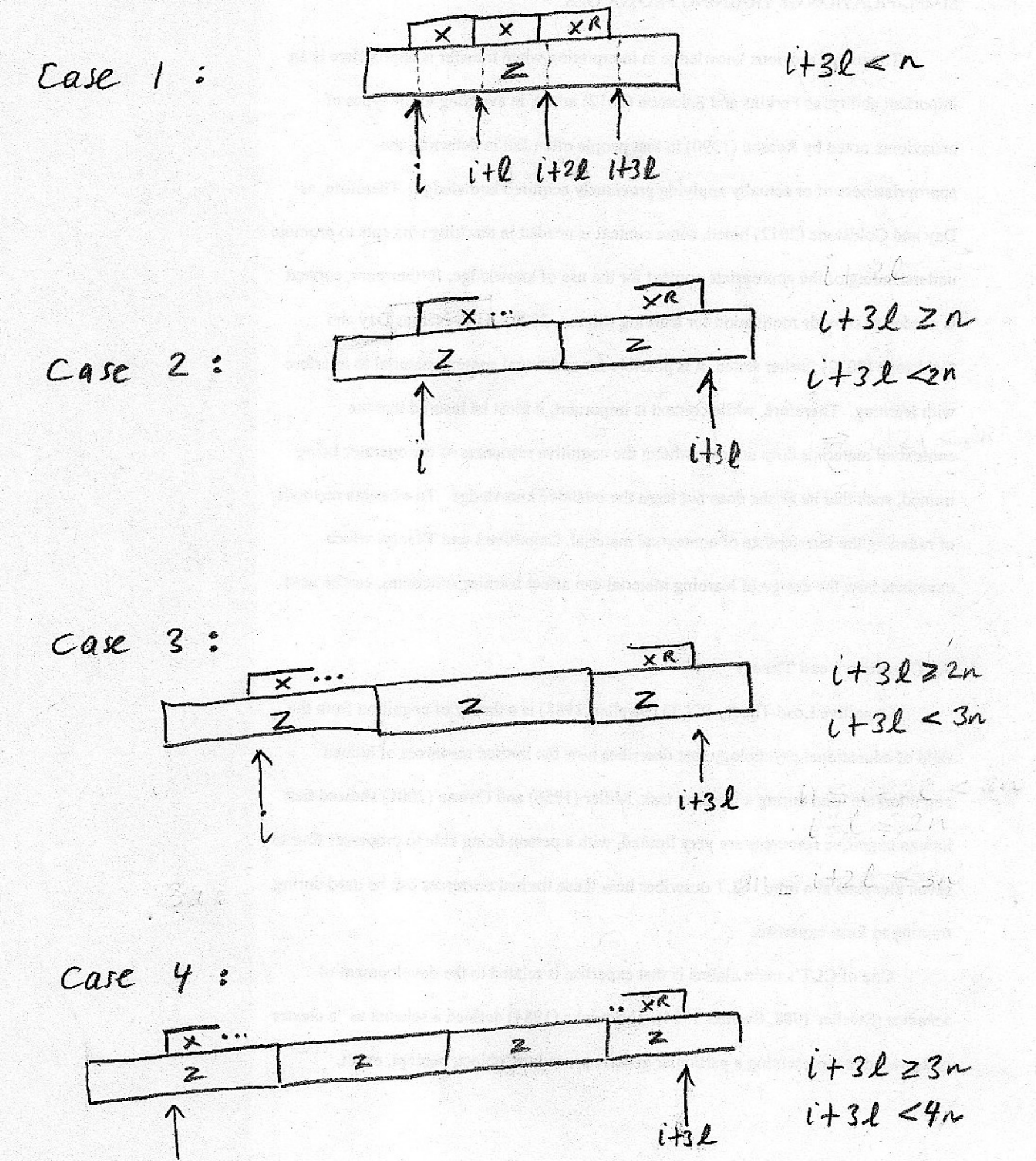}
\caption{Cases of the argument}
\label{rotecon}
\end{center}
\end{figure}

\medskip
Case 1:  If $xxx^R$ lies entirely within
one copy of $z$, it also lies in $\bf r$, which we have already seen
cannot happen, in Theorem~\ref{rote-avoid-thm}.  This case therefore cannot
occur.

\medskip
Case 2:    We use the predicate
$$
\exists i \ \exists \ell \ 
(i+3\ell \geq n) \ \wedge \ (i+3\ell < 2n) \ \wedge \ 
(\forall j < \ell\  \cmp(i+j, i+\ell+j ) ) \ \wedge \ 
(\forall k < \ell\ \cmp(i+k,i+3\ell-1-k) ) $$
to assert that there is a repetition of the form $x x x^R$.

\medskip
Case 3:  We use the predicate
$$
\exists i \ \exists \ell \ 
(i + 3\ell \geq 2n) \ \wedge \ (i+3\ell < 3n) \ \wedge \ 
(\forall j < \ell\ \cmp(i+j, i+\ell+j-n) ) \ \wedge \ 
(\forall k < \ell\ \cmp(i+k,i+3\ell-1-k-n) ) ) .$$

\medskip
Case 4:  We use the predicate 
$$
\exists i \ \exists \ell \ 
(i+3 \ell \geq 3n) \ \wedge \ (i+3\ell < 4n) \ \wedge \ 
(\forall j < \ell\ \cmp(i+j, i+\ell+j-n) ) \ \wedge \ 
(\forall k < \ell\ \cmp(i+k, i+3\ell-1-k-2n ) ) .$$

When we checked each of the cases 2 through 4 with our program,
we discovered that $n = 2F_{3k+2}$ is
never accepted.  Actually, for cases (2)--(4) we had to employ
one additional trick, because the computation for the
predicates as stated required more space than was available on our machine.
Here is the additional trick:  instead of attempting to run the
predicate for all $n$, we ran it only for $n$ whose Fibonacci
representation was of the form $10010^*$.  This significantly restricted
the size of the automata we created and allowed the computation
to terminate.  In fact, we propagated this condition throughout
the predicate.

We therefore eliminated all possibilities for the occurrence
of $x x x^R$ in $z^\omega$ and so it
follows that
no $x x x^R$ occurs in $z^\omega$.
\end{proof}

\begin{openproblem}
How many binary words of length $n$ avoid the pattern $x x x^R$?
Is it polynomial in $n$ or exponential?
How about the number of binary words of length $n$
avoiding $x x x^R$ and simultaneously avoiding
$(2+\alpha)$-powers?
\end{openproblem}

Consider finite words of the form $x x x^R$ having no
proper factor of the form $w w w^R$.    

\begin{conjecture}
For $n = F_{3k+1}$ there are $4$ such words of length $n$.
For $n = F_{3k+1} \pm F_{3k-2}$ there are $2$ such words.
Otherwise there are none.

For $k \geq 3$ the $4$ words of length $n = F_{3k+1}$ are given
by ${\bf r}[p_i..p_i+n-1]$, $i = 1,2,3,4$,
where
\begin{align*}
(p_1)_F &= 1000 (010)^{k-3} 001 \\
(p_2)_F &= 10 (010)^{k-2} 001 \\
(p_3)_F &= 1001000 (010)^{k-3} 001 \\
(p_4)_F &= 1010 (010)^{k-2} 001 
\end{align*}

For $k \geq 3$ the $2$ words of length
$n = F_{3k+1}-F_{3k-2}$ are given by
${\bf r}[q_i..q_i+n-1]$, $i = 1,2$, where
\begin{align*}
(q_1)_F &= 10 (010)^{k-3} 001 \\
(q_2)_F &= 10000 (010)^{k-3} 001 
\end{align*}

For $k \geq 3$ the $2$ words of length
$n = F_{3k+1}+F_{3k-2}$ are given by
${\bf r}[s_i..s_i+n-1]$, $i = 1,2$, where
\begin{align*}
(s_1)_F &= 10 (010)^{k-3} 001 \\
(s_2)_F &= 1000 (01)^{k-2} 001
\end{align*}

\end{conjecture}

\section{Other sequences}
\label{other}

In this section we briefly apply our method to some other 
Fibonacci-automatic sequences, obtaining several new results.

Consider a Fibonacci analogue of the Thue-Morse sequence 
$${\bf v} = (v_n)_{n \geq 0} = 0111010010001100010111000101 \cdots$$ 
where $v_n$ is the sum of the bits,
taken modulo $2$, of the Fibonacci representation of $n$.  This
sequence was introduced in \cite[Example 2, pp.\ 12--13]{Shallit:1988a}.  

We recall that an {\it overlap} is a word of the form $axaxa$ where
$x$ may be empty; its
order is defined to be $|ax|$.  Similarly, a {\it super-overlap} is
a word of the form $abxabxab$; an example of a super-overlap in English is 
the word {\tt tingalingaling} with the first letter removed.

\begin{theorem}
The only squares in $\bf v$ are of order $4$ and $F_n$ for $n \geq 2$,
and a square of each such order occurs.
The only cubes in $\bf v$ are the strings $000$ and $111$.   
The only overlaps in $\bf v$ are of order $F_{2n}$ for $n \geq 1$, and
an overlap of each such order occurs.  There are no super-overlaps
in $\bf v$.
\end{theorem}

\begin{proof}
As before.  We omit the details.
\end{proof}

We might also like to show that $\bf v$ is recurrent.  The obvious
predicate for this property holding for all words of length $n$ is
$$ \forall i\ \exists j\  ((j>i) \wedge ( \forall t \ ((t<n) \implies
({\bf v}[i+t]={\bf v}[j+t])))) .$$
Unfortunately, when we attempt to run this with our prover, we
get an intermediate NFA of 1159 states that we cannot determinize
within the available space.  

Instead, we rewrite the predicate, setting $k := j-i$ and $u := i+t$.
This gives
$$ \forall i\ \exists j \ (j>i) \wedge 
\forall k \ \forall u \ 
((k \geq 1) \wedge (i=j+k) \wedge (u \geq i) \wedge (u < n+i))
\implies {\bf v}[u]={\bf v}[u+k] .$$
When we run this we discover that $\bf v$ is indeed recurrent.
Here the computation takes a nontrivial 814007 ms, and the largest
intermediate automaton has 625176 states.  This proves


\begin{theorem}
The word $\bf v$ is recurrent.
\end{theorem}

Another quantity of interest for the Thue-Morse-Fibonacci word $\bf v$
is its subword complexity $\rho_{\bf v}(n)$.
It is not hard to see that it is linear.
To obtain a deeper understanding of it, let us compute the first difference
sequence $d(n) = \rho_{\bf v}(n+1) - \rho_{\bf v}(n)$.   It is easy
to see that $d(n)$ is the number of words $w$ of length $n$ with the
property that both $w0$ and $w1$ appear in $\bf v$.
The natural way to count this is to count those $i$ such that
$t:= {\bf v}[i..i+n-1]$ is the first appearance of that factor in $\bf v$,
and there exists a factor ${\bf v}[k..k+n]$ of length $n+1$ whose
length-$n$-prefix equals $t$ and whose last letter ${\bf v}[k+n]$ differs
from ${\bf v}[i+n]$.  
$$ (\forall j<i \ \exists t<n \ {\bf v}[i+t] \not= {\bf v}[j+t]) \ \wedge \ 
(\exists k\ (\forall u <n\  {\bf v}[i+u]={\bf v}[k+u]) \wedge 
	{\bf v}[i+n] \not= {\bf v}[k+n]).$$
Unfortunately the same blowup appears as in the recurrence predicate,
so once agin we need to substitute, resulting in the predicate

\begin{multline*}
(\forall j<i \ \exists k\geq 1\ \exists v\ 
(i=j+k) \wedge (v \geq j) \wedge (v<n+j) \wedge {\bf v}[u] \not= {\bf v}[u+k] )
\wedge  \\
(\exists l>i \ {\bf v}[i+n] \not= {\bf v}[l+n] ) \wedge \\
(\forall k' \ \forall u' \ 
(k'\geq 1) \wedge (l = i+k') \wedge (u' \geq i) \wedge (v' < n+i)
\implies {\bf v}[k'+u']={\bf v}[u'] ) .
\end{multline*}


From this we obtain a linear representation of rank $46$.  We can now
consider all vectors of the form $u \{ M_0, M_1 \}^*$.  There are only
finitely many and we can construct an automaton out of them computing
$d(n)$.

\begin{theorem}
The first difference sequence $(d(n))_{n \geq 0}$ of the subword complexity
of $\bf v$ is Fibonacci-automatic, and is accepted by the following
machine.

\begin{figure}[H]
\begin{center}
\includegraphics[width=6.5in]{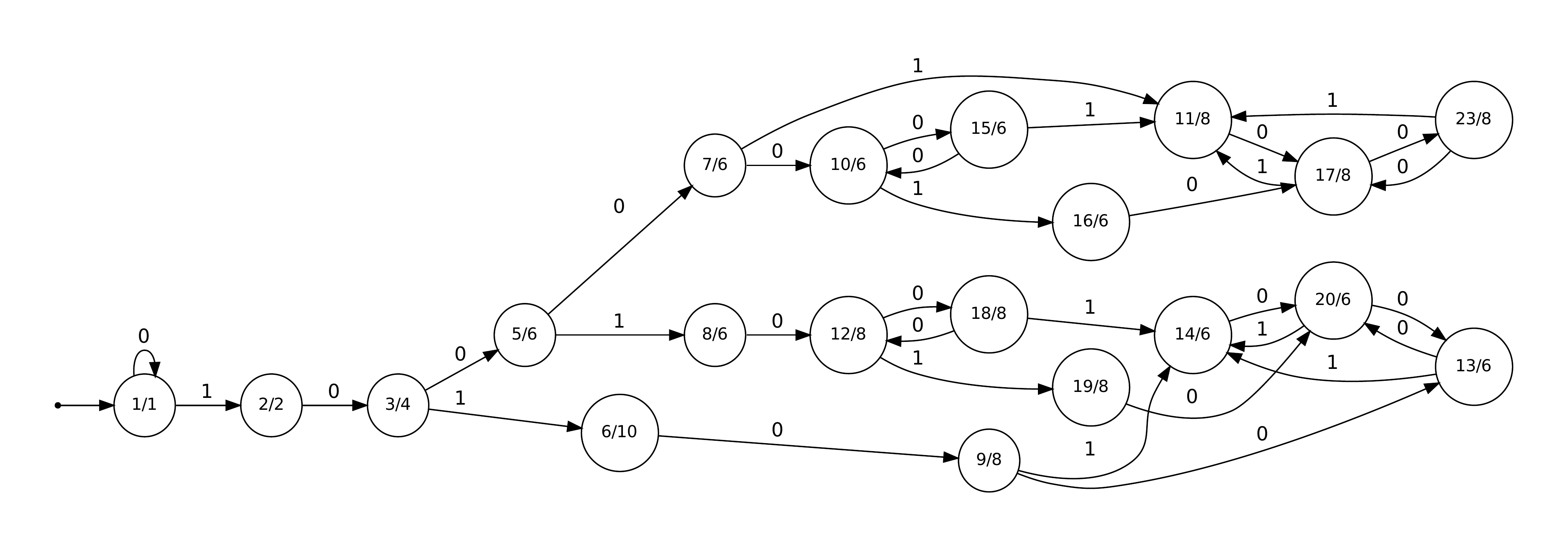}
\caption{Automaton computing $d(n)$}
\label{tmf-specialf}
\end{center}
\end{figure}

\end{theorem}

\section{Combining two representations and avoidability}
\label{additive}

In this section we show how our decidability method can be used to handle
an avoidability question where two different representations arise.

Let $x$ be a finite word over the alphabet
$\Enn^* = \lbrace 1, 2, 3 \ldots \rbrace$.    
We say that $x$ is an {\it additive square\/} if 
$x = x_1 x_2$ with $|x_1| = |x_n|$ and $\sum x_1 = \sum x_2$.
For example, with the usual association of
${\tt a} = 1$, ${\tt b} = 2$, and so forth, up to
${\tt z} = 26$, we have that the English word
{\tt baseball} is an additive square, as {\tt base} and
{\tt ball} both sum to $27$.

An infinite word ${\bf x}$ over $\Enn^*$ is said to {\it avoid
additive squares} if no factor is an additive square.  
It is currently unknown, and a relatively famous open problem,
whether there exists an infinite word over a {\it finite\/} subset
of $\Enn^*$ that avoids additive squares
\cite{Brown&Freedman:1987,Pirillo&Varricchio:1994,Halbeisen&Hungerbuhler:2000}.., although
it is known that additive cubes can be avoided over an alphabet
of size $4$ \cite{Cassaigne&Currie&Schaeffer&Shallit:2013}. 
(Recently this was improved to alphabet size $3$; see \cite{Rao:2013}.)

However, it is easy to avoid additive squares over an {\it infinite}
subset of $\Enn^*$; for example, any sequence that grows sufficiently
quickly will have the desired property.  Hence it is reasonable to
ask about the {\it lexicographically least} sequence over $\Enn^*$
that avoids additive squares.  Such a sequence begins
$$ 1 2 1 3 1 2 1 4 2 1 2 5 2 1 3 1 2 1 3 4 1 2 1 7 2 \cdots  ,$$
but we do not even know if this sequence is unbounded.

Here we consider the following variation on this problem.  Instead
of considering arbitrary sequences, we start with a sequence
${\bf b} = b_0 b_1 b_2 \cdots$ over $\Enn^+$ and from it construct the
sequence $S({\bf b}) = a_1 a_2 a_3 \cdots$ defined by
$$ {\bf a}[i] = {\bf b}[\nu_2 (i)]$$
for $i \geq 1$, where $\nu_2(i)$ is the exponent of the largest power 
of $2$ dividing $i$.  (Note that ${\bf a}$ and ${\bf b}$ are indexed
differently.)  For example, if ${\bf b} = 123\cdots$, then
${\bf a} = 1213121412131215 \cdots$, the so-called ``ruler
sequence''.  It is known that this sequence is squarefree and is,
in fact, the lexicographically least sequence over $\Enn^*$
avoiding squares \cite{Guay-Paquet&Shallit:2009}.

We then ask:  what is the lexicographically least sequence
avoiding additive squares that is of the form
$S({\bf b})$?  The following theorem gives the answer.

\begin{theorem} \label{thm:lex-least-add-sq}
The lexicographically least sequence over $\Enn \delete \{0\}$ of the
form $S({\bf b})$ that avoids additive squares is defined by ${\bf
b}[i] \coloneq F_{i+2}$.
\label{additive-thm}
\end{theorem}

\begin{proof}
First, we show that ${\bf a} \coloneq S({\bf b}) = \prod_{k=1}^\infty {\bf b}[\nu_2(k)] = \prod_{k=1}^\infty F_{\nu_2(k)+2}$ avoids additive squares.

For $m,n,j \in \Enn$, let $A(m,n,j)$ denote the number of occurrences of $j$ in $\nu_2(m+1), \dots, \nu_2(m+n)$.

(a):   Consider two consecutive blocks of the same size 
say $a_{i+1} \cdots a_{i+n}$ and $a_{i+n+1} \cdots a_{i+2n}$.  
Our goal is to compare the sums $\sum_{i < j \leq i+n} a_j$ and
$\sum_{i+n < j \leq i+2n} a_j$.  

First we prove

\begin{lemma}
Let $m,j \geq 0$ and $n \geq 1$ be  integers.  Let
$A(m, n,j)$ denote the number of occurrences of $j$
in $\nu_2 (m+1), \ldots, \nu_2 (m+n)$.  Then
for all $m, m' \geq 0$ we have
$|A(m', n, j) - A(m, n, j)| \leq 1$.
\label{flemm}
\end{lemma}

\begin{proof}
We start by observing that the number of positive integers $\leq n$
that are divisible by $t$ is exactly $\lfloor n/t \rfloor$.  It 
follows that the number $B(n,j)$ of positive integers $\leq n$ that are
divisible by $2^j$ but not by $2^{j+1}$ is 
\begin{equation}
B(n,j) = \lfloor {n \over {2^j}} \rfloor - \lfloor {n \over {2^{j+1}}} \rfloor .
 \label{fl}
\end{equation}
Now from the well-known identity
$$ \lfloor x \rfloor + \lfloor x + {1 \over 2} \rfloor = \lfloor 2x \rfloor,$$
valid for all real numbers $x$, substitute $x = n/2^{j+1}$ to get
$$ \lfloor {n \over {2^{j+1}}} \rfloor +
\lfloor {n \over {2^{j+1}}} + {1 \over 2} \rfloor = \lfloor {n \over {2^j}} \rfloor ,$$
which, combined with \eqref{fl}, shows that
$$B(n,j) = \lfloor {n \over {2^{j+1}}} + {1 \over 2} \rfloor .$$
Hence 
\begin{equation}
{n \over {2^{j+1}}} - {1 \over 2} \leq B(n,j) <
{n \over {2^{j+1}}} + {1 \over 2} .
\label{flooreq}
\end{equation}

Now the 
number of occurrences of $j$ in $\nu_2(m+1), \ldots, \nu_2(m+n)$
is $A(m,n,j) = B(m+n,j)-B(m,j)$.
From \eqref{flooreq} we get
\begin{equation}
{n \over {2^{j+1}}} - 1 < A(m,n,j) < {n \over {2^{j+1}}} + 1 
\label{flreq2}
\end{equation}
for all $m \geq 0$.
Since $A(m,n,j)$ is an integer, the inequality \eqref{flreq2} 
implies that
$|A(m',n,j)-A(m,n,j)| \leq 1$ for all $m, m'$.  
\end{proof}

Note that for all $i,n \in \Enn$, we have $\sum_{k=i}^{i+n-1} {\bf a}[k] = \sum_{j=0}^{\floor{\log_2(i+n)}} A(i,n,j) F_{j+2}$, so for adjacent blocks of length $n$, $\sum_{k=i+n}^{i+2n-1} {\bf a}[k] - \sum_{k=i}^{i+n-1} {\bf a}[k] = \sum_{j=0}^{\floor{\log_2(i+2n)}} (A(i+n,n,j)-A(i,n,j)) F_{j+2}$.
Hence, ${\bf a}[i \dotdot i+2n-1]$ is an additive square iff $\sum_{j=0}^{\floor{\log_2(i+2n)}} (A(i+n,n,j)-A(i,n,j)) F_{j+2} = 0$, and by above, each $A(i+n,n,j)-A(i,n,j) \in \{-1,0,1\}$.

The above suggests that we can take advantage of ``unnormalized'' Fibonacci representation in our computations.
For $\Sigma \subseteq \Zee$ and $w \in \Sigma^*$, we let the unnormalized Fibonacci representation $\ip{w}_{uF}$ be defined in the same way as $\ip{w}_F$, except over the alphabet $\Sigma$.

In order to use Procedure~\ref{proc:Fib-auto-decide}, we need two auxiliary DFAs: one that, given $i,n \in \Enn$ (in any representation; we found that base 2 works), computes $\ip{A(i+n,n,\_)-A(i,n,\_)}_{uF}$, and another that, given $w \in \{{\tt -1},{\tt 0},{\tt 1}\}^*$, decides whether $\ip{w}_{uF} = 0$.
The first task can be done by a 6-state (incomplete) DFA $M_\text{add22F}$ that accepts the language $\{ z \in (\Sigma_2^2 \times \{{\tt -1},{\tt 0},{\tt 1}\})^* \st \forall j (\pi_3(z)[j] = A(\ip{\pi_1(z)}_2+\ip{\pi_2(z)}_2,\ip{\pi_2(z)}_2,j) - A(\ip{\pi_1(z)}_2,\ip{\pi_2(z)}_2,j))\}$.
The second task can be done by a 5-state (incomplete) DFA $M_\text{1uFisZero}$ that accepts the language $\{ w \in \{{\tt -1},{\tt 0},{\tt 1}\}^* \st \ip{w}_{uF} = 0 \}$.

We applied a modified Procedure~\ref{proc:Fib-auto-decide} to the predicate $n \geq 1 \AND \exists w ({\tt add22F}(i,n,w) \AND {\tt 1uFisZero}(w))$ and obtained as output a DFA that accepts nothing, so ${\bf a}$ avoids additive squares.

Next, we show that ${\bf a}$ is the lexicographically least sequence over $\Enn \delete \{0\}$ of the form $S({\bf b})$ that avoids additive squares.

Note that for all ${\bf x},{\bf y} \in \Enn \delete \{0\}$, $S({\bf x}) < S({\bf y})$ iff ${\bf x} < {\bf y}$ in the lexicographic ordering.
Thus, we show that if any entry ${\bf b}[s]$ with ${\bf b}[s] > 1$ is changed to some $t \in [1,{\bf b}[s]-1]$, then ${\bf a} = S({\bf b})$ contains an additive square using only the first occurrence of the change at ${\bf a}[2^s-1]$.
More precisely, we show that for all $s,t \in \Enn$ with $t \in [1,F_{s+2}-1]$, there exist $i,n \in \Enn$ with $n \geq 1$ and $i+2n < 2^{s+1}$ such that either 
($2^s-1 \in [i,i+n-1]$ and $\sum_{k=i+n}^{i+2n-1} {\bf a}[k] - \sum_{k=i}^{i+n-1} {\bf a}[k] + t = 0$)
or
($2^s-1 \in [i+n,i+2n-1]$ and $\sum_{k=i+n}^{i+2n-1} {\bf a}[k] - \sum_{k=i}^{i+n-1} {\bf a}[k] - t = 0$).

Setting up for a modified Procedure~\ref{proc:Fib-auto-decide}, we use the following predicate, which says ``$r$ is a power of $2$ and changing ${\bf a}[r-1]$ to any smaller number results in an additive square in the first $2r$ positions", and six auxiliary DFAs. Note that all arithmetic and comparisons are in base 2.
\begin{align*}
&{\tt powOf2}(r) \AND \forall t ((t \geq 1 \AND t<r \AND {\tt canonFib}(t)) \IMPLY \exists i \exists n (n \geq 1 \AND i+2n < 2r \AND {} \\
&\quad ((i<r \AND r\leq i+n \AND \forall w ({\tt add22F}(i,n,w) \IMPLY \forall x ({\tt bitAdd}(t,w,x) \IMPLY {\tt 2uFisZero}(x)))) \OR {} \\
&\quad \hphantom{(}(i+n<r \AND r \leq i+2n \AND \forall w ({\tt add22F}(i,n,w) \IMPLY \forall x ({\tt bitSub}(t,w,x) \IMPLY {\tt 2uFisZero}(x))))))).
\end{align*}
\vspace{-2em}
\begin{align*}
L(M_\text{powOf2}) &= \{w \in \Sigma_2^* \st \exists n (w=(2^n)_2)\}. \\
L(M_\text{canonFib}) &= \{w \in \Sigma_2^* \st \exists n (w=(n)_F)\}. \\
L(M_\text{bit(Add/Sub)}) &= \{z \in (\Sigma_2 \times \{{\tt -1},{\tt 0},{\tt 1}\} \times \{{\tt -1},{\tt 0},{\tt 1},{\tt 2}\})^* \st \forall i (\pi_1(z)[i] \pm \pi_2(z)[i] = \pi_3(z)[i]) \}. \\
L(M_\text{2uFisZero}) &= \{w \in \{{\tt -1},{\tt 0},{\tt 1},{\tt 2}\}^* \st \ip{w}_{uF} = 0\}.
\end{align*}
We applied a modified Procedure~\ref{proc:Fib-auto-decide} to the above predicate and auxiliary DFAs and obtained as output $M_\text{powOf2}$, so ${\bf a}$ is the lexicographically least sequence over $\Enn \delete \{0\}$ of the form $S({\bf b})$ that avoids additive squares.
\end{proof}

\section{Enumeration}
\label{enumer}

Mimicking the base-$k$ ideas in
\cite{Charlier&Rampersad&Shallit:2012}, we can also mechanically enumerate
many aspects of Fibonacci-automatic sequences.  We do this
by encoding the factors having the
property in terms of paths of an automaton.  This gives the
concept of {\it Fibonacci-regular sequence} as previously studied in
\cite{Allouche&Scheicher&Tichy:2000}.  Roughly speaking, a sequence
$(a(n))_{n \geq 0}$ taking values in $\Enn$ is Fibonacci-regular
if the set of sequences
$$ \{ (a([xw]_F)_{w \in \Sigma_2^*} \ : \ x \in \Sigma_2^* \} $$
is finitely generated.  Here we assume that $a([xw]_F)$ evaluates
to $0$ if $xw$ contains the string $11$.  Every Fibonacci-regular
sequence $(a(n))_{n \geq 0}$ has a {\it linear representation}
of the form $(u, \mu, v)$ where $u$ and $v$ are row and column
vectors, respectively, and $\mu:\Sigma_2 \rightarrow \Enn^{d \times d}$
is a matrix-valued morphism, where $\mu(0) = M_0$ and $\mu(1) = M_1$
are $d \times d$ matrices for some $d \geq 1$, such that
$$a(n) = u \cdot \mu(x) \cdot v$$
whenever $[x]_F = n$.  The {\it rank} of the representation is
the integer $d$.
As an example,
we exhibit a rank-$6$ linear representation for the sequence $a(n) = n+1$:
\begin{align*}
u &= [1 \ 2 \ 2 \ 3 \ 3 \ 2] \\
M_0 &= \left[ \begin{array}{cccccc}
	1 & 1 & 0 & 0 & 0 & 0 \\
	0 & 0 & 0 & 0 & 0 & 0 \\
	0 & 1 & 0 & 1 & 1 & 0 \\
	0 & 0 & 1 & 1 & 1 & 1 \\
	0 & 0 & 0 & 0 & 0 & 0 \\
	0 & 0 & 0 & 0 & 0 & 0 
	\end{array} \right ] \\
M_1 &= \left[ \begin{array}{cccccc}
        0 & 0 & 0 & 0 & 0 & 0 \\
	1 & 0 & 0 & 0 & 0 & 0 \\
	0 & 0 & 0 & 0 & 0 & 0 \\
	0 & 0 & 0 & 0 & 0 & 0 \\
	0 & 0 & 1 & 1 & 0 & 0 \\
	0 & 0 & 0 & 1 & 0 & 0 
	\end{array} \right ] \\
v &= [1 \ 0 \ 0 \ 0 \ 0 \ 0 ]^T .
\end{align*}
This can be proved by a simple induction on the claim that 
$$u \cdot \mu(x) = [ x_F + 1 \ (1x)_F + 1 \ (10x)_F - x_F 
\ (100x)_F - x_F \ (101x)_F - (1x)_F \ (1001x)_F - (101x)_F ] $$
for strings $x$.

Recall that if $\bf x$ is an infinite word, then the subword
complexity function $\rho_{\bf x} (n)$ counts the number of
distinct factors of length $n$.  Then, in analogy with
\cite[Thm.~27]{Charlier&Rampersad&Shallit:2012}, we have

\begin{theorem}
If $\bf x$ is Fibonacci-automatic, then the subword complexity
function of $\bf x$ is Fibonacci-regular.
\end{theorem}

Using our implementation, we can obtain a linear representation
of the subword complexity function for $\bf f$.  To do so, we
use the predicate
$$ \{ (n,i)_F \ : \ \forall i' < i \ {\bf f}[i..i+n-1] \not=
{\bf f}[i'..i'+n-1] \} ,$$  
which expresses the assertion that the factor of length
$n$ beginning at position $i$ has never appeared before.
Then, for each $n$, the number of corresponding $i$ gives
$\rho_{\bf f}(n)$.  When we do this for $\bf f$, we get the
following linear representation $(u', \mu', v')$ of 
rank $10$:
\begin{align*}
u' &= [0 \ 0 \  0 \ 1 \ 0 \ 0 \ 0 \ 0 \ 0\  0] \\
M'_0 &= \left[ \begin{array}{ccccccccccc}
0 & 1 & 1 & 0 & 0 & 0 & 0 & 0 & 0 & 0 \\
0 & 0 & 0 & 0 & 0 & 0 & 0 & 0 & 0 & 0\\
1 & 0 & 0 & 0 & 0 & 0 & 0 & 0 & 0 & 0\\
0 & 0 & 0 & 1 & 0 & 0 & 0 & 0 & 0 & 0\\
1 & 0 & 0 & 0 & 0 & 1 & 0 & 0 & 0 & 0\\
0 & 0 & 0 & 0 & 0 & 0 & 0 & 1 & 0 & 0\\
0 & 0 & 0 & 0 & 0 & 1 & 0 & 1 & 0 & 0\\
0 & 0 & 0 & 0 & 0 & 1 & 0 & 1 & 0 & 0\\
1 & 0 & 0 & 0 & 0 & 0 & 0 & 0 & 0 & 0\\
0 & 0 & 0 & 0 & 0 & 0 & 0 & 1 & 0 & 0
\end{array}
\right] \\
M'_1 &= \left[ \begin{array}{ccccccccccc}
0 & 0 & 0 & 0 & 1 & 0 & 0 & 0 & 0 & 1 \\
0 & 0 & 0 & 0 & 0 & 0 & 0 & 0 & 1 & 0 \\
0 & 0 & 0 & 0 & 0 & 0 & 1 & 0 & 0 & 0 \\
0 & 0 & 0 & 0 & 0 & 0 & 1 & 0 & 1 & 0 \\
0 & 0 & 0 & 0 & 0 & 0 & 0 & 0 & 0 & 0 \\
0 & 0 & 0 & 0 & 0 & 0 & 1 & 0 & 0 & 0 \\
0 & 0 & 0 & 0 & 0 & 0 & 0 & 0 & 0 & 0 \\
0 & 0 & 0 & 0 & 0 & 0 & 1 & 0 & 0 & 1 \\
0 & 0 & 0 & 0 & 0 & 0 & 0 & 0 & 0 & 0 \\
0 & 0 & 0 & 0 & 0 & 0 & 0 & 0 & 0 & 0
\end{array}
\right] \\
v' &= [1\ 0\ 1\ 1\ 1\ 1\ 1\ 1\ 1\ 1]^T 
\end{align*}

To show that this computes the function $n+1$, it suffices
to compare the values of the linear representations 
$(u, \mu, v)$ and $(u', \mu', v')$ for all strings of length
$\leq 10 + 6 = 16$ (using \cite[Corollary 3.6]{Berstel&Reutenauer:2011}).
After checking this, 
we have reproved the following classic
theorem of Morse and Hedlund \cite{Morse&Hedlund:1940}:

\begin{theorem}
The subword complexity function of $\bf f$ is $n+1$.
\label{sturmcomp}
\end{theorem}

We now turn to a result of Fraenkel and Simpson \cite{Fraenkel&Simpson:1999}.
They computed the exact number of squares appearing in the
finite Fibonacci words $X_n$; this was previously estimated by
\cite{Crochemore:1981}.

There are two variations:  we could count the number of distinct
squares in $X_n$, or what Fraenkel and Simpson called the number of
``repeated squares'' in $X_n$ (i.e., the total number of {\it
occurrences} of squares in $X_n$).

To solve this using our approach, we generalize the problem to
consider any length-$n$ prefix of $X_n$, and not simply the prefixes
of length $F_n$.

We can easily write down predicates for these.  The first represents
the number of distinct squares in ${\bf f}[0..n-1]$:
\begin{multline*}
L_{\rm ds} :=
\{ (n,i,j)_F \ : \ (j \geq 1) \text{ and } (i+2j \leq n) \text{ and }
	{\bf f}[i..i+j-1] = {\bf f}[i+j..i+2j-1]  \\
\text{ and } \forall i' < i \ 
{\bf f}[i'..i'+2j-1] \not= {\bf f}[i..i+2j-1] \} .
\end{multline*}
This predicate asserts that ${\bf f}[i..i+2j-1]$ is a square 
occurring in ${\bf f}[0..n-1]$ and 
that furthermore it is the first occurrence of this particular
string in ${\bf f}[0..n-1]$.

The second represents the total number of occurrences of squares
in ${\bf f}[0..n-1]$:
$$ L_{\rm dos} := \{ (n,i,j)_F \ : \ (j \geq 1) \text{ and }
(i+2j \leq n) \text{ and }
        {\bf f}[i..i+j-1] = {\bf f}[i+j..i+2j-1] \} .$$
This predicate asserts that ${\bf f}[i..i+2j-1]$ is a square 
occurring in ${\bf f}[0..n-1]$.

We apply our method to the second example, leaving the first to the reader.
Let $b(n)$ denote the number of
occurrences of squares in ${\bf f}[0..n-1]$.  
First, we use our method to find
a DFA $M$ accepting $L_{\rm dos}$.
This (incomplete) DFA has 27 states.

Next, we compute matrices $M_0$ and $M_1$, indexed by states of $M$,
such that $(M_a)_{k,l}$ counts the number of edges (corresponding to the
variables $i$ and $j$) from state $k$ to state
$l$ on the digit $a$ of $n$.   We also compute a vector $u$ corresponding
to the initial state of $M$ and a vector $v$ corresponding to the final
states of $M$.   This gives us the
following linear representation of the sequence
$b(n)$:
if $x = a_1 a_2 \cdots a_t$ is the Fibonacci representation of $n$,
then
\begin{equation}
b(n) = u M_{a_1} \cdots M_{a_t} v ,
\label{linrep}
\end{equation}
which, incidentally, gives a fast algorithm for computing $b(n)$ for
any $n$.

Now let $B(n)$ denote the number of square occurrences in 
the finite Fibonacci word $X_n$.  This corresponds to considering
the Fibonacci representation of the form $10^{n-2}$; that is,
$B(n+1) = b([10^{n-1}]_F)$.
The matrix $M_0$ is the following $27 \times 27$ array
\begin{equation}
\left[
\begin{array}{ccccccccccccccccccccccccccc}
1&1&0&0&0&0&0&0&0&0&0&0&0&0&0&0&0&0&0&1&0&0&0&0&0&0&0 \\
1&0&0&0&0&0&0&0&0&0&0&0&0&0&0&0&0&0&0&0&1&0&0&0&0&0&0 \\
1&0&0&0&0&0&0&0&0&0&0&0&0&0&0&1&0&0&0&0&0&0&0&0&0&0&0 \\
0&0&0&0&0&0&0&0&0&0&0&0&0&0&0&0&0&0&0&0&1&0&0&0&0&0&0 \\
0&0&0&0&0&0&0&0&0&0&0&0&0&0&0&0&0&0&0&0&0&0&0&0&0&0&1 \\
1&0&0&0&0&0&0&0&0&0&0&0&0&0&0&1&0&0&0&0&0&0&0&0&0&0&0 \\
0&0&0&0&0&0&0&0&0&0&0&0&0&0&0&1&0&0&0&0&0&0&1&0&0&0&0 \\
0&0&0&0&0&0&0&0&0&0&0&0&0&0&0&0&0&1&0&0&0&0&0&0&0&0&0 \\
0&0&0&1&0&0&0&0&0&0&0&0&0&0&0&0&0&0&0&0&1&0&0&0&0&0&0 \\
0&0&0&0&0&0&0&0&0&0&0&0&0&0&0&0&0&1&0&0&0&0&0&0&0&0&0 \\
1&0&0&0&0&1&0&0&0&0&0&0&0&0&0&0&0&0&0&1&0&0&0&0&0&0&0 \\
0&0&0&0&0&0&0&0&0&0&0&0&0&0&0&0&0&0&0&0&0&0&0&0&0&0&0 \\
1&1&0&0&0&0&0&0&0&0&0&0&0&0&0&0&0&0&0&1&0&0&0&0&0&0&0 \\
0&0&0&0&0&0&0&0&0&0&1&0&0&0&0&0&0&0&0&0&0&0&1&0&0&1&0 \\
0&0&0&0&0&0&0&0&0&0&0&0&0&0&0&0&0&0&0&1&1&0&0&0&0&0&0 \\
0&0&0&0&0&0&0&0&0&1&0&0&0&0&0&0&0&0&0&0&1&0&0&0&0&0&0 \\
1&0&0&0&0&0&0&0&0&0&0&0&0&0&0&0&0&0&0&0&1&0&0&0&0&0&0 \\
0&0&0&0&0&0&0&0&0&0&0&0&0&0&0&0&0&1&0&0&0&0&0&0&0&0&0 \\
0&0&0&0&0&0&0&0&0&0&1&0&0&0&0&0&0&0&0&0&0&0&1&0&0&1&0 \\
0&0&0&0&0&0&0&0&0&0&0&0&0&0&0&0&0&0&0&0&0&0&0&0&0&0&1 \\
0&0&0&1&0&0&0&0&0&0&0&0&0&0&0&0&0&0&0&0&1&0&0&0&0&0&0 \\
0&0&0&0&0&0&0&0&0&0&0&0&0&0&0&0&0&0&0&0&0&1&0&0&0&0&0 \\
0&0&0&0&0&0&0&0&0&0&0&1&0&0&0&0&0&0&0&0&0&0&0&0&0&0&0 \\
0&0&0&0&0&0&0&0&0&0&0&0&0&0&0&0&0&0&0&0&0&1&0&0&0&0&0 \\
0&0&0&0&0&0&0&0&0&0&0&0&0&0&0&0&0&0&0&0&1&0&0&0&0&0&0 \\
0&0&0&0&0&0&0&0&0&0&0&0&0&0&0&0&0&0&1&0&0&0&0&0&0&0&0 \\
0&0&0&0&0&0&0&0&0&0&0&0&0&0&0&0&0&0&0&1&1&0&0&0&0&0&0 
\end{array}
\right]
\end{equation}
and has minimal polynomial
$$ X^4 (X-1)^2(X+1)^2(X^2-X-1)^2.$$
It now follows
from the theory of linear recurrences
that there are constants $c_1, c_2, \ldots, c_8$
such that
$$ B(n+1) = (c_1n + c_2) \alpha^n + (c_3n+c_4) \beta^n + c_5n+c_6 +
	(c_7n+c_8)(-1)^n $$
for $n \geq 3$, where $\alpha = (1+\sqrt{5})/2$, $\beta = (1-\sqrt{5})/2$
are the roots of $X^2 - X - 1$.
We can find these constants by computing $B(4), B(5), \ldots, B(11)$
(using Eq.~\eqref{linrep}) and then solving for the values
of the constants $c_1, \ldots, c_8$.

When we do so, we find 
\begin{align*}
c_1 &= {2 \over 5}  \quad\quad & c_2 &= {-{2\over{25}}}\sqrt{5} - 2 \\
c_3 &= {2 \over 5} \quad\quad & c_4 &= {{2\over{25}}}\sqrt{5} - 2 \\
c_5 &= 1 \quad\quad & c_6 &= 1  \\
c_7 &= 0 \quad\quad & c_8 &= 0
\end{align*}

A little simplification,
using the fact that $F_n = (\alpha^n - \beta^n)/(\alpha - \beta)$, leads to 
\begin{theorem}
Let $B(n)$ denote the number of square occurrences in $X_n$.  Then
$$B(n+1) = {4 \over 5} n F_{n+1} - {2 \over 5} (n+6) F_{n} - 4F_{n-1} + n + 1 $$
for $n \geq 3$.
\end{theorem}

This statement corrects a small error in
Theorem 2 in \cite{Fraenkel&Simpson:1999} (the coefficient of $F_{n-1}$
was wrong; note that their $F$ and their Fibonacci words
are indexed differently from ours),
which was first
pointed out to us by Kalle Saari.

In a similar way, we can count the cube occurrences in $X_n$.
Using analysis exactly like the square case, we easily find

\begin{theorem}
Let $C(n)$ denote the number of cube occurrences in the Fibonacci
word $X_n$.  Then for $n \geq 3$ we have
$$ C(n) = (d_1 n+ d_2) \alpha^n + (d_3 n+d_4) \beta^n + d_5n + d_6$$
where
\begin{align*}
d_1 &= {{3-\sqrt{5}}\over {10}} \quad\quad & d_2 &= {{17}\over {50}}\sqrt{5} - {3 \over 2} \\
d_3 &= {{3+\sqrt{5}}\over {10}} \quad\quad & d_4 &= -{{17}\over {50}}\sqrt{5} - {3 \over 2} \\
d_5 &= 1 \quad\quad & d_6 &= -1 .
\end{align*}
\end{theorem}
We now turn to a question of Chuan and Droubay.  Let us consider the
prefixes of $\bf f$.  For each prefix of length $n$, form all of its
$n$ shifts, and let us count the number of these shifts that are
palindromes; call this number $d(n)$.
(Note that in the case where a prefix is a power,
two different shifts could be identical; we count these with 
multiplicity.)  

Chuan \cite[Thm.~7, p.~254]{Chuan:1993b} proved

\begin{theorem}
For $i > 2$ we have
$d(F_i) = 0$ iff $i \equiv \modd{0} {3}$.
\label{chuan93-thm}
\end{theorem}

\begin{proof}
Along the way we 
actually prove a lot more, characterizing $d(n)$ for all $n$, not
just those $n$ equal to a Fibonacci number.

We start by showing that $d(n)$ takes only three values:  $0$, $1$, and $2$.
To do this, we construct an automaton accepting the language
$$ \{ (n,i)_F \ : \ (0 \leq i < n) \ \wedge\ {\bf f}[i..n-1]{\bf f}[0..i-1]
	\text{ is a palindrome } \} .$$
From this we construct the linear representation
$(u, M_0, M_1, v)$ of $d(n)$ as discussed
above; it has rank $27$.

The range of $c$ is finite if
the monoid ${\cal M} = \langle M_0, M_1 \rangle$ is finite.
This can be checked
with a simple queue-based algorithm, and $\cal M$ turns out to have
cardinality $151$.    From these a simple computation proves
$$\lbrace uMv \ : \ M \in {\cal M} \rbrace 
= \lbrace 0, 1, 2 \rbrace,$$
and so our claim about the range of $c$ follows.

Now that we know the range of $c$ we can create predicates 
$P_0(n), P_1(n), P_2(n)$
asserting
that (a) there are no length-$n$ shifts that are palindromes 
(b) there is exactly one
shift that is a palindrome and (c) more than one shift is a palindrome,
as follows:
$$P_0 :  \neg \exists i,  (0 \leq i < n), {\bf f}[i..n-1]{\bf f}[0..i-1]
        \text{ is a palindrome } $$
$$P_1 : \exists i,  (0 \leq i < n), {\bf f}[i..n-1]{\bf f}[0..i-1]
        \text{ is a palindrome  and } \neg\exists j \not= i 
	(0 \leq j < n), {\bf f}[j..n-1]{\bf f}[0..j-1] $$
$$ P_2 : \exists i, j, 0 \leq i < j < n 
{\bf f}[i..n-1]{\bf f}[0..i-1] \text{ and }
{\bf f}[j..n-1]{\bf f}[0..j-1] \text{ are both palindromes }$$
For each one, we can compute a finite automaton characterizing the
Fibonacci representations of those $n$ for which $d(n)$ equals,
respectively, $0$, $1$, and $2$.

For example, we computed the automaton corresponding to $P_0$, and it is
displayed in Figure~\ref{noshifts} below.

\begin{figure}[H]
\begin{center}
\includegraphics[width=4in]{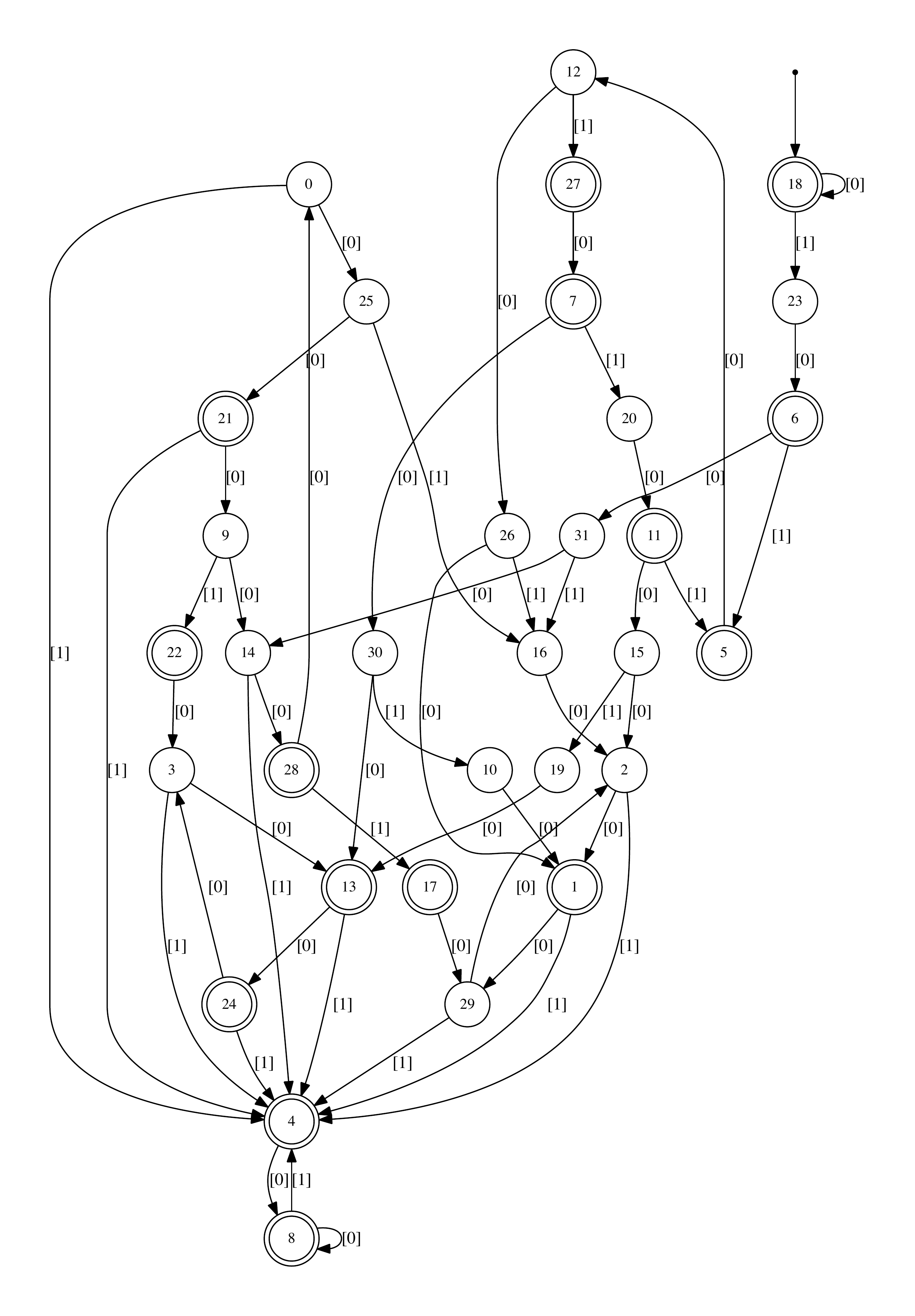}
\caption{Automaton accepting lengths of prefixes for which
no shifts are palindromes}
\label{noshifts}
\end{center}
\end{figure}

By tracing the path labeled $10^*$ starting at the initial state
labeled $18$, we see that the ``finality'' of the states encountered
is ultimately periodic with period $3$, proving Theorem~\ref{chuan93-thm}.
\end{proof}

To finish this section, we reprove a result of Kolpakov and
Kucherov \cite{Kolpakov&Kucherov:1999a}.  Recalling the
definition of maximal repetition from Section~\ref{repe-subsec}, they
counted the number $\mr(F_n)$ of occurrences of maximal repetitions in the
prefix of $\bf f$ of length $F_n$:

\begin{theorem}
For $n \geq 5$ we have $\mr(F_n) = 2F_{n-2} - 3$.
\end{theorem}

\begin{proof}
We create an automaton for the language
$$ \lbrace (n,i,j)_F \ : \ 0 \leq i \leq j < n \text{ and } {\bf f}[i..j]
	\text{ is a maximal repetition of } {\bf f}[0..n-1] \rbrace ,$$
using the predicate
\begin{multline*}
(i \leq j) \ \wedge\ (j<n)\ \wedge \  
\exists p \text{ with } 1 \leq p \leq (j+1-i)/2 \text{ such that } \\
	( (\forall k\leq j-i-p \ {\bf f}[i+k]= {\bf f}[i+k+p]) \ \wedge \  \\
	(i \geq 1) \implies (\forall q \text{ with } 1 \leq q \leq p \ 
		\exists \ell \leq j-i-q+1 \ {\bf f}[i-i+\ell] \not= {\bf f}[i-1+\ell+q]) 
		\ \wedge\ \\
	(j+1\leq n-1) \implies (\forall r \text{ with } 1 \leq r \leq p\ 
		\exists m \leq j+1-r-i \ {\bf f}[i+m] \not= {\bf f}[i+m+r] )  ) .
\end{multline*}
Here the second line of the predicate specifies that there is a period $p$
of ${\bf f}[i..j]$ corresponding to a repetition of exponent at least $2$.
The third line specifies that no period $q$ of ${\bf f}[i-1..j]$ (when
this makes sense) can be $\leq p$, and the fourth line specifies that
no period $r$ of ${\bf f}[i..j+1]$ (when $j+1 \leq n-1$) can be
$\leq p$.  

From the automaton we deduce a linear representation
$(u, \mu, v)$ of rank 59.  Since $(F_n)_F = 10^{n-2}$, it suffices
to compute the minimal polynomial of $M_0 = \mu(0)$.  When we
do this, we discover it is $X^4(X^2 - X - 1)(X-1)^2(X+1)^2$.
It follows from the theory of linear recurrences that
$$\mr(F_n) = e_1 \alpha^n  + e_2 \beta^n + e_3 n + e_4 + (e_5n + e_6)(-1)^n $$
for constants $e_1, e_2, e_3, e_4, e_5, e_6$ and $n \geq 6$.
When we solve for $e_1, \ldots, e_6$ by using the first few values
of $\mr(F_n)$ (computed from the linear representation or directly)
we discover that $e_1 = (3\sqrt{5} - 5)/5$,
$e_2 = (-3\sqrt{5} -5)/5$, $e_3 = e_5 = e_6 = 0$, and $e_4 = -3$.
From this the result immediately follows.
\end{proof}

In fact, we can prove even more.  

\begin{theorem}
For $n \geq 0$ the difference $\mr(n+1) - \mr(n)$ is either $0$ or $1$.
Furthermore there is a finite automaton with 10 states that
accepts $(n)_F$ precisely when $\mr(n+1) - \mr(n) = 1$.
\end{theorem}

\begin{proof}
Every maximal repetition ${\bf f}[i..j]$
of ${\bf f}[0..n-1]$ is either a maximal
repetition of ${\bf f}[0..n]$ with $j \leq n-1$, or is a maximal
repetition with $j = n-1$ that, when considered in ${\bf f}[0..n]$,
can be extended one character to the right to become
one with $j = n$.  So the only maximal repetitions of ${\bf f}[0..n]$
not (essentially) counted by $\mr(n)$ are those such that 
\begin{multline}
{\bf f}[i..n] \text{ is a maximal repetition of } {\bf f}[0..n] 
\text{ and } \\
{\bf f}[i..n-1] \text{ is {\it not\/} a maximal repetition of } 
{\bf f}[0..n-1].
\label{condit}
\end{multline}

We can easily create a predicate asserting this latter condition, and
from this obtain the linear representation of $\mr(n+1) - \mr(n)$:
\begin{align*}
u &= [0\ 0\ 0\ 0\ 0\ 1\ 0\ 0\ 0\ 0\ 0\ 0\ ] \\
\mu(0) &= \left[ \begin{array}{cccccccccccc}
0&0&0&0&0&0&0&0&0&1&0&0\\
0&0&0&0&0&0&0&0&0&0&1&0\\
0&0&0&1&0&0&0&0&0&1&0&0\\
0&0&0&0&0&0&0&1&0&0&0&0\\
0&0&0&0&0&0&0&0&0&0&0&0\\
0&0&0&0&0&1&0&0&0&0&0&0\\
0&0&0&0&0&0&0&0&0&0&1&0\\
0&0&0&0&1&0&0&0&0&0&0&0\\
0&0&0&0&1&0&0&0&0&0&0&0\\
0&0&0&0&0&0&0&0&0&0&0&1\\
0&0&0&0&0&0&1&0&0&0&0&0\\
0&0&0&0&0&0&0&0&0&1&0&0
\end{array} \right] \\
\mu(1) &= \left[ \begin{array}{cccccccccccc}
0&0&0&0&0&0&0&0&0&0&0&0\\
0&0&0&0&0&0&0&0&0&0&0&0\\
0&0&0&0&0&0&0&0&0&0&0&0\\
0&0&0&0&0&0&0&0&0&0&0&0\\
0&0&0&0&0&0&0&0&1&0&0&0\\
0&1&1&0&0&0&0&0&0&0&0&0\\
1&1&0&0&0&0&0&0&0&0&0&0\\
0&0&0&0&0&0&0&0&0&0&0&0\\
0&0&0&0&0&0&0&0&0&0&0&0\\
0&0&0&0&0&0&0&0&0&0&0&0\\
0&1&1&0&0&0&0&0&0&0&0&0\\
0&0&0&0&0&0&0&0&1&0&0&0
\end{array} \right] \\
v &=  [0\ 0\ 0\ 0\ 1\ 0\ 0\ 0\ 1\ 0\ 0\ 1 ] \\
\end{align*}

We now use the trick we previously used for the proof of Theorem~\ref{chuan93-thm}; the monoid generated by $\mu(0)$ and $\mu(1)$ has size $61$ and for each
matrix $M$ in this monoid we have $u M v \in \lbrace 0, 1 \rbrace$.
It follows that
$\mr(n+1) - \mr(n) \in \lbrace 0, 1 \rbrace$ for all $n \geq 0$.

Knowing this, we can now build an automaton accepting those $n$ for
which there exists an $i$ for which \eqref{condit} holds.  When we do so
we get the automaton depicted below in Figure~\ref{maxrepp}.

\begin{figure}[H]
\begin{center}
\includegraphics[width=6in]{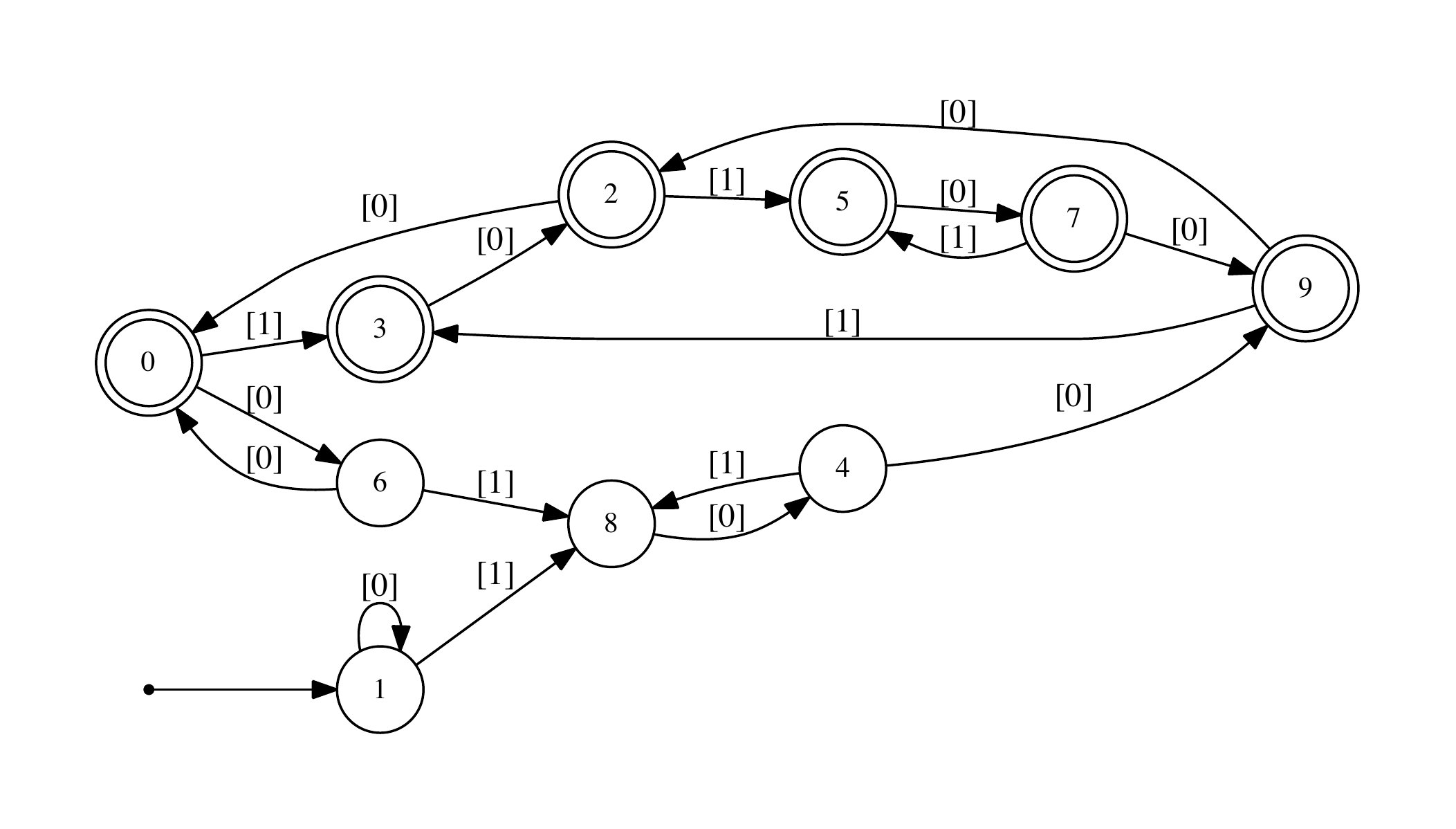}
\caption{Automaton accepting $(n)_F$ such that $\mr(n+1) - \mr(n) = 1$}
\label{maxrepp}
\end{center}
\end{figure}

\end{proof}

\section{Abelian properties}

Our decision procedure
does not apply, in complete generality, to abelian
properties of infinite words.  This is because there is no obvious
way to express assertions like $\psi(x) = \psi(x')$ for two factors
$x, x'$ of an infinite word.  (Here $\psi:\Sigma^* \rightarrow \Enn^{|\Sigma|}$
is the Parikh map that sends a word to the number of occurrences
of each letter.)  Indeed, in the $2$-automatic case it is provable that
there is at least one abelian property that is inexpressible
\cite[\S 5.2]{Schaeffer:2013}.

However, the special nature of
the Fibonacci word $\bf f$ allows us to mechanically prove some
assertions involving abelian properties.  In this section we describe
how we did this.

By an {\it abelian square of order $n$} we mean a factor of the
form $x x'$ where $\psi(x) = \psi(x')$, where $n = |x|$.  
In a similar way we can define abelian cubes and higher powers.

We start with the elementary observation that $\bf f$ is defined
over the alphabet $\lbrace 0, 1 \rbrace$.  Hence, to understand the
abelian properties of a factor $x$ it suffices to know $|x|$ and
$|x|_0$.  Next, we observe that the map that sends
$n$ to $a_n := |{\bf f}[0..n-1]|_0$ (that is, the number of $0$'s in the
length-$n$ prefix of $\bf f$), is actually {\it synchronized}
(see \cite{Carpi&Maggi:2001,Carpi&DAlonzo:2009,Carpi&DAlonzo:2010,Goc&Schaeffer&Shallit:2013}).  That is, 
there is a DFA accepting the Fibonacci representation of the
pairs $(n,a_n)$.  In fact we have the following

\begin{theorem}
Suppose the Fibonacci representation of $n$ is 
$e_1 e_2 \cdots e_i$.  Then $a_n = [e_1 e_2 \cdots e_{i-1}]_F + e_i$.
\label{fibr}
\end{theorem}

\begin{proof}
First, we observe that an easy induction on $m$ proves that
$|X_m|_0 = F_{m-1}$ for $m \geq 2$.  We will use this in a moment.

The theorem's claim is easily checked for $n = 0,1$.  We prove it for
$F_{m+1} \leq n < F_{m+2}$ by induction on $m$.  The base
case is $m = 1$, which corresponds to $n = 1$.    

Now assume the theorem's claim is true for $m-1$; we prove it for $m$.
Write $(n)_F = e_1 e_2 \cdots e_m$.  
Then, using the fact that ${\bf f}[0..F_{m+2}-1] = X_{m+2} = X_{m+1} X_m$,
we get
\begin{align*}
|{\bf f}[0..n-1]|_0 &= |{\bf f}[0..F_{m+1}-1]|_0 + |{\bf f}[F_{m+1}..n-1]|_0 \\
&= |X_{m+1}|_0 + |{\bf f}[0..n-1-F_{m+1}]|_0 \\
&= F_m + |{\bf f}[0..n-1-F_{m+1}|_0 \\
&= F_m + [e_2\cdots e_{m-1}]_F + e_m \\
&= [e_1 \cdots e_{m-1}]_F + e_m ,
\end{align*}
as desired.
\end{proof}

In fact, the synchronized automaton for $(n,a_n)_F$ is given in
the following diagram:

\begin{figure}[H]
\begin{center}
\includegraphics[width=6in]{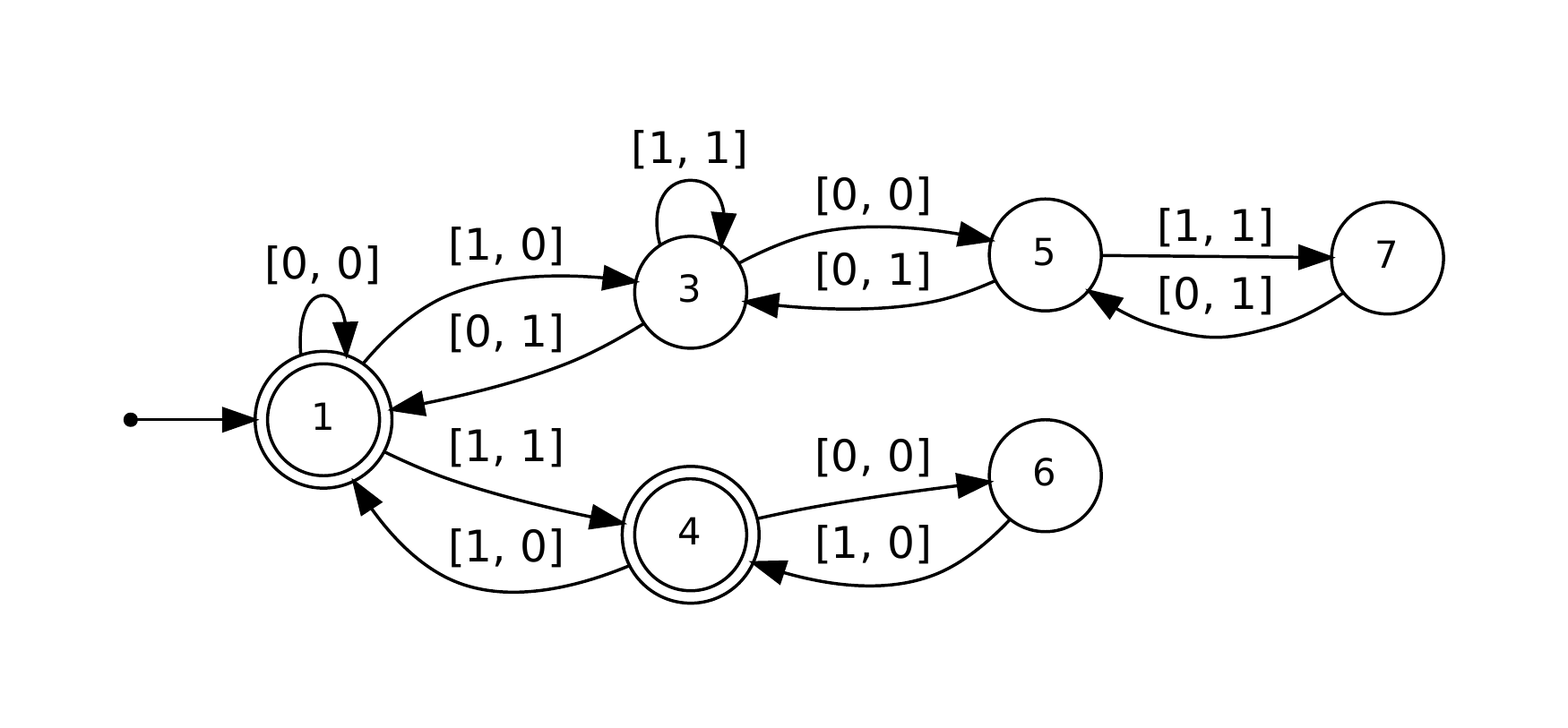}
\caption{Automaton accepting $(n,a_n)_F$}
\label{synchro}
\end{center}
\end{figure}

Here the missing state numbered $2$ is a ``dead'' state that is the
target of all undrawn transitions.

The correctness of this automaton can be checked using our prover.
Letting $\zc(x,y)$ denote $1$ if $(x,y)_F$ is accepted, it suffices
to check that
\begin{enumerate}
\item $ \forall x\ \exists y\ \zc(x,y)= 1$ (that is, for each $x$
there is at least one corresponding $y$ accepted);

\item $\forall x\ \forall y\ \forall z\ (\zc(x,y) = \zc(x,z)) \implies
y = z$ (that is, for each $x$ at most one corresponding $y$ is
accepted);

\item $\forall x \ \forall y \ ((\zc(x,y)=1) \ \wedge \ ({\bf f}[x] = 1))
\implies (\zc(x+1,y+1) = 1)$;

\item $\forall x \ \forall y \ ((\zc(x,y)=1) \ \wedge \ ({\bf f}[x] = 0))
\implies (\zc(x+1,y) = 1)$;

\end{enumerate}

Another useful automaton computes, on input $n, i, j$ the function
$$\fab(n,i,j) :=
|{\bf f}[i..i+n-1]|_0 - |{\bf f}[j..j+n-1]|_0 =
a_{i+n}-a_i - a_{j+n}+a_j.$$
From the known fact that the factors of $\bf f$ are ``balanced'' 
we know that $\fab$ takes only the values $-1, 0, 1$.  This automaton can
be deduced from the one above.  However, we calculated it by
``guessing'' the right automaton and then verifying the correctness
with our prover.    

The automaton for $\fab(n,i,j)$ has 30 states, numbered 
from $1$ to $30$.  Inputs are
in $\Sigma_2^3$.  The transitions, as well as the outputs,
are given in the table below.

\begin{table}[H]
\begin{center}
\begin{tabular}{|c|cccccccc|c}
$q$ & $[0,0,0]$ & $[0,0,1]$ & $[0,1,0]$ & $[0,1,1]$ & $[1,0,0]$ &
	$[1,0,1]$ & $[1,1,0]$ & $[1,1,1]$ & $\tau(q)$ \\
\hline
1& 1& 2& 3& 4& 4& 5& 6& 7& 0\\
2& 8& 1& 9& 3& 3& 4&10& 6& 0\\
3&11&12& 1& 2& 2&13& 4& 5& 0\\
4&14&11& 8& 1& 1& 2& 3& 4& 0\\
5&15&11&16& 1& 1& 2& 3& 4& 1\\
6&17&18& 8& 1& 1& 2& 3& 4&$-1$\\
7&19&18&16& 1& 1& 2& 3& 4& 0\\
8& 1& 2& 3& 4& 4&20& 6&21& 0\\
9&11&12& 1& 2& 2&22& 4&20& 0\\
10&18&23& 1& 2& 2&13& 4& 5&$-1$\\
11& 1& 2& 3& 4& 4& 5&24&25& 0\\
12& 8& 1& 9& 3& 3& 4&26&24& 0\\
13&16& 1&27& 3& 3& 4&10& 6& 1\\
14& 1& 2& 3& 4& 4&20&24&28& 0\\
15& 2&13& 4& 5& 5&20&25&28&$-1$\\
16& 2&13& 4& 5& 5&20& 7&21&$-1$\\
17& 3& 4&10& 6& 6&21&24&28& 1\\
18& 3& 4&10& 6& 6& 7&24&25& 1\\
19& 4& 5& 6& 7& 7&21&25&28& 0\\
20&15&14&16& 8& 8& 1& 9& 3& 1\\
21&19&17&16& 8& 8& 1& 9& 3& 0\\
22&16& 8&27& 9& 9& 3&29&10& 1\\
23& 9& 3&29&10&10& 6&26&24& 1\\
24&17&18&14&11&11&12& 1& 2&$-1$\\
25&19&18&15&11&11&12& 1& 2& 0\\
26&18&23&11&12&12&30& 2&13&$-1$\\
27&12&30& 2&13&13&22& 5&20&$-1$\\
28&19&17&15&14&14&11& 8& 1& 0\\
29&18&23& 1& 2& 2&22& 4&20&$-1$\\
30&16& 1&27& 3& 3& 4&26&24& 1
\end{tabular}
\end{center}
\caption{Automaton to compute $\fab$}
\end{table}

Once we have guessed the automaton, we can verify it as follows:

\begin{enumerate}
\item $\forall i \  \forall j\ \fab[0][i][j]=0$.  This is the basis
for an induction.
\item Induction steps:

	\begin{itemize}
	
	\item $\forall i\  \forall j\  \forall n \  ({\bf f}[i+n]={\bf f}[j+n]) 
\implies (\fab[n][i][j]=\fab[n+1][i][j])$.  

\item $\forall i\ \forall j\  \forall n\ (({\bf f}[i+n]=0) \wedge
({\bf f}[j+n]=1)) \implies
 (((\fab[n][i][j]=-1) \wedge (\fab[n+1][i][j]=0)) \vee
   ((\fab[n][i][j]=0) \wedge (\fab[n+1][i][j]=1))) $

\item $\forall i\ \forall j\  \forall n\ (({\bf f}[i+n]=0) \wedge
({\bf f}[j+n]=1)) \implies
 (((\fab[n][i][j]=1) \wedge (\fab[n+1][i][j]=0)) \vee
   ((\fab[n][i][j]=0) \wedge (\fab[n+1][i][j]=-1))) $.  
\end{itemize}
\end{enumerate}




As the first application, we prove

\begin{theorem}
The Fibonacci word $\bf f$ has abelian squares of all orders.
\end{theorem}

\begin{proof}
We use the predicate
$$  \exists i \ (\fab[n][i][i+n] = 0) .$$
The resulting automaton accepts all $n \geq 0$.  The total computing
time was 141 ms.
\end{proof}


Cummings and Smyth \cite{Cummings&Smyth:1997}
counted the total number of all
occurrences of (nonempty) abelian squares in the 
Fibonacci words $X_i$.  We can do this by using the predicate
$$ (k>0) \wedge (i+2k \leq n) \wedge (\fab[k][i][i+k]=0),$$
using the techniques in Section~\ref{enumer}
and considering the case where $n = F_i$.


When we do, we get 
a linear representation of rank 127 that counts the
total number $w(n)$ of occurrences of abelian squares in the 
prefix of length $n$ of the Fibonacci word.

To recover the Cummings-Smyth result we compute the minimal polynomial
of the matrix $M_0$ corresponding to the predicate above.  It is
$$x^4 (x-1)(x+1)(x^2+x+1)(x^2-3x+1)(x^2-x+1)(x^2+x-1)(x^2-x-1).$$

This means that $w(F_n)$, that is, $w$ evaluated at $10^{n-2}$ in Fibonacci
representation, is a linear combination of the roots of this polynomial to the
$n$'th power (more precisely, the $(n-2)$th, but this detail is
unimportant).
The roots of the polynomial are
$$ -1, 1, (-1 \pm i \sqrt{3})/2, (3 \pm \sqrt{5})/2, (1 \pm i \sqrt{3})/2,
(-1 \pm \sqrt{5})/2, (1 \pm \sqrt{5})/2.$$
Solving for the coefficients as we did in Section~\ref{enumer} we get

\begin{theorem}
For all $n \geq 0$ we have
\begin{multline*}
w(F_n) =
c_1 \left({{3+\sqrt{5}}\over 2}\right)^n + c_1 \left({{3-\sqrt{5}}\over 2}\right)^n 
+
c_2 \left( {{1+\sqrt{5}}\over 2} \right)^n + c_2 \left( {{1-\sqrt{5}}\over 2} \right)^n 
+ \\
c_3 \left({{1+i\sqrt{3}}\over 2} \right)^n + \overline{c_3} \left({{1-i\sqrt{3}}\over 2} \right)^n
+
c_4 \left({{-1+i\sqrt{3}}\over 2} \right)^n + \overline{c_4} \left({{-1-i\sqrt{3}}\over 2} \right)^n
+
c_5 (-1)^n,
\end{multline*}
where
\begin{align*}
c_1 &= 1/40 \\
c_2 &= -\sqrt{5}/20 \\
c_3 &= (1 - i\sqrt{3})/24 \\
c_4 &= i\sqrt{3}/24 \\
c_5 &= -2/15,
\end{align*}
and here $\overline{x}$ denotes complex conjugate.
Here the parts corresponding to the
constants $c_3, c_4, c_5$ form a periodic sequence of period 6.
\end{theorem}

Next, we turn to what is apparently a new result.
Let 
$h(n)$ denote the total number of distinct factors (not
occurrences of factors) that are abelian squares in the Fibonacci
word $X_n$.

In this case we need the predicate
$$ (k \geq 1) \wedge (i+2k \leq n) \wedge
(\fab[k][i][i+k]=0) \wedge
(\forall j<i \  
(\exists t<2k\ ({\bf f}[j+t]\not= {\bf f}[i+t]))).$$


We get the minimal polynomial
$$ x^4(x+1)(x^2+x+1)(x^2-3x+1)(x^2-x+1)(x^2+x-1)(x^2-x-1)(x-1)^2.$$
Using the same technique as above we get

\begin{theorem}
For $n \geq 2$ we have 
$h(n) = a_1c_1^n + \cdots + a_{10}c_{10}^n   $
where
\begin{align*}
a_1 &= (-2+\sqrt{5})/20 \\
a_2 &= (-2-\sqrt{5})/20 \\
a_3 &= (5-\sqrt{5})/20 \\
a_4 &= (5+\sqrt{5})/20 \\
a_5 &= 1/30 \\
a_6 &= -5/6 \\
a_7 &= (1/12)-i \sqrt{3}/12 \\
a_8 &= (1/12)+i \sqrt{3}/12 \\
a_9 &= (1/6) + i \sqrt{3}/12 \\
a_{10}&= (1/6) - i \sqrt{3}/12 
\end{align*}
and
\begin{align*}
c_1 &= (3+\sqrt{5})/2 \\
c_2 &= (3-\sqrt{5})/2 \\
c_3 &= (1+\sqrt{5})/2 \\
c_4 &= (1-\sqrt{5})/2 \\
c_5 &= -1 \\
c_6 &= 1 \\
c_7 &= (1/2)+i \sqrt{3}/2 \\
c_8 &= (1/2)-i \sqrt{3}/2 \\
c_9 &= (-1/2)+i \sqrt{3}/2 \\
c_{10} &= (-1/2)-i \sqrt{3}/2 .
\end{align*}
\end{theorem}

For another new result,
consider counting the total number $a(n)$ of distinct factors of length $2n$ 
of the infinite word $\bf f$ that are abelian squares.

This function is rather erratic.  The following table gives the
first few values:
\begin{table}[H]
\begin{center}
\begin{tabular}{c|cccccccccccccccccccccccccccccc}
$n$ & 1 & 2 & 3 & 4 & 5 & 6 & 7 & 8 & 9 & 10 & 11 & 12 & 13 & 14 & 15 &
16 & 17 & 18 & 19 & 20 \\
\hline
$a(n)$ & 1&3&5&1&9&5&5&15&3&13&13&5&25&9&15&25&1&27&19&11
\end{tabular}
\end{center}
\end{table}

We use the predicate
$$ (n \geq 1) \wedge (\fab[n][i][i+n]=0) \wedge
(\forall j < i\ ( \exists t<2n \ ({\bf f}[j+t]\not= {\bf f}[i+t]))).$$
to create the matrices and vectors.  

\begin{theorem}
$a(n) = 1$ infinitely often and $a(n) = 2n-1$ infinitely often.
More precisely
$a(n) = 1$ iff $(n)_F = 1$ or $(n)_F = (100)^i 101$ for $i \geq 0$,
and
$a(n) = 2n-1$ iff $(n)_F = 10^i$ for $i \geq 0$.
\end{theorem}

\begin{proof}
For the first statement, we create a DFA accepting those $(n)_F$ for
which $a(n) = 1$, via the predicate
$$ \forall i\ \forall j \ ((\fab[n][i][i+n]=0) \wedge  (\fab[n][j][j+n]=0)) \implies
(\forall t<2n \ ({\bf f}[j+t] = {\bf f}[i+t])).$$
The resulting $6$-state automaton accepts the set specified.


For the second result, we first compute the minimal polynomial of the matrix
$M_0$ of the linear representation.  It is $x^5 (x-1)(x+1)(x^2-x-1)$.
This means that, for $n \geq 5$, we have
$a(F_n) = c_1 + c_2 (-1)^n + c_3 \alpha^n + c_4 \beta^n$ where,
as usual, $\alpha = (1+\sqrt{5})/2$ and $\beta=(1-\sqrt{5})/2$.
Solving for the constants, we determine that 
$a(F_n) = 2F_n - 1$ for $n \geq 2$, as desired.

To show that these are the only cases for which $a(n) = 2n-1$, we use a predicate
that says that there are not at least three different factors of length $2n$ that
are not abelian squares.  Running this through our program results in only the
cases previously discussed.
\end{proof}

Finally, we turn to abelian cubes.  Unlike the case of squares, some
orders do not appear in $\bf f$.

\begin{theorem}
The Fibonacci word $\bf f$ contains, as a factor, an abelian cube
of order $n$ iff $(n)_F$ is accepted by the automaton below.
\end{theorem}

\begin{figure}[H]
\begin{center}
\includegraphics[width=6.5in]{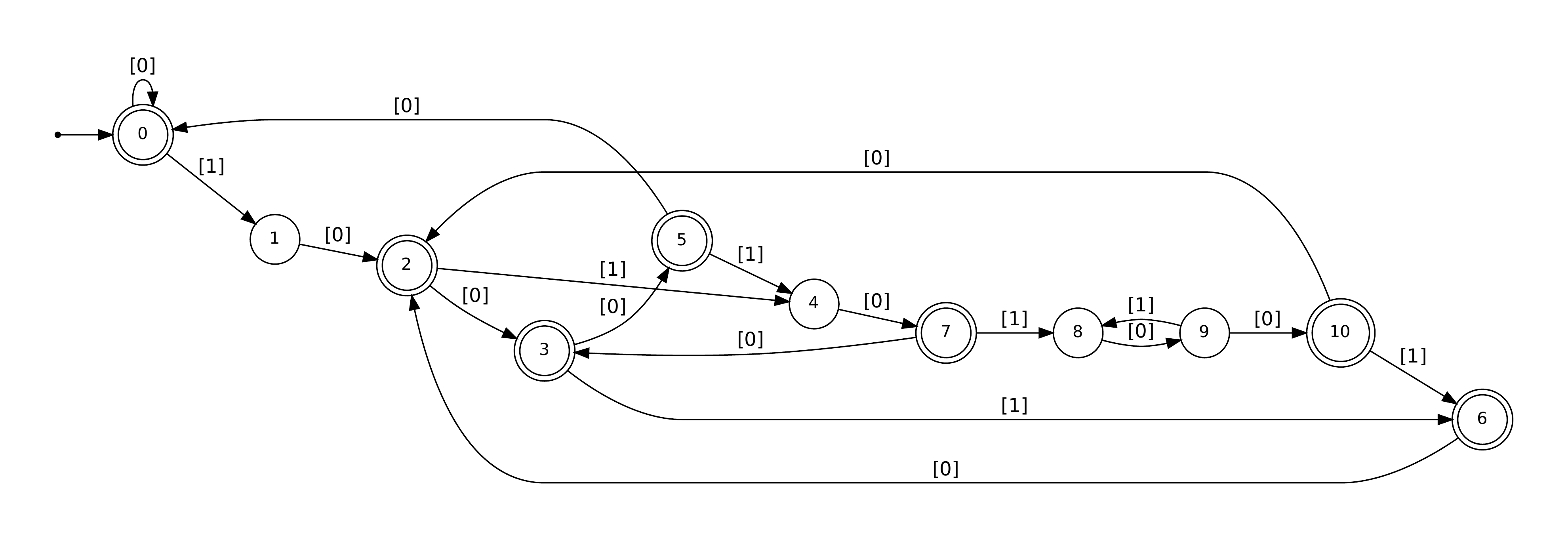}
\caption{Automaton accepting orders of abelian cubes in $\bf f$}
\label{fibabelcube}
\end{center}
\end{figure}

Theorem~\ref{fibr} has the following interesting corollary.

\begin{corollary}
Let $h:\lbrace 0, 1 \rbrace^* \rightarrow \Delta^*$ be an arbitrary
morphism such that $h(01) \not= \epsilon$.  Then $h({\bf f})$ is an
infinite Fibonacci-automatic word.
\end{corollary}

\begin{proof}
From Theorem~\ref{fibr} we see that there is a predicate $\zc(n,n')$
which is true if 
$n' = |{\bf f}[0..n-1]|_0$ and false otherwise, and this predicate
can be implemented as a finite automaton taking the inputs $n$ and $n'$
in Fibonacci representation.

Suppose $h(0) = w$ and $h(1) = x$.   Now, to show that 
h({\bf f}) is Fibonacci-automatic,
it suffices to show that, for each letter $a \in \Delta$,
the language of ``fibers''
$$  L_a = \{ (n)_F : (h({\bf f}))[n] = a  \} $$
is regular.

To see this, we write a predicate for the $n$ in the definition of $L_a$, 
namely
\begin{multline*}
\exists q\ \exists r_0 \ \exists r_1 \ \exists m \ 
(q \leq n < q+ |h({\bf f}[m])|) \ \wedge \  \zc(m,r_0) \ \wedge \  
(r_0+r_1=m) \wedge \\
(r_0 |w| + r_1 |x| = q)  \ \wedge \
(( {\bf f}[m]=0 \ \wedge \  w[n-q] = a) \ \vee \  ({\bf f}[m] = 1 \ \wedge\  x[n-q] = a) ) .
\end{multline*}

Notice that the predicate looks like it uses multiplication, but this
multiplication can be replaced by repeated addition since $|w|$ and $|x|$
are constants here.

Unpacking this predicate we see that it asserts the existence of
$m$, $q$, $r_0$, and $r_1$ having the meaning that
\begin{itemize}
\item the $n$'th symbol of h({\bf f}) lies inside the block
$h({\bf f}[m])$ and is in fact the $(n-q)$'th symbol in the
block (with the first symbol being symbol 0)
\item ${\bf f}[0..m-1]$ has $r_0$ 0's in it
\item $ {\bf f}[0..m-1]$ has $r_1$ 1's in it
\item the length of $h({\bf f}[0..m-1])$ is $q$
\end{itemize}
			     
Since everything in this predicate is in the logical theory
$(\Enn, +, <, F)$ where $F$ is the predicate for the Fibonacci word,
the language $L_a$ is regular.
\end{proof}

\begin{remark}
Notice that everything in this proof goes through for other
numeration systems, provided the original
word has the property that the Parikh vector of the prefix of length $n$
is synchronized.
\end{remark}

\section{Details about our implementation}

Our program is written in JAVA, and was developed using the
{\tt Eclipse} development environment.\footnote{Available from {\tt http://www.eclipse.org/ide/} .} 
We used the {\tt dk.brics.automaton}
package, developed by Anders M{\o}ller at Aarhus University, for
automaton minimization.\footnote{Available from {\tt http://www.brics.dk/automaton/} .}
{\tt Maple 15} was used
to compute characteristic polynomials.\footnote{Available from {\tt http://www.maplesoft.com} .}
The {\tt GraphViz} package was used to display automata.\footnote{Available
from {\tt http://www.graphviz.org} .}

Our program consists of about 2000 lines of code.  We used
Hopcroft's algorithm for DFA minimization.

A user interface is provided to enter queries in a language very
similar to the language of first-order logic. The intermediate and
final result of a query are all automata. At every intermediate step,
we chose to do minimization and determinization, if necessary. Each
automaton accepts tuples of integers in the numeration system of choice.
The built-in numeration systems are ordinary base-$k$ representations
and Fibonacci base. However, the program can be used with any numeration
system for which an automaton for addition and ordering can be
provided. These numeration system-specific automata can be declared in
text files following a simple syntax. For the automaton resulting from a 
query it is always guaranteed that if a tuple $t$ of integers is
accepted, all tuples obtained from $t$ by addition or truncation of
leading zeros are also accepted. In Fibonacci representation, we make sure that
the accepting integers do not contain consecutive $1$'s.

The program was tested against hundreds of different test cases varying
in simplicity from the most basic test cases testing only one feature
at a time, to more comprehensive ones with many alternating quantifiers.
We also used known facts about automatic sequences and Fibonacci
word in the literature to test our program, and in all those cases we
were able to get the same result as in the literature. In a few cases,
we were even able to find small errors in those earlier
results.

The source code and manual will soon be available for free download.


\section{Acknowledgments}

We thank Kalle Saari for bringing our attention to the small
error in \cite{Fraenkel&Simpson:1999}.  We thank Narad Rampersad and
Michel Rigo for useful suggestions.

Eric Rowland thought about the proof of Theorem~\ref{additive-thm} with us
in 2010, and was able to prove at that time
that the word $1213121512131218\cdots$
avoids additive squares.  We acknowledge his prior work on this problem
and thank him for allowing us to quote it here.

\newcommand{\noopsort}[1]{} \newcommand{\singleletter}[1]{#1}

\end{document}